\documentclass[journal,twoside]{IEEEtran}

\usepackage{cite}

\usepackage[cmex10]{amsmath}
\usepackage{array}

\usepackage[pdftex]{graphicx}

\usepackage{amsfonts,amssymb,amsthm,mathtools,bbm,verbatim}
\allowdisplaybreaks 
\usepackage[usenames,dvipsnames]{color}
\usepackage[]{inputenc}
\usepackage{hyperref}
\hypersetup{colorlinks = true, linkcolor = blue, citecolor = blue, urlcolor  = blue}
\usepackage[all]{hypcap} 
\usepackage{balance}

\def\G{{\mathcal{G}}}

\def\X{{\mathcal{X}}}
\def\Y{{\mathcal{Y}}}

\def\F{{\mathcal{F}}}
\def\P{{\mathbb{P}}}

\def\E{{\mathbb{E}}}
\def\VAR{{\mathbb{V}\mathbb{A}\mathbb{R}}}

\def\R{{\mathbb{R}}}
\def\N{{\mathbb{N}}}
\def\1{{\textbf{1}}}
\def\I{{\mathbbm{1}}}

\DeclareMathOperator*{\br}{br}

\DeclarePairedDelimiter{\ceil}{\lceil}{\rceil}
\DeclarePairedDelimiter{\floor}{\lfloor}{\rfloor}
\def\maj{{\mathsf{majority}}}
\def\Ber{{\mathsf{Bernoulli}}}

\def\etaTV{{\eta_{\mathsf{TV}}}}

\makeatletter
\newcommand{\vast}{\bBigg@{3}}
\makeatother

\newtheorem{theorem}{Theorem}
\newtheorem{lemma}{Lemma}
\newtheorem{proposition}{Proposition}
\newtheorem{corollary}{Corollary}

\theoremstyle{definition}

\begin{document}

\title{Broadcasting on Random Directed Acyclic Graphs}

\author{Anuran~Makur,~Elchanan~Mossel,~and~Yury~Polyanskiy%
\thanks{This work was supported in part by the Center for Science of Information, a National Science Foundation (NSF) Science and Technology Center, under Grant CCF-09-39370, in part by the NSF CAREER Award under Grant CCF-12-53205, in part by the Schneider Electric, Lenovo Group (China) Limited and the Hong Kong Innovation and Technology Fund (ITS/066/17FP) under the HKUST-MIT Research Alliance Consortium, in part by the NSF under Grants CCF-1665252 and DMS-1737944, and in part by the DOD ONR under Grant N00014-17-1-2598. This work was presented at the 2019 IEEE International Symposium on Information Theory \cite{MakurMosselPolyanskiy2019Conf}.}%
\thanks{A. Makur and Y. Polyanskiy are with the Department of Electrical Engineering and Computer Science, Massachusetts Institute of Technology, Cambridge, MA 02139, USA (e-mail: a\_makur@mit.edu; yp@mit.edu).}%
\thanks{E. Mossel is with the Department of Mathematics, Massachusetts Institute of Technology, Cambridge, MA 02139, USA (e-mail: elmos@mit.edu).}%
\thanks{Copyright (c) 2019 IEEE. Personal use of this material is permitted. However, permission to use this material for any other purposes must be obtained from the IEEE by sending a request to pubs-permissions@ieee.org.}}%



\maketitle

\begin{abstract}
We study the following generalization of the well-known model of broadcasting on trees. Consider an infinite directed acyclic graph (DAG) with a unique source vertex $X$. Let the collection of vertices at distance $k$ from $X$ be called the $k$th layer, and suppose every non-source vertex has indegree $d \geq 2$. At layer $0$, the source vertex is given a random bit. At layer $k\geq 1$, each vertex receives $d$ bits from its parents in the $(k-1)$th layer, which are transmitted along edges that are independent binary symmetric channels (BSCs) with crossover probability $\delta \in \big(0,\frac{1}{2}\big)$. Each vertex combines its $d$ noisy inputs using a deterministic $d$-ary Boolean processing function that generates the value at the vertex. The goal is to be able to reconstruct the original bit $X$ with probability of error bounded away from $\frac{1}{2}$ using the values of all vertices at an arbitrarily deep layer $k$. This question is closely related to models of reliable computation and storage, and information flow in biological networks.

In this paper, we treat the case of randomly constructed DAGs, for which we show that broadcasting is only possible if the BSC noise level $\delta$ is below a certain (degree and function dependent) critical threshold. For $d \geq 3$, and random DAGs with layers of size $\Omega(\log(k))$ and majority processing functions, we identify the critical threshold. For $d = 2$, we establish a similar result for the NAND processing function. We also prove a partial converse result for odd $d \geq 3$ illustrating that the identified thresholds are impossible to improve by selecting different processing functions if the decoder is restricted to using a single vertex's value.

Finally, for any BSC noise level $\delta$, we construct explicit DAGs (using regular bipartite lossless expander graphs) with bounded degree and layers of size $\Theta(\log(k))$ admitting reconstruction. In particular, we show that the first $r$ layers of such DAGs can be generated in either deterministic quasi-polynomial time or randomized polylogarithmic time in $r$. These results portray a doubly-exponential advantage for storing a bit in bounded degree DAGs compared to trees, where $d = 1$ but layer sizes need to grow exponentially with depth in order for broadcasting to be possible. 
\end{abstract}


\tableofcontents
\hypersetup{linkcolor = red}

\section{Introduction}
\label{Introduction}

In this paper, we study a generalization of the well-known problem of \textit{broadcasting on trees}, cf. \cite{Evansetal2000}, to the setting of bounded indegree directed acyclic graphs (DAGs). In the broadcasting on trees problem, we are given a noisy tree $T$ whose vertices are Bernoulli random variables and edges are independent binary symmetric channels (BSCs) with common crossover probability $\delta \in \big(0,\frac{1}{2}\big)$. Given that the root is an unbiased random bit, the objective is to decode the bit at the root from the bits at the $k$th layer of the tree as $k \rightarrow \infty$. The authors of \cite{Evansetal2000} characterize the sharp threshold for when such reconstruction is possible:
\begin{enumerate}
\item If $(1-2\delta)^2 \br(T) > 1$, then the minimum probability of error in decoding is bounded away from $\frac{1}{2}$ for all $k$, 
\item If $(1-2\delta)^2 \br(T) < 1$, then the minimum probability of error in decoding tends to $\frac{1}{2}$ as $k \rightarrow \infty$,
\end{enumerate}
where $\br(T)$ denotes the \textit{branching number} of the tree (see \cite[Chapter 1.2]{LyonsPeres2017}), and the condition $(1-2\delta)^2 \br(T) \gtrless 1$ is known as the \textit{Kesten-Stigum threshold} in the regular tree setting. This result on reconstruction on trees generalizes results from random processes and statistical physics that hold for regular trees, cf. \cite{KestenStigum1966} (which proves achievability) and \cite{BleherRuizZagrebnov1995} (which proves the converse), and has had numerous extensions and further generalizations including \cite{Ioffe1996a,Ioffe1996b,Mossel1998,Mossel2001,PemantlePeres2010,Sly2009,Sly2011,JansonMossel2004,Bhatnagaretal2011}. (We refer readers to \cite{Mossel2004survey} for a survey of the reconstruction problem on trees.) A consequence of this result is that reconstruction is impossible for trees with sub-exponentially many vertices at each layer. Indeed, if $L_k$ denotes the number of vertices at layer $k$ and $\lim_{k \rightarrow \infty}{L_k^{1/k}} = 1$, then it is straightforward to show that $\br(T) = 1$, which in turn implies that $(1-2\delta)^2 \br(T) < 1$.

Instead of analyzing trees, we consider the problem of broadcasting on \textit{bounded indegree DAGs}. As in the setting of trees, all vertices in our graphs are Bernoulli random variables and all edges are independent BSCs. Furthermore, the values of variables located at vertices with indegree $2$ or more are obtained by applying \textit{Boolean processing functions} to the noisy inputs of the vertices. Hence, compared to the setting of trees, broadcasting on DAGs has two principal differences: 
\begin{enumerate}
\item In trees, layer sizes must scale exponentially in the depth to permit reconstruction, while in DAGs, they may scale as small as logarithmically in the depth.
\item In trees, the indegree of each vertex is $1$, while in DAGs, each vertex has several incoming signals.
\end{enumerate}
The latter enables the possibility of information fusion at the vertices of DAGs, and our main goal is to understand whether the benefits of 2 overpower the shortcoming of 1 and permit reconstruction of the root bit with sub-exponential layer size.

This paper has two main contributions. Firstly, via a probabilistic argument using random DAGs, we demonstrate the existence of bounded indegree DAGs with $L_k = \Omega(\log(k))$ which permit recovery of the root bit for sufficiently low $\delta$'s. Secondly, we provide explicit deterministic constructions of such DAGs using regular bipartite lossless expander graphs. In particular, the constituent expander graphs for the first $r$ layers of such DAGs can be constructed in either deterministic quasi-polynomial time or randomized polylogarithmic time in $r$. Together, these results imply that in terms of economy of storing information, \textit{DAGs are doubly-exponentially more efficient than trees}. 

\subsection{Motivation}
\label{Motivation}

Broadcasting on DAGs has several natural interpretations. Perhaps most pertinently, it captures the feasibility of reliably communicating through Bayesian networks in the field of \textit{communication networks}. Indeed, suppose a sender communicates a sequence of bits to a receiver through a large network. If broadcasting is impossible on this network, then the ``wavefront of information'' for each bit decays irrecoverably through the network, and the receiver cannot reconstruct the sender's message regardless of the coding scheme employed.

The problem of broadcasting on DAGs is also closely related to the problem of \textit{reliable computation using noisy circuits}, cf. \cite{vonNeumann1956,EvansSchulman1999}. The relation between the two models can be understood in the following way. Suppose we want to remember a bit using a noisy circuit of depth $k$. The ``von Neumann approach'' is to take multiple perfect clones of the bit and recursively apply noisy gates in order to reduce the overall noise \cite{HajekWeller1991,EvansSchulman2003}. In contrast, the broadcasting perspective is to start from a single bit and repeatedly create noisy clones and apply perfect gates to these clones so that one can recover the bit reasonably well from the vertices at depth $k$. Thus, the broadcasting model can be construed as a noisy circuit that remembers a bit using perfect logic gates at the vertices and edges or wires that independently make errors. (It is worth mentioning that broadcasting DAG circuits are much smaller, and hence more desirable, than broadcasting tree circuits because bounded degree logic gates can be used to reduce noise.)

Furthermore, special cases of the broadcasting model have found applications in various discrete probability questions. For example, broadcasting on trees corresponds to \textit{ferromagnetic Ising models in statistical physics}. Specifically, if we associate bits $\{0,1\}$ with spins $\{-1,+1\}$, then the joint distribution of the vertices of any finite broadcasting subtree (e.g. up to depth $k$) corresponds to the \textit{Boltzmann-Gibbs distribution} of the configuration of spins in the subtree (where we assume that the strictly positive common interaction strength of the Ising model is fixed and that there is no external magnetic field) \cite[Section 2.2]{Evansetal2000}.\footnote{In particular, each value of $\delta \in \big(0,\frac{1}{2}\big)$ corresponds to a unique value of temperature such that the broadcasting distribution defined by $\delta$ is equivalent to the Boltzmann-Gibbs distribution with the associated temperature parameter \cite[Equation (11)]{Evansetal2000}.} In the theory of Ising models, weak limits of Boltzmann-Gibbs distributions on finite subgraphs of an infinite graph with different boundary conditions yield different \textit{Gibbs measures or states} on the infinite graph, cf. \cite[Chapters 3 and 6]{FriedliVelenik2018}. For instance, the broadcasting distribution on an infinite tree corresponds to the Gibbs measure with \textit{free boundary conditions}, which is obtained by taking a weak limit of the broadcasting distributions over finite subtrees. Moreover, it is well-known that under general conditions, the Dobrushin-Lanford-Ruelle (DLR) formalism for defining Gibbs measures (or DLR states) using \textit{Gibbsian specifications} produces a convex Choquet simplex of Gibbs measures corresponding to any particular specification \cite[Chapters 1, 2, and 7]{Georgii2011} (also see \cite[Chapter 6]{FriedliVelenik2018}). Hence, it is of both mathematical and physical interest to find the extremal Gibbs measures of this simplex.\footnote{Indeed, extremal Gibbs measures are precisely those Gibbs measures that have trivial tail $\sigma$-algebra, i.e. tail events exhibit a zero-one law for extremal Gibbs measures, cf. \cite[Section 7.1]{Georgii2011}, \cite[Section 6.8]{FriedliVelenik2018}. As explained in \cite[Comment (7.8)]{Georgii2011} and \cite[p.302]{FriedliVelenik2018}, since tail events correspond to macroscopic events that are not affected by the behavior of any finite subset of spins, extremal Gibbs measures have deterministic macroscopic events. Thus, from a physical perspective, only extremal Gibbs measures are suitable to characterize the equilibrium states of a statistical mechanical system.} It turns out that reconstruction is impossible on a broadcasting tree if and only if the Gibbs measure with free boundary conditions of the corresponding ferromagnetic Ising model is extremal \cite{BleherRuizZagrebnov1995}, \cite[Section 2.2]{Evansetal2000}. This portrays a strong connection between broadcasting on trees and the theory of Ising models. We refer readers to \cite{Ioffe1996a,Ioffe1996b,PemantlePeres2010,Sly2011} for related work, and to \cite[Section 2.2]{Evansetal2000} for further references on the Ising model literature. 

A second example of a related discrete probability question stems from the theory of \textit{probabilistic cellular automata}. Indeed, another motivation for our problem is to understand whether it is possible to propagate information in regular grids. Notice that the broadcasting model on a two-dimensional regular grid can be perceived as a one-dimensional probabilistic cellular automaton (see e.g. \cite[Section 1]{Gray2001} for a definition) with boundary conditions that limit the layer sizes. So, the feasibility of broadcasting on two-dimensional regular grids intuitively corresponds to a question about the ergodicity of one-dimensional probabilistic cellular automata (along with sufficiently fast convergence rate). Inspired by work on the \textit{positive rates conjecture} for one-dimensional probabilistic cellular automata, cf. \cite[Section 1]{Gray2001}, our conjecture is that propagation of information is impossible for a two-dimensional regular grid regardless of the noise level and of the choice of Boolean processing function (which is the same for every vertex). Our results towards establishing this conjecture will be the focus of a forthcoming paper.

Finally, the broadcasting process on trees also plays a fundamental role in understanding various questions in theoretical computer science and learning theory. For example, results on trees have been exploited in problems of \textit{ancestral data and phylogenetic tree reconstruction}\textemdash see e.g. \cite{Mossel2003,Mossel2004,DaskalakisMosselRoch2006,Roch2010}. In fact, the existence results obtained in this paper suggest that it might be possible to reconstruct other biological networks, such as phylogenetic networks (see e.g. \cite{HusonRuppScornavacca2010}) or pedigrees (see e.g. \cite{Thompson1986,SteelHein2006}), even if the growth of the network is very mild. Moreover, reconstruction on trees can be used to understand phase transitions for \textit{random constraint satisfaction} problems\textemdash see e.g. \cite{MezardMontanari2006,Krzakalaetal2007,GerschenfeldMontanari2007,MontanariRestrepoTetali2011} and follow-up work. It is an interesting future endeavor to explore if there are connections between broadcasting on general DAGs and random constraint satisfaction problems. Currently, we are not aware that such connections have been established. Lastly, we note that broadcasting on trees has also been used to prove impossibility of weak recovery (or detection) in the problem of \textit{community detection in stochastic block models}, cf. \cite[Section 5.1]{Abbe2018}.

\subsection{Outline}

We briefly outline the rest of this paper. Since we will use probabilistic arguments to establish the existence of bounded indegree DAGs where reconstruction of root bit is possible, we will prove many of our results for random DAGs. So, the next subsection \ref{Random Grid Model} formally defines the random DAG model. In section \ref{Main Results}, we present our three main results (as well as some auxiliary results) pertaining to the random DAG model, and discuss several related results in the literature. Then, we prove these main results in sections \ref{Analysis of Majority Rule Processing in Random Grid}, \ref{Analysis of And-Or Rule Processing in Random Grid}, and \ref{Deterministic Quasi-Polynomial Time and Randomized Polylogarithmic Time Constructions of DAGs where Broadcasting is Possible}, respectively. In particular, section \ref{Analysis of Majority Rule Processing in Random Grid} analyzes broadcasting with majority processing functions when the indegree of each vertex is $3$ or more, section \ref{Analysis of And-Or Rule Processing in Random Grid} analyzes broadcasting with AND and OR processing functions when the indegree of each vertex is $2$, and section \ref{Deterministic Quasi-Polynomial Time and Randomized Polylogarithmic Time Constructions of DAGs where Broadcasting is Possible} illustrates our explicit constructions of DAGs where reconstruction of the root bit is possible using expander graphs. Finally, we conclude our discussion and list some open problems in section \ref{Conclusion}.

\subsection{Random DAG Model}
\label{Random Grid Model}

A \textit{random DAG model} consists of an infinite DAG with fixed vertices that are Bernoulli ($\{0,1\}$-valued) random variables and randomly generated edges which are independent BSCs. We first define the vertex structure of this model, where each vertex is identified with the corresponding random variable. Let the root (or ``source'') random variable be $X_{0,0} \sim \Ber\big(\frac{1}{2}\big)$. Furthermore, we define $X_k = (X_{k,0},\dots,X_{k,L_k-1})$ as the vector of vertex random variables at distance (i.e. length of shortest path) $k \in \N \triangleq \{0,1,2,\dots\}$ from the root, where $L_k \in \N\backslash\!\{0\}$ denotes the number of vertices at distance $k$. In particular, we have $X_0 = (X_{0,0})$ so that $L_0 = 1$, and we are typically interested in the regime where $L_k \rightarrow \infty$ as $k \rightarrow \infty$.

We next define the edge structure of the random DAG model. For any $k \in \N\backslash\!\{0\}$ and any $j \in [L_{k}] \triangleq \{0,\dots,L_{k} - 1\}$, we independently and uniformly select $d \in \N \backslash \! \{0\}$ vertices $X_{k-1,i_1},\dots,X_{k-1,i_d}$ with replacement from $X_{k-1}$ (i.e. $i_1,\dots,i_d$ are i.i.d. uniform on $[L_{k-1}]$), and then construct $d$ directed edges: $(X_{k-1,i_1},X_{k,j}),\dots,$ $(X_{k-1,i_d},X_{k,j})$. (Here, $i_1,\dots,i_d$ are independently chosen for each $X_{k,j}$.) This random process generates the underlying DAG structure. In the sequel, we will let $G$ be a random variable representing this underlying (infinite) random DAG, i.e. $G$ encodes the random configuration of the edges between the vertices.

To define a \textit{Bayesian network} (or directed graphical model) on this random DAG, we fix some sequence of Boolean functions $f_{k}:\{0,1\}^d \rightarrow \{0,1\}$ for $k \in \N\backslash\!\{0\}$ (that depend on the level index $k$, but typically not on the realization of $G$), and some crossover probability $\delta \in \big(0,\frac{1}{2}\big)$ (since this is the interesting regime of $\delta$). Then, for any $k \in \N\backslash\!\{0\}$ and $j \in [L_{k}]$, given $i_1,\dots,i_d$ and $X_{k-1,i_1},\dots,X_{k-1,i_d}$, we define:
\begin{equation}
\label{eq:propagation}
X_{k,j} = f_{k}(X_{k-1,i_1} \oplus Z_{k,j,1},\dots,X_{k-1,i_d}\oplus Z_{k,j,d})
\end{equation}
where $\oplus$ denotes addition modulo $2$, and $\{Z_{k,j,i} : k \in \N\backslash\!\{0\}, \, j \in [L_{k}], \, i \in \{1,\dots,d\}\}$ are i.i.d. $\Ber(\delta)$ random variables that are independent of everything else. This means that each edge is a BSC with parameter $\delta$, denoted $\mathsf{BSC}(\delta)$. Moreover, \eqref{eq:propagation} characterizes the conditional distribution of $X_{k,j}$ given its parents. In this model, the Boolean processing function used at a vertex $X_{k,j}$ depends only on the level index $k$. A more general model can be defined where each vertex $X_{k,j}$ has its own Boolean processing function $f_{k,j}:\{0,1\}^d \rightarrow \{0,1\}$ for $k \in \N\backslash\!\{0\}$ and $j \in [L_{k}]$. However, with the exception of a few converse results, we will mainly analyze instances of the simpler model in this paper.

Note that although we will analyze this model for convenience, as stated, our underlying graph is really a directed multigraph rather than a DAG, because we select the parents of a vertex with replacement. It is straightforward to construct an equivalent model where the underlying graph is truly a DAG. For each vertex $X_{k,j}$ with $k \in \N\backslash\!\{0\}$ and $j \in [L_{k}]$, we first construct $d$ intermediate parent vertices $\{X_{k,j}^i : i \in \{1,\dots,d\}\}$ that live between layers $k$ and $k-1$, where each $X_{k,j}^i$ has a single edge pointing to $X_{k,j}$. Then, for each $X_{k,j}^i$, we independently and uniformly select a vertex from layer $k-1$, and construct a directed edge from that vertex to $X_{k,j}^i$. This defines a valid (random) DAG. As a result, every realization of $G$ can be perceived as either a directed multigraph or its equivalent DAG. Furthermore, the Bayesian network on this true DAG is defined as follows: each $X_{k,j}$ is the output of a Boolean processing function $f_{k}$ with inputs $\{X_{k,j}^i : i \in \{1,\dots,d\}\}$, and each $X_{k,j}^i$ is the output of a BSC whose input is the unique parent of $X_{k,j}^i$ in layer $k-1$. 

Finally, we define the ``empirical probability of unity'' at level $k \in \N$ as:
\begin{equation}
\label{Eq: Empirical Probability of Unity Definition}
\sigma_k \triangleq \frac{1}{L_k} \sum_{m = 0}^{L_k - 1}{X_{k,m}}
\end{equation}
where $\sigma_0 = X_{0,0}$ is just the root vertex. Observe that given $\sigma_{k-1} = \sigma$, $X_{k-1,i_1},\dots,X_{k-1,i_d}$ are i.i.d. $\Ber(\sigma)$, and as a result, $X_{k-1,i_1} \oplus Z_{k,j,1},\dots,X_{k-1,i_d} \oplus Z_{k,j,d}$ are i.i.d. $\Ber(\sigma * \delta)$, where $\sigma * \delta \triangleq \sigma(1-\delta) + \delta(1-\sigma)$ is the convolution of $\sigma$ and $\delta$. Therefore, $X_{k,j}$ is the output of $f_{k}$ upon inputting a $d$-length i.i.d. $\Ber(\sigma * \delta)$ string.

Under this formal setup, our objective is to determine whether or not the value at the root $\sigma_0 = X_{0,0}$ can be decoded from the observations $X_k$ as $k \rightarrow \infty$. Since $X_k$ is an exchangeable sequence of random variables given $\sigma_0$, for any values $x_{0,0},x_{k,0},\dots,x_{k,L_k-1} \in \{0,1\}$ and any permutation $\pi$ of $[L_k]$, we have:
\begin{equation}
\begin{aligned}
& P_{X_k|\sigma_0}(x_{k,0},\dots,x_{k,L_k-1}|x_{0,0}) \\
&  \quad \quad \quad \quad = P_{X_k|\sigma_0}(x_{k,\pi(0)},\dots,x_{k,\pi(L_k-1)}|x_{0,0}) \, .
\end{aligned}
\end{equation}
Letting $\sigma = \frac{1}{L_k}\sum_{j = 0}^{L_k - 1}{x_{k,j}}$, we can factorize $P_{X_k|\sigma_0}$ as:
\begin{equation}
\label{Eq: Exchangeability}
P_{X_k|\sigma_0}(x_{k,0},\dots,x_{k,L_k-1}|x_{0,0}) = {\binom{L_k}{L_k \sigma}}^{\! -1} P_{\sigma_k|\sigma_0}(\sigma|x_{0,0}) \, . 
\end{equation}
Using the Fisher-Neyman factorization theorem \cite[Theorem 3.6]{Keener2010}, this implies that $\sigma_k$ is a \textit{sufficient statistic} of $X_k$ for performing inference about $\sigma_0$. Therefore, we restrict our attention to the Markov chain $\{\sigma_k : k \in \N\}$ in our achievability proofs, since if decoding is possible from $\sigma_k$, then it is also possible from $X_k$. Given $\sigma_k$, inferring the value of $\sigma_0$ is a binary hypothesis testing problem with minimum achievable probability of error:
\begin{equation}
\P\!\left(f_{\mathsf{ML}}^k(\sigma_k) \neq \sigma_0\right) = \frac{1}{2}\left(1-\left\|P_{\sigma_k}^+ - P_{\sigma_k}^-\right\|_{\mathsf{TV}}\right)
\end{equation}
where $f_{\mathsf{ML}}^k:\{m/L_k : m \in \{0,\dots,L_k\}\} \rightarrow \{0,1\}$ is the maximum likelihood (ML) decision rule based on the empirical probability of unity at level $k$ in the absence of knowledge of the random DAG realization $G$, $P_{\sigma_k}^+$ and $P_{\sigma_k}^-$ are the conditional distributions of $\sigma_k$ given $\sigma_0 = 1$ and $\sigma_0 = 0$ respectively, and for any two probability measures $P$ and $Q$ on the same measurable space $(\Omega,\F)$, their \textit{total variation (TV) distance} is defined as:
\begin{align}
\left\|P - Q\right\|_{\mathsf{TV}} & \triangleq \sup_{A \in \F}{\left|P(A) - Q(A)\right|} \\
& = \frac{1}{2}\left\|P - Q\right\|_{1}
\label{Eq: TV Distance Definition}
\end{align}
where $\left\|\cdot\right\|_{1}$ denotes the $\mathcal{L}^1$-norm. We say that reconstruction of the root bit $\sigma_0$ is possible when:\footnote{The limits in \eqref{Eq:TV Reconstruction Possible}, \eqref{Eq:Strong TV Reconstruction Impossible}, and \eqref{Eq:TV Reconstruction Impossible} exist because $\P\big(f_{\mathsf{ML}}^k(\sigma_k) \neq \sigma_0\big)$, $\P\big(h_{\mathsf{ML}}^k(X_k) \neq X_0\big)$, and $\P\big(h_{\mathsf{ML}}^k(X_k,G) \neq X_0 \big| G\big)$ (for any fixed realization $G$) are monotone non-decreasing sequences in $k$ that are bounded above by $\frac{1}{2}$. This can be deduced either from the data processing inequality for TV distance, or from the fact that a randomized decoder at level $k$ can simulate the stochastic transition to level $k+1$.}
\begin{equation}
\label{Eq:TV Reconstruction Possible}
\begin{aligned}
& \lim_{k \rightarrow \infty}{\P\!\left(f_{\mathsf{ML}}^k(\sigma_k) \neq \sigma_0\right)} < \frac{1}{2} \\
& \quad \quad \quad \quad \quad \quad \Leftrightarrow \quad \lim_{k \rightarrow \infty}{\left\|P_{\sigma_k}^+ - P_{\sigma_k}^-\right\|_{\mathsf{TV}}} > 0 \, . 
\end{aligned}
\end{equation}
In the sequel, to simplify our analysis when proving that reconstruction is possible, we will sometimes use other (sub-optimal) decision rules rather than the ML decision rule.

On the other hand, we will consider the Markov chain $\{X_k : k \in \N\}$ conditioned on $G$ in our converse proofs. We say that reconstruction of the root bit $X_{0}$ is impossible when:
\begin{equation}
\label{Eq:Strong TV Reconstruction Impossible}
\begin{aligned}
& \lim_{k \rightarrow \infty}{\P\!\left(h_{\mathsf{ML}}^k(X_k,G) \neq X_0 \middle| G\right)} = \frac{1}{2} \quad G\text{-}a.s. \\
& \quad \quad \quad \Leftrightarrow \quad \lim_{k \rightarrow \infty}{\left\|P_{X_k|G}^+ - P_{X_k|G}^-\right\|_{\mathsf{TV}}} = 0 \quad G\text{-}a.s.
\end{aligned}
\end{equation}
where $h_{\mathsf{ML}}^k(\cdot,G):\{0,1\}^{L_k} \rightarrow \{0,1\}$ is the ML decision rule based on the full state at level $k$ given knowledge of the random DAG realization $G$, $P_{X_k|G}^+$ and $P_{X_k|G}^-$ denote the conditional distributions of $X_k$ given $\{X_0 = 1,G\}$ and $\{X_0 = 0,G\}$ respectively, and the notation $G$-$a.s.$ (almost surely) implies that the conditions in \eqref{Eq:Strong TV Reconstruction Impossible} hold with probability $1$ with respect to the distribution of the random DAG $G$. Note that applying the bounded convergence theorem to the TV distance condition in \eqref{Eq:Strong TV Reconstruction Impossible} yields $\lim_{k \rightarrow \infty}{\E\big[\big\|P_{X_k|G}^+ - P_{X_k|G}^-\big\|_{\mathsf{TV}}\big]} = 0$, and employing Jensen's inequality here establishes the weaker impossibility result:
\begin{equation}
\label{Eq:TV Reconstruction Impossible}
\begin{aligned}
& \lim_{k \rightarrow \infty}{\P\!\left(h_{\mathsf{ML}}^k(X_k) \neq X_0\right)} = \frac{1}{2} \\
& \quad \quad \quad \quad \quad \Leftrightarrow \quad \lim_{k \rightarrow \infty}{\left\|P_{X_k}^+ - P_{X_k}^-\right\|_{\mathsf{TV}}} = 0
\end{aligned} 
\end{equation}
where $h_{\mathsf{ML}}^k:\{0,1\}^{L_k} \rightarrow \{0,1\}$ is the ML decision rule based on the full state at level $k$ in the absence of knowledge of the random DAG realization $G$, and $P_{X_k}^+$ and $P_{X_k}^-$ are the conditional distributions of $X_k$ given $X_0 = 1$ and $X_0 = 0$, respectively. Since $\sigma_k$ is a sufficient statistic of $X_k$ for performing inference about $\sigma_0$ when we average over $G$, we have $\P\big(h_{\mathsf{ML}}^k(X_k) \neq X_0\big) = \P\big(f_{\mathsf{ML}}^k(\sigma_k) \neq \sigma_0\big)$ (or equivalently, $\big\|P_{X_k}^+ - P_{X_k}^-\big\|_{\mathsf{TV}} = \big\|P_{\sigma_k}^+ - P_{\sigma_k}^-\big\|_{\mathsf{TV}}$), and the condition in \eqref{Eq:TV Reconstruction Impossible} is a counterpart of \eqref{Eq:TV Reconstruction Possible}.

\section{Main Results and Discussion}
\label{Main Results}

In this section, we state our main results, briefly delineate the main techniques or intuition used in the proofs, and discuss related literature.

\subsection{Results on Random DAG Models}
\label{Results on Random DAG Models}

We prove two main results on the random DAG model. The first considers the setting where the indegree of each vertex (except the root) is $d \geq 3$. In this scenario, taking a majority vote of the inputs at each vertex intuitively appears to have good ``local error correction'' properties. So, we fix all Boolean functions in the random DAG model to be the ($d$-input) \textit{majority} rule, and prove that this model exhibits a phase transition phenomenon around a critical threshold of:
\begin{equation}
\label{Eq:Critical Noise Level}
\delta_{\mathsf{maj}} \triangleq \frac{1}{2} - \frac{\displaystyle{2^{d-2}}}{\displaystyle{\ceil[\bigg]{\frac{d}{2}} \binom{d}{\ceil[\big]{\frac{d}{2}}}}} 
\end{equation}
where $\ceil{\cdot}$ denotes the ceiling function, and we will use $\floor{\cdot}$ to denote the floor function. Indeed, the theorem below illustrates that for $\delta < \delta_{\mathsf{maj}}$, the majority decision rule $\hat{S}_k \triangleq \I\big\{\sigma_k \geq \frac{1}{2}\big\}$ can asymptotically decode $\sigma_0$, but for $\delta > \delta_{\mathsf{maj}}$, the ML decision rule with knowledge of $G$ cannot asymptotically decode $\sigma_0$. 

\begin{theorem}[Phase Transition in Random DAG Model with Majority Rule Processing]
\label{Thm:Phase Transition in Random Grid with Majority Rule Processing}
Let $C(\delta,d)$ and $D(\delta,d)$ be the constants defined in \eqref{Eq: Large Constant before Log} and \eqref{Eq:Lipschitz constant of g} in section \ref{Analysis of Majority Rule Processing in Random Grid}. For a random DAG model with $d \geq 3$ and majority processing functions (where ties are broken by outputting random bits), the following phase transition phenomenon occurs around $\delta_{\mathsf{maj}}$:
\begin{enumerate}
\item If $\delta \in (0,\delta_{\mathsf{maj}})$, and the number of vertices per level satisfies $L_k \geq C(\delta,d)\log(k)$ for all sufficiently large $k$ (depending on $\delta$ and $d$), then reconstruction is possible in the sense that:\footnote{Throughout this paper, $\log(\cdot)$ and $\exp(\cdot)$ denote the natural logarithm and natural exponential (with base $e$), respectively.} 
$$ \limsup_{k \rightarrow \infty}{\P(\hat{S}_{k} \neq \sigma_0)} < \frac{1}{2} $$
where we use the majority decoder $\hat{S}_k = \I\big\{\sigma_k \geq \frac{1}{2}\big\}$ at level $k$.
\item If $\delta \in \big(\delta_{\mathsf{maj}},\frac{1}{2}\big)$, and the number of vertices per level satisfies $L_k = o\big(D(\delta,d)^{- k}\big)$, then reconstruction is impossible in the sense of \eqref{Eq:Strong TV Reconstruction Impossible}:
$$ \lim_{k \rightarrow \infty}{\left\|P_{X_k|G}^+ - P_{X_k|G}^-\right\|_{\mathsf{TV}}} = 0 \quad G\text{-}a.s. $$
\end{enumerate} 
\end{theorem}

Theorem \ref{Thm:Phase Transition in Random Grid with Majority Rule Processing} is proved in section \ref{Analysis of Majority Rule Processing in Random Grid}. Intuitively, the proof considers the conditional expectation function, $g:[0,1] \rightarrow [0,1]$, $g(\sigma) = \E[\sigma_k|\sigma_{k-1} = \sigma]$ (see \eqref{Eq:Binomial form of g} and \eqref{Eq:Conditional Expectation} in section \ref{Analysis of Majority Rule Processing in Random Grid}), which provides the approximate value of $\sigma_k$ given the value of $\sigma_{k-1}$ for large $k$. This function turns out to have three fixed points when $\delta \in (0,\delta_{\mathsf{maj}})$, and only one fixed point when $\delta \in \big(\delta_{\mathsf{maj}},\frac{1}{2}\big)$. In the former case, $\sigma_k$ ``moves'' to the largest fixed point when $\sigma_0 = 1$, and to the smallest fixed point when $\sigma_0 = 0$. In the latter case, $\sigma_k$ ``moves'' to the unique fixed point of $\frac{1}{2}$ regardless of the value of $\sigma_0$ (see Proposition \ref{Prop: Majority Grid Almost Sure Convergence} in section \ref{Analysis of Majority Rule Processing in Random Grid}).\footnote{Note, however, that $\sigma_k \rightarrow \frac{1}{2}$ almost surely as $k \rightarrow \infty$ does not imply the impossibility of reconstruction in the sense of \eqref{Eq:TV Reconstruction Impossible}, let alone \eqref{Eq:Strong TV Reconstruction Impossible}. So, a different argument is required to establish such impossibility results.} This provides the guiding intuition for why we can asymptotically decode $\sigma_0$ when $\delta \in (0,\delta_{\mathsf{maj}})$, but not when $\delta \in \big(\delta_{\mathsf{maj}},\frac{1}{2}\big)$.

The recursive (or fixed point) structure of $g$ in the special case where $d = 3$ and $\delta_{\mathsf{maj}} = \frac{1}{6}$ can be traced back to the work of von Neumann in \cite{vonNeumann1956}. So, it is worth comparing Theorem \ref{Thm:Phase Transition in Random Grid with Majority Rule Processing} with von Neumann's results in \cite[Section 8]{vonNeumann1956}, where the threshold of $\frac{1}{6}$ is also significant. In \cite[Section 8]{vonNeumann1956}, von Neumann demonstrates the possibility of reliable computation by constructing a circuit with successive layers of computation and local error correction using $3$-input $\delta$-noisy majority gates (i.e. the gates independently make errors with probability $\delta$). In his analysis, he first derives a simple recursion that captures the effect on the probability of error after applying a single noisy majority gate. Then, he uses a ``heuristic'' fixed point argument to show that as the depth of the circuit grows, the probability of error asymptotically stabilizes at a fixed point value less than $\frac{1}{2}$ if $\delta < \frac{1}{6}$, and the probability of error tends to $\frac{1}{2}$ if $\delta \geq \frac{1}{6}$. Moreover, he rigorously proves that reliable computation is possible for $\delta < 0.0073$. 

As we mentioned in subsection \ref{Motivation}, von Neumann's approach to remembering a random initial bit entails using multiple clones of the initial bit as inputs to a noisy circuit with one output, where the output equals the initial bit with probability greater than $\frac{1}{2}$ for ``good'' choices of noisy gates. It is observed in \cite[Section 2]{HajekWeller1991} that a balanced ternary tree circuit, with $k$ layers of $3$-input noisy majority gates and $3^k$ inputs that are all equal to the initial bit, can be used to remember the initial bit. In fact, von Neumann's heuristic fixed point argument that yields a critical threshold of $\frac{1}{6}$ for reconstruction is rigorous in this scenario. From this starting point, Hajek and Weller prove the stronger impossibility result that reliable computation is impossible for formulae (i.e. circuits where the output of each intermediate gate is the input of only one other gate) with general $3$-input $\delta$-noisy gates when $\delta \geq \frac{1}{6}$ \cite[Proposition 2]{HajekWeller1991}. This development can be generalized for any odd $d \geq 3$, and \cite[Theorem 1]{EvansSchulman2003} conveys that reliable computation is impossible for formulae with general $d$-input $\delta$-noisy gates when $\delta \geq \delta_{\mathsf{maj}}$. 

The discussion heretofore reveals that the critical thresholds in von Neumann's circuit for remembering a bit and in our model in Theorem \ref{Thm:Phase Transition in Random Grid with Majority Rule Processing} are both $\delta_{\mathsf{maj}}$. It turns out that this is a consequence of the common fixed point iteration structure of the two problems (as we will explain below). Indeed, the general recursive structure of $g$ for any odd value of $d$ was analyzed in \cite[Section 2]{EvansSchulman2003}. On a related front, the general recursive structure of $g$ was also analyzed in \cite{Mossel1998} in the context of performing recursive reconstruction on periodic trees, where the critical threshold of $\delta_{\mathsf{maj}}$ again plays a crucial role. In fact, we will follow the analysis in \cite{Mossel1998} to develop these recursions in section \ref{Analysis of Majority Rule Processing in Random Grid}.

We now elucidate the common fixed point iteration structure between von Neumann's model and our model in Theorem \ref{Thm:Phase Transition in Random Grid with Majority Rule Processing}. Suppose $d \geq 3$ is odd, and define the function $h:[0,1] \rightarrow [0,1]$, $h(p) \triangleq \P(\maj(Y_1,\dots,Y_d) = 1)$ for $Y_1,\dots,Y_d$ i.i.d. $\Ber(p)$. Consider von Neumann's balanced $d$-ary tree circuit with $k$ layers of $d$-input $\delta$-noisy majority gates and $d^k$ inputs that are all equal to the initial bit. In this model, it is straightforward to verify that the probability of error (i.e. output vertex $\neq$ initial bit) is $f^{(k)}(0)$, where $f:[0,1] \rightarrow [0,1]$ is given by \cite[Equation (3)]{EvansSchulman2003}:
\begin{equation}
\label{Eq: General Von Neumann Recursion}
f(\sigma) \triangleq \delta * h(\sigma) \, ,
\end{equation}
and $f^{(k)}$ denotes the $k$-fold composition of $f$ with itself. On the other hand, as explained in the brief intuition for our proof of Theorem \ref{Thm:Phase Transition in Random Grid with Majority Rule Processing} earlier, assuming that $\sigma_0 = 0$, the relevant recursion for our model is given by the repeated composition $g^{(k)}(0)$ (which captures the average value of $\sigma_k$ after $k$ layers). According to \eqref{Eq:Binomial form of g} in section \ref{Analysis of Majority Rule Processing in Random Grid}, $g(\sigma) = h(\delta * \sigma)$, which yields the relation:
\begin{equation}
\label{Eq: Equivalence of Recursions}
\forall k \in \N\backslash\!\{0\}, \enspace f^{(k+1)}(0) = \delta * g^{(k)}(0) 
\end{equation}
by induction. Therefore, the fixed point iteration structures of $f$ and $g$ are identical, and $\delta_{\mathsf{maj}}$ is the common critical threshold that determines when there is a unique fixed point. In particular, the fact that gates (or vertices) are noisy in von Neumann's model, while edges (or wires) are noisy in our model, has no bearing on this fixed point structure.\footnote{We refer readers to \cite{DobrushinOrtyukov1977} for general results on the relation between vertex noise and edge noise.}

Although both the aforementioned models use majority gates and share a common fixed point structure, it is important to recognize that our overall analysis differs from von Neumann's analysis in a crucial way. Since our recursion pertains to conditional expectations of the proportion of $1$'s in different layers (rather than the probabilities of error in von Neumann's setting), our proof requires exponential concentration inequalities to formalize the intuition provided by the fixed point analysis.

We now make several other pertinent remarks about Theorem \ref{Thm:Phase Transition in Random Grid with Majority Rule Processing}. Firstly, reconstruction is possible in the sense of \eqref{Eq:TV Reconstruction Possible} when $\delta \in (0,\delta_{\mathsf{maj}})$ since the ML decision rule achieves lower probability of error than the majority decision rule,\footnote{It can be seen from monotonicity and symmetry considerations that without knowledge of the random DAG realization $G$, the ML decision rule $f^k_{\mathsf{ML}}(\sigma_k)$ is equal to the majority decision rule $\hat{S}_k$. (So, the superior limit in part 1 of Theorem \ref{Thm:Phase Transition in Random Grid with Majority Rule Processing} can be replaced by a true limit.) On the other hand, with knowledge of the random DAG realization $G$, the ML decision rule $f^k_{\mathsf{ML}}(\sigma_k,G)$ based on $\sigma_k$ is not the majority decision rule.} and reconstruction is impossible in the sense of \eqref{Eq:TV Reconstruction Impossible} when $\delta \in \big(\delta_{\mathsf{maj}},\frac{1}{2}\big)$ (as explained at the end of subsection \ref{Random Grid Model}). Furthermore, while part 1 of Theorem \ref{Thm:Phase Transition in Random Grid with Majority Rule Processing} only shows that the ML decoder $f^k_{\mathsf{ML}}(\sigma_k)$ based on $\sigma_k$ is optimal in the absence of knowledge of the particular graph realization $G$, part 2 establishes that even if the ML decoder knows the graph $G$ and has access to the full $k$-layer state $X_k$, it cannot beat the $\delta_{\mathsf{maj}}$ threshold in all but a zero measure set of DAGs. 

Secondly, the following conjecture is still open: In the random DAG model with $L_k = O(\log(k))$ and fixed $d \geq 3$, reconstruction is impossible for all choices of Boolean processing functions when $\delta \geq \delta_{\mathsf{maj}}$. A consequence of this conjecture is that majority processing functions are optimal, i.e. they achieve the $\delta_{\mathsf{maj}}$ reconstruction threshold. The results in \cite{Mossel1998} provide strong evidence that this conjecture is true when all vertices in the random DAG use the same odd Boolean processing function. Indeed, for fixed $\delta \in \big(0,\frac{1}{2}\big)$ and any odd Boolean function $\mathsf{gate}:\{0,1\}^d \rightarrow \{0,1\}$, let $\tilde{g}:[0,1] \rightarrow [0,1]$ be defined as $\tilde{g}(\sigma) \triangleq \P(\mathsf{gate}(Y_1,\dots,Y_d) = 1)$ for $Y_1,\dots,Y_d$ i.i.d. $\Ber(\delta * \sigma)$.\footnote{A Boolean function is said to be \textit{odd} if flipping all its input bits also flips the output bit. The assumption that $\mathsf{gate}$ is odd ensures that the function $R_{\mathsf{gate}}^{\delta}(\sigma)$ in \cite[Definition 2.1, Lemma 2.3]{Mossel1998} (for a $1$-level $d$-regular tree) is precisely equal to the function $\tilde{g}(\sigma)$.} Then, \cite[Lemma 2.4]{Mossel1998} establishes that $\tilde{g}(\sigma) \leq g(\sigma)$ for all $\sigma \geq \frac{1}{2}$ and $\tilde{g}(\sigma) \geq g(\sigma)$ for all $\sigma \leq \frac{1}{2}$, where the function $g$ is given in \eqref{Eq:Binomial form of g} (and corresponds to the majority rule). Hence, if $g$ has a single fixed point at $\sigma = \frac{1}{2}$, $\tilde{g}$ also has a single fixed point at $\sigma = \frac{1}{2}$. This intuitively suggests that if reconstruction of the root bit is impossible using majority processing functions, it is also impossible using any odd processing function. Furthermore, our proof of part 2 of Theorem \ref{Thm:Phase Transition in Random Grid with Majority Rule Processing} in section \ref{Analysis of Majority Rule Processing in Random Grid} yields that reconstruction is impossible for all choices of odd and monotone non-decreasing Boolean processing functions when $\delta > \delta_{\mathsf{maj}}$, modulo the following conjecture (which we did not verify): among all odd and monotone non-decreasing Boolean functions, the maximum Lipschitz constant of $\tilde{g}$ is attained by the majority rule at $\sigma = \frac{1}{2}$.

Thirdly, the sub-exponential layer size condition $L_k = o\big(D(\delta,d)^{-k}\big)$ in part 2 of Theorem \ref{Thm:Phase Transition in Random Grid with Majority Rule Processing} is intuitively necessary. Suppose every Boolean processing function in our random DAG model simply outputs the value of its first input bit. This effectively sets $d = 1$, and reduces our problem to one of broadcasting on a random tree model. If $L_k = \Omega(E(\delta)^k)$ for some large enough constant $E(\delta)$, then most realizations of the random tree will have branching numbers greater than $(1-2\delta)^{-2}$. As a result, reconstruction will be possible for most realizations of the random tree (cf. the Kesten-Stigum threshold delineated at the outset of section \ref{Introduction}). Thus, when we are proving impossibility results, $L_k$ (at least intuitively) cannot be exponential in $k$ with a very large base. 

Fourthly, it is worth mentioning that for any fixed DAG with indegree $d \geq 3$ and sub-exponential $L_k$, for any choices of Boolean processing functions, and any choice of decoder, it is impossible to reconstruct the root bit when $\delta > \frac{1}{2} - \frac{1}{2\sqrt{d}}$. This follows from Evans and Schulman's result in \cite{EvansSchulman1999}, which we will discuss further in subsection \ref{Further Discussion}. 

Lastly, in the context of the random DAG model studied in Theorem \ref{Thm:Phase Transition in Random Grid with Majority Rule Processing}, the ensuing proposition illustrates that the problem of reconstruction using the information contained in just a single vertex, e.g. $X_{k,0}$, exhibits a similar phase transition phenomenon to that in Theorem \ref{Thm:Phase Transition in Random Grid with Majority Rule Processing}.

\begin{proposition}[Single Vertex Reconstruction]
\label{Prop: Single Vertex Reconstruction}
Let $C(\delta,d)$ be the constant defined in \eqref{Eq: Large Constant before Log} in section \ref{Analysis of Majority Rule Processing in Random Grid}. For a random DAG model with $d \geq 3$, the following phase transition phenomenon occurs around $\delta_{\mathsf{maj}}$:
\begin{enumerate}
\item If $\delta \in (0,\delta_{\mathsf{maj}})$, the number of vertices per level satisfies $L_k \geq C(\delta,d)\log(k)$ for all sufficiently large $k$ (depending on $\delta$ and $d$), and all Boolean processing functions are the majority rule (where ties are broken by outputting random bits), then reconstruction is possible in the sense that:
$$ \limsup_{k \rightarrow \infty}{\P(X_{k,0} \neq X_{0,0})} < \frac{1}{2} $$
where we use a single vertex $X_{k,0}$ as the decoder at level $k$.
\item If $\delta \in \big[\delta_{\mathsf{maj}},\frac{1}{2}\big)$, $d$ is odd, and the number of vertices per level satisfies $\lim_{k \rightarrow \infty}{L_k} = \infty$ and $R_k \triangleq \inf_{n \geq k}{L_n} = O\big(d^{2k}\big)$, then for all choices of Boolean processing functions (which may vary between vertices and be graph dependent), reconstruction is impossible in the sense that:
$$ \lim_{k \rightarrow \infty}{\E\!\left[\left\|P_{X_{k,0}|G}^+ - P_{X_{k,0}|G}^-\right\|_{\mathsf{TV}}\right]} = 0 $$
where $P_{X_{k,0}|G}^+$ and $P_{X_{k,0}|G}^-$ are the conditional distributions of $X_{k,0}$ given $\{X_{0,0} = 1,G\}$ and $\{X_{0,0} = 0,G\}$, respectively.
\end{enumerate}
\end{proposition}

Proposition \ref{Prop: Single Vertex Reconstruction} is proved in Appendix \ref{Proof of Proposition Single Vertex Reconstruction}. In particular, part 2 of Proposition \ref{Prop: Single Vertex Reconstruction} demonstrates that when $\delta \geq \delta_{\mathsf{maj}}$, the ML decoder based on a single vertex $X_{k,0}$ (with knowledge of the random DAG realization $G$) cannot reconstruct $X_{0,0}$ in all but a vanishing fraction of DAGs. It is worth mentioning that much like part 2 of Theorem \ref{Thm:Phase Transition in Random Grid with Majority Rule Processing}, part 2 of Proposition \ref{Prop: Single Vertex Reconstruction} also implies that:
\begin{equation}
\liminf_{k \rightarrow \infty}{\left\|P_{X_{k,0}|G}^+ - P_{X_{k,0}|G}^-\right\|_{\mathsf{TV}}} = 0 \quad G\text{-}a.s.
\end{equation}
since applying Fatou's lemma to part 2 of Proposition \ref{Prop: Single Vertex Reconstruction} yields $\E\big[\liminf_{k \rightarrow \infty}{\big\|P_{X_{k,0}|G}^+ - P_{X_{k,0}|G}^-\big\|_{\mathsf{TV}}}\big] = 0$. Thus, if reconstruction is possible in the range $\delta \geq \delta_{\mathsf{maj}}$, the decoder should definitely use more than one vertex. This converse result relies on the aforementioned impossibility results on reliable computation. Specifically, the exact threshold $\delta_{\mathsf{maj}}$ that determines whether or not reliable computation is possible using formulae is known for odd $d \geq 3$, cf. \cite{HajekWeller1991, EvansSchulman2003}. Therefore, we can exploit such results to obtain a converse for odd $d \geq 3$ which holds for all choices of Boolean processing functions and at the critical value $\delta = \delta_{\mathsf{maj}}$ (although only for single vertex decoding). In contrast, when $d \geq 4$ is even, it is not even known whether such a critical threshold exists (as noted in \cite[Section 7]{EvansSchulman2003}), and hence, we cannot easily prove such converse results for even $d \geq 4$.\footnote{Note, however, that if all Boolean processing functions are the majority rule and the conditions of part 2 of Theorem \ref{Thm:Phase Transition in Random Grid with Majority Rule Processing} are satisfied, then part 2 of Theorem \ref{Thm:Phase Transition in Random Grid with Majority Rule Processing} implies (using the data processing inequality for TV distance and the bounded convergence theorem) that single vertex reconstruction is also impossible in the sense presented in part 2 of Proposition \ref{Prop: Single Vertex Reconstruction}.} 

We next present an immediate corollary of Theorem \ref{Thm:Phase Transition in Random Grid with Majority Rule Processing} which states that there exist constant indegree deterministic DAGs with $L_k = \Omega(\log(k))$ (i.e. $L_k \geq C(\delta,d) \log(k)$ for some large constant $C(\delta,d)$ and all sufficiently large $k$) such that reconstruction of the root bit is possible. Note that \textit{deterministic DAGs} refer to Bayesian networks on specific realizations of $G$ in the sequel. We will use the same notation as subsection \ref{Random Grid Model} to analyze deterministic DAGs with the understanding that the randomness is engendered by $X_{0,0}$ and the edge BSCs, but not $G$. Formally, we have the following result which is proved in Appendix \ref{Proof of Corollary Existence of Grids where Reconstruction is Possible}.

\begin{corollary}[Existence of DAGs where Reconstruction is Possible]
\label{Cor: Existence of Grids where Reconstruction is Possible}
For every indegree $d \geq 3$, every noise level $\delta \in (0,\delta_{\mathsf{maj}})$, and every sequence of level sizes satisfying $L_k \geq C(\delta,d)\log(k)$ for all sufficiently large $k$, there exists a deterministic DAG $\mathcal{G}$ with these parameters such that if we use majority rules as our Boolean processing functions, then there exists $\epsilon = \epsilon(\delta,d) > 0$ (that depends on $\delta$ and $d$) such that the probability of error in ML decoding is bounded away from $\frac{1}{2} - \epsilon$:
$$ \forall k \in \N, \enspace \P\!\left(h_{\mathsf{ML}}^k(X_k,\mathcal{G}) \neq X_{0}\right) \leq \frac{1}{2} - \epsilon $$
where $h_{\mathsf{ML}}^k(\cdot,\mathcal{G}) : \{0,1\}^{L_k} \rightarrow \{0,1\}$ denotes the ML decision rule at level $k$ based on the full $k$-layer state $X_k$ (given knowledge of the DAG $\mathcal{G}$).
\end{corollary}

Since the critical threshold $\delta_{\mathsf{maj}} \rightarrow \frac{1}{2}$ as $d \rightarrow \infty$, a consequence of Corollary \ref{Cor: Existence of Grids where Reconstruction is Possible} is that for any $\delta \in \big(0,\frac{1}{2}\big)$, any sufficiently large indegree $d$ (that depends on $\delta$), and any sequence of level sizes satisfying $L_k \geq C(\delta,d)\log(k)$ for all sufficiently large $k$, there exists a deterministic DAG $\mathcal{G}$ with these parameters and all majority processing functions such that reconstruction of the root bit is possible in the sense shown above.

Until now, we have restricted ourselves to the $d \geq 3$ case of the random DAG model. Our second main result considers the setting where the indegree of each vertex (except the root) is $d = 2$, because it is not immediately obvious that deterministic DAGs (for which reconstruction is possible) exist for $d = 2$. Indeed, it is not entirely clear which Boolean processing functions are good for ``local error correction'' in this scenario. We choose to fix all Boolean functions at even levels of the random DAG model to be the AND rule, and all Boolean functions at odd levels of the model to be the OR rule. We then prove that this random DAG model also exhibits a phase transition phenomenon around a critical threshold of:
\begin{equation}
\label{Eq: NAND threshold} 
\delta_{\mathsf{andor}} \triangleq \frac{3 - \sqrt{7}}{4} \, . 
\end{equation}
As before, the next theorem illustrates that for $\delta < \delta_{\mathsf{andor}}$, the ``biased'' majority decision rule $\hat{T}_k \triangleq \I\{\sigma_k \geq t\}$, where $t \in (0,1)$ is defined in \eqref{Eq: Middle Fixed Point} in section \ref{Analysis of And-Or Rule Processing in Random Grid}, can asymptotically decode $\sigma_0$, but for $\delta > \delta_{\mathsf{andor}}$, the ML decision rule with knowledge of $G$ cannot asymptotically decode $\sigma_0$. For simplicity, we only analyze this model at even levels in the achievability case.

\begin{theorem}[Phase Transition in Random DAG Model with AND-OR Rule Processing]
\label{Thm:Phase Transition in Random Grid with And-Or Rule Processing}
Let $C(\delta)$ and $D(\delta)$ be the constants defined in \eqref{Eq: Large Constant before Log 2} and \eqref{Eq: Lipschitz constant} in section \ref{Analysis of And-Or Rule Processing in Random Grid}. For a random DAG model with $d = 2$, AND processing functions at even levels, and OR processing functions at odd levels, the following phase transition phenomenon occurs around $\delta_{\mathsf{andor}}$:
\begin{enumerate}
\item If $\delta \in (0,\delta_{\mathsf{andor}})$, and the number of vertices per level satisfies $L_k \geq C(\delta) \log(k)$ for all sufficiently large $k$ (depending on $\delta$), then reconstruction is possible in the sense that:
$$ \limsup_{k \rightarrow \infty}{\P(\hat{T}_{2k} \neq \sigma_0)} < \frac{1}{2} $$
where we use the decoder $\hat{T}_{2k} = \I\{\sigma_{2k} \geq t\}$ at level $2k$, which recovers the root bit by thresholding at the value $t \in (0,1)$ in \eqref{Eq: Middle Fixed Point}.
\item If $\delta \in \big(\delta_{\mathsf{andor}},\frac{1}{2}\big)$, and the number of vertices per level satisfies $L_k = o\big(E(\delta)^{-\frac{k}{2}}\big)$ and $\liminf_{k \rightarrow \infty}{L_k} > \frac{2}{E(\delta) - D(\delta)}$ for any $E(\delta) \in (D(\delta),1)$ (that depends on $\delta$), then reconstruction is impossible in the sense of \eqref{Eq:Strong TV Reconstruction Impossible}:
$$ \lim_{k \rightarrow \infty}{\left\|P_{X_{k}|G}^+ - P_{X_{k}|G}^-\right\|_{\mathsf{TV}}} = 0 \quad G\text{-}a.s. $$
\end{enumerate} 
\end{theorem} 

Theorem \ref{Thm:Phase Transition in Random Grid with And-Or Rule Processing} is proved in section \ref{Analysis of And-Or Rule Processing in Random Grid}, and many of the remarks pertaining to Theorem \ref{Thm:Phase Transition in Random Grid with Majority Rule Processing} as well as the general intuition for Theorem \ref{Thm:Phase Transition in Random Grid with Majority Rule Processing} also hold for Theorem \ref{Thm:Phase Transition in Random Grid with And-Or Rule Processing}. Furthermore, a proposition analogous to part 1 of Proposition \ref{Prop: Single Vertex Reconstruction} and a corollary analogous to Corollary \ref{Cor: Existence of Grids where Reconstruction is Possible} also hold here (but we omit explicit statements of these results for brevity).

It is straightforward to verify that the random DAG in Theorem \ref{Thm:Phase Transition in Random Grid with And-Or Rule Processing} with alternating layers of AND and OR processing functions is equivalent to a random DAG with all NAND processing functions for the purposes of broadcasting.\footnote{Indeed, we can introduce pairs of NOT gates into every edge of our DAG that goes from an AND gate to an OR gate without affecting the statistics of the model. Since an AND gate followed by a NOT gate and an OR gate whose inputs pass through NOT gates are both NAND gates, we obtain an equivalent model where all processing functions are NAND gates. We remark that analyzing this random DAG model with NAND processing functions yields a version of Theorem \ref{Thm:Phase Transition in Random Grid with And-Or Rule Processing} with the same essential characteristics (albeit with possibly weaker conditions on $L_k$).} Recall that in the discussion following Theorem \ref{Thm:Phase Transition in Random Grid with Majority Rule Processing}, we noted how the critical threshold $\delta_{\mathsf{maj}}$ was already known in the reliable computation literature (because it characterized when reliable computation is possible), cf. \cite{EvansSchulman2003}. It turns out that $\delta_{\mathsf{andor}}$ has also appeared in the reliable computation literature in a similar vein. In particular, although the existence of critical thresholds on $\delta$ for reliable computation using formulae of $\delta$-noisy gates is not known for any even $d \geq 4$, the special case of $d = 2$ has been resolved. Indeed, Evans and Pippenger showed in \cite{EvansPippenger1998} that reliable computation using formulae consisting of $\delta$-noisy NAND gates is possible when $\delta < \delta_{\mathsf{andor}}$, and impossible when $\delta > \delta_{\mathsf{andor}}$ for ``soft inputs.'' Moreover, Unger established in \cite{Unger2007,Unger2008} that reliable computation using formulae with general $2$-input $\delta$-noisy gates is impossible when $\delta \geq \delta_{\mathsf{andor}}$. 

\subsection{Explicit Construction of Deterministic DAGs where Broadcasting is Possible}
\label{Explicit Construction of DAGs where Broadcasting is Possible}

Although Corollary \ref{Cor: Existence of Grids where Reconstruction is Possible} illustrates the existence of DAGs where broadcasting (i.e. reconstruction of the root bit) is possible, it does not elucidate the structure of such DAGs. Moreover, Theorem \ref{Thm:Phase Transition in Random Grid with Majority Rule Processing} suggests that reconstruction on such deterministic DAGs should be possible using the algorithmically simple majority decision rule, but Corollary \ref{Cor: Existence of Grids where Reconstruction is Possible} is proved for the typically more complex ML decision rule. In this subsection, we address these deficiencies of Corollary \ref{Cor: Existence of Grids where Reconstruction is Possible} by presenting an explicit construction of deterministic bounded degree DAGs such that $L_k = \Theta(\log(k))$ and reconstruction of the root bit is possible using the majority decision rule. 

Our construction is based on \textit{regular bipartite lossless expander graphs}. Historically, the notion of an expander graph goes back to the work of Kolmogorov and Barzdin in \cite{KolmogorovBarzdin1967}. Soon afterwards, Pinsker independently discovered such graphs and coined the term ``expander graph'' in \cite{Pinsker1973}.\footnote{In fact, expander graphs are called ``expanding'' graphs in \cite{Pinsker1973}.} Both \cite{KolmogorovBarzdin1967} and \cite[Lemma 1]{Pinsker1973} prove the existence of expander graphs using probabilistic techniques. On the other hand, the first explicit construction of expander graphs appeared in \cite{Margulis1973}, and more recently, lossless expander graphs were constructed using simpler ideas in \cite{Capalboetal2002}. We next define a pertinent variant of lossless expander graphs. 

Consider a \textit{$d$-regular bipartite graph} $B = (U,V,E)$, where $U$ and $V$ are two disjoint sets of vertices such that $|U| = |V| = n \in \N\backslash\!\{0\}$, every vertex in $U \cup V$ has degree $d \in \N\backslash\!\{0\}$, and $E$ is the set of undirected edges between $U$ and $V$. Note that we allow multiple edges to exist between two vertices in $B$. For any subset of vertices $S \subseteq U$, we define the \textit{neighborhood} of $S$ as:
\begin{equation}
\Gamma(S) \triangleq \left\{v \in V : \exists u \in S, \, \{u,v\} \in E \right\}
\end{equation}
which is the set of all vertices in $V$ that are adjacent to some vertex in $S$. For any fraction $\alpha \in (0,1)$ and any expansion factor $\beta > 0$, $B$ is called an $(\alpha,\beta)$-\textit{expander graph} if for every subset of vertices $S \subseteq U$, we have:
\begin{equation}
\label{Eq: Expansion Property}
|S| \leq \alpha n \enspace \Rightarrow \enspace |\Gamma(S)| \geq \beta |S| \, . 
\end{equation}
Note that we only require subsets of vertices in $U$ to expand (not $V$). Intuitively, such expander graphs are sparse due to the $d$-regularity constraint, but have high connectivity due to the expansion property \eqref{Eq: Expansion Property}. Furthermore, when $\alpha \leq \frac{1}{d}$, the best expansion factor one can hope for is $\beta$ as close as possible to $d$. Hence, $(\alpha,(1 - \epsilon) d)$-expander graphs with $\alpha \leq \frac{1}{d}$ and very small $\epsilon > 0$ are known as \textit{lossless} expander graphs \cite[Section 1.1]{Capalboetal2002}. 

We utilize a slightly relaxed version of lossless expander graphs in our construction. In particular, using existing results from the literature, we establish in Corollary \ref{Cor: Lossless Expander Graph} of section \ref{Deterministic Quasi-Polynomial Time and Randomized Polylogarithmic Time Constructions of DAGs where Broadcasting is Possible} that for large values of the degree $d$ and any sufficiently large $n$ (depending on $d$), there exists a $d$-regular bipartite graph $B = (U,V,E)$ with $|U| = |V| = n$ such that for every subset of vertices $S \subseteq U$, we have:\footnote{We do not explicitly impose the constraint that $\epsilon = 2/d^{1/5} < 1$ because the constraint \eqref{Eq: Relation between d and delta} in Theorem \ref{Thm: Reconstruction in Expander DAGs} implicitly ensures this.}
\begin{equation}
\label{Eq: Expansion Property 2}
|S| = \frac{n}{d^{6/5}} \enspace \Rightarrow \enspace |\Gamma(S)| \geq (1-\epsilon) \frac{n}{d^{1/5}} \text{ with } \epsilon = \frac{2}{d^{1/5}} \, .
\end{equation}
Unlike \eqref{Eq: Expansion Property}, the expansion in \eqref{Eq: Expansion Property 2} only holds for subsets $S \subseteq U$ with cardinality exactly $|S| = n d^{-6/5}$. However, we can still (loosely) perceive the graph $B$ as a $d$-regular bipartite lossless $(\alpha,\beta)$-expander graph with $\alpha = d^{-6/5}$ and $\beta = (1 - \epsilon)d$. (Strictly speaking, $n d^{-6/5}$ must be an integer, but we neglect this detail throughout our exposition for simplicity.) In the remainder of our discussion, we refer to graphs like $B$ that satisfy \eqref{Eq: Expansion Property 2} as \textit{$d$-regular bipartite lossless $(d^{-6/5},d - 2 d^{4/5})$-expander graphs} with abuse of standard nomenclature. 

A $d$-regular bipartite lossless $(d^{-6/5},d - 2 d^{4/5})$-expander graph $B$ can be construed as representing two consecutive levels of a deterministic DAG upon which we are broadcasting. Indeed, we can make every edge in $E$ directed by making them point from $U$ to $V$, where $U$ represents a particular level in the DAG and $V$ the next level. In fact, we can construct deterministic DAGs where broadcasting is possible by concatenating several such $d$-regular bipartite lossless expander graphs together. The ensuing theorem details our expander-based DAG construction, and illustrates that reconstruction of the root bit is possible when we use majority Boolean processing functions and the majority decision rule $\hat{S}_k = \I\big\{\sigma_k \geq \frac{1}{2}\big\}$, where $\sigma_k$ is defined in \eqref{Eq: Empirical Probability of Unity Definition}.

\begin{theorem}[Reconstruction in Expander DAGs]
\label{Thm: Reconstruction in Expander DAGs}
Fix any noise level $\delta \in \big(0,\frac{1}{2}\big)$, any sufficiently large odd degree $d = d(\delta) \geq 5$ (that depends on $\delta$) satisfying:
\begin{equation}
\label{Eq: Relation between d and delta}
\frac{8}{d^{1/5}} + d^{6/5} \exp\!\left(-\frac{(1 - 2\delta)^2 (d - 4)^2}{8 d}\right) \leq \frac{1}{2} \, ,
\end{equation}
and any sufficiently large constant $N = N(\delta) \in \N$ (that depends on $\delta$) such that the constant $M \triangleq \exp(N/(4 d^{12/5})) \geq 2$ and for every $n \geq N$, there exists a $d$-regular bipartite lossless $(d^{-6/5},d - 2 d^{4/5})$-expander graph $B_n = (U_n,V_n,E_n)$ with $|U_n| = |V_n| = n$ that satisfies \eqref{Eq: Expansion Property 2} for every subset $S \subseteq U_n$. Consider an infinite deterministic DAG with indegrees bounded by $d$, outdegrees bounded by $2d$, sequence of level sizes $\{L_k : k \in \N\}$ given by:
\begin{subequations}
\label{Eq: Expander Level Sizes}
\begin{align}
L_0 & = 1 \\
\forall k \in \left\{1,\dots,\floor{M}\right\}, \enspace L_k & = N \\
\forall m \in \N\backslash\!\{0\}, \, \forall k \in \N \enspace \text{such that} \quad \quad \enspace & \nonumber \\
M^{2^{m-1}} < k \leq M^{2^{m}}, \enspace L_k & = 2^m N  
\end{align}
\end{subequations}
where $L_k = \Theta(\log(k))$, and the following edge configuration:
\begin{enumerate}
\item[E-1)] Every vertex in $X_1$ has one directed edge coming from $X_{0,0}$.
\item[E-2)] For every pair of consecutive levels $k$ and $k+1$ such that $L_{k + 1} = L_k$, the directed edges from $X_k$ to $X_{k+1}$ are given by the edges of $B_{L_k}$, where we identify the vertices in $U_{L_k}$ with $X_k$ and the vertices in $V_{L_k}$ with $X_{k+1}$, respectively.
\item[E-3)] For every pair of consecutive levels $k$ and $k+1$ such that $L_{k+1} = 2 L_k$, we partition the vertices in $X_{k+1}$ into two sets, $X_{k+1}^{1} = (X_{k+1,0},\dots,X_{k+1,L_k - 1})$ and $X_{k+1}^{2} = (X_{k+1,L_k},\dots,X_{k+1,L_{k+1} - 1})$, so that the directed edges from $X_k$ to $X_{k+1}^{i}$ are given by the edges of $B_{L_k}$ for $i = 1,2$, where we identify the vertices in $U_{L_k}$ with $X_k$ and the vertices in $V_{L_k}$ with $X_{k+1}^{i}$, respectively, as before.
\end{enumerate}
For the Bayesian network defined on this fixed DAG (similar to subsection \ref{Random Grid Model}) with $X_{0,0} \sim \Ber\big(\frac{1}{2}\big)$, independent $\mathsf{BSC}(\delta)$ edges, all identity Boolean processing functions in level $k = 1$, and all majority rule Boolean processing functions in levels $k \geq 2$, reconstruction is possible in the sense that:
$$ \limsup_{k \rightarrow \infty}{\P\!\left(\hat{S}_{k} \neq X_{0}\right)} < \frac{1}{2} $$
where we use the majority decoder $\hat{S}_k = \I\big\{\sigma_k \geq \frac{1}{2}\big\}$ at level $k$.
\end{theorem}

Theorem \ref{Thm: Reconstruction in Expander DAGs} is proved in section \ref{Deterministic Quasi-Polynomial Time and Randomized Polylogarithmic Time Constructions of DAGs where Broadcasting is Possible}. This proof of feasibility of reconstruction follows the same overarching strategy as the proof of Theorem \ref{Thm:Phase Transition in Random Grid with Majority Rule Processing}, but obviously makes essential use of the expansion property \eqref{Eq: Expansion Property 2}. We emphasize that although we reuse the majority decoder notation $\hat{S}_k$ from the random DAG setting, the reconstruction statement in Theorem \ref{Thm: Reconstruction in Expander DAGs} pertains to a fixed deterministic DAG, i.e. there is no averaging over a random DAG in the above probability of error for majority decoding. 

We next present a proposition that describes the computational complexity of our expander-based DAG construction in Theorem \ref{Thm: Reconstruction in Expander DAGs} for which broadcasting is possible.
 
\begin{proposition}[Computational Complexity of DAG Construction]
\label{Prop: DAG Construction using Expander Graphs}
For any fixed noise level $\delta \in \big(0,\frac{1}{2}\big)$, consider the infinite deterministic DAG from Theorem \ref{Thm: Reconstruction in Expander DAGs} with sufficiently large odd indegree $d = d(\delta) \geq 5$ satisfying \eqref{Eq: Relation between d and delta}, level sizes given by \eqref{Eq: Expander Level Sizes} for sufficiently large $N = N(\delta) \in \N$ and $M = \exp(N/(4 d^{12/5})) \geq 2$, and edge configuration given by E-1, E-2, and E-3. Then, we can construct the constituent $d$-regular bipartite lossless $(d^{-6/5},d - 2 d^{4/5})$-expander graphs for levels $0,\dots,r$ of this deterministic DAG in:
\begin{enumerate}
\item either deterministic quasi-polynomial time:
$$ O(\exp(\Theta(\log(r) \log\log(r)))) $$
\item or randomized polylogarithmic time:
$$ O(\log(r) \log\log(r)) $$ 
with strictly positive success probability \eqref{Eq: Success Probability of Monte Carlo Algorithm}, if $N$ additionally satisfies \eqref{Eq: Additional Assumption}.
\end{enumerate}
\end{proposition}

Proposition \ref{Prop: DAG Construction using Expander Graphs} is also proved in section \ref{Deterministic Quasi-Polynomial Time and Randomized Polylogarithmic Time Constructions of DAGs where Broadcasting is Possible}. We remark that at its heart, Proposition \ref{Prop: DAG Construction using Expander Graphs} is not concerned with the noise level $\delta$, and holds verbatim for the infinite deterministic DAGs described in Theorem \ref{Thm: Reconstruction in Expander DAGs} with large values of degree $d$ and sufficiently large $N$ (depending on $d$). However, the above statement of Proposition \ref{Prop: DAG Construction using Expander Graphs}, together with Theorem \ref{Thm: Reconstruction in Expander DAGs}, transparently portrays that for every $\delta \in \big(0,\frac{1}{2}\big)$, a deterministic DAG with sufficiently large indegree and $L_k = \Theta(\log(k))$ that admits reconstruction can be \textit{efficiently} constructed. Specifically, Proposition \ref{Prop: DAG Construction using Expander Graphs} conveys that the constituent expander graphs of such deterministic DAGs can be constructed either in quasi-polynomial time or in randomized polylogarithmic time in the number of levels. Theorem \ref{Thm: Reconstruction in Expander DAGs} conveys that once such a deterministic DAG is constructed, reconstruction of the root bit is guaranteed to succeed using the majority decoder. Finally, we note that the question of finding a deterministic polynomial time algorithm to construct DAGs where reconstruction is possible remains open.

\subsection{Further Discussion and Impossibility Results}
\label{Further Discussion}

In this subsection, we present and discuss some impossibility results pertaining to both deterministic and random DAGs. The first result illustrates that if $L_k \leq \log(k)/(d \log(1/(2\delta)))$ for every sufficiently large $k$ (i.e. $L_k$ grows very ``slowly''), then reconstruction is impossible regardless of the choices of Boolean processing functions and the choice of decision rule.

\begin{proposition}[Slow Growth of Layers]
\label{Prop: Slow Growth of Layers}
For any noise level $\delta \in \big(0,\frac{1}{2}\big)$ and indegree $d \in \N\backslash\!\{0\}$, if the number of vertices per level satisfies $L_k \leq \log(k)/(d \log(1/(2\delta)))$ for all sufficiently large $k$, then for all choices of Boolean processing functions (which may vary between vertices and be graph dependent), reconstruction is impossible in the sense that:
\begin{enumerate}
\item for a deterministic DAG:
$$ \lim_{k \rightarrow \infty}{\left\|P_{X_k}^+ - P_{X_k}^-\right\|_{\mathsf{TV}}} = 0 $$
where $P_{X_k}^+$ and $P_{X_k}^-$ denote the conditional distributions of $X_k$ given $X_0 = 1$ and $X_0 = 0$, respectively.
\item for a random DAG:
$$ \lim_{k \rightarrow \infty}{\left\|P_{X_{k}|G}^+ - P_{X_{k}|G}^-\right\|_{\mathsf{TV}}} = 0 \quad \text{pointwise} $$
which means that the condition holds for every realization of the random DAG $G$.
\end{enumerate}
\end{proposition} 

This proposition is proved in Appendix \ref{Proof of Proposition Slow Growth of Layers}. Part 1 of Proposition \ref{Prop: Slow Growth of Layers} illustrates that when $L_k$ is sub-logarithmic, the ML decoder based on the entire $k$-layer state $X_k$ with knowledge of the deterministic DAG fails to reconstruct the root bit. Similarly, part 2 of Proposition \ref{Prop: Slow Growth of Layers} shows that reconstruction is impossible for random DAGs even if the particular DAG realization $G$ is known and the ML decoder can access $X_k$. Therefore, Proposition \ref{Prop: Slow Growth of Layers} illustrates that our assumption that $L_k \geq C \log(k)$, for some constant $C$ (that depends on $\delta$ and $d$) and all sufficiently large $k$, for reconstruction to be possible in Theorems \ref{Thm:Phase Transition in Random Grid with Majority Rule Processing} and \ref{Thm:Phase Transition in Random Grid with And-Or Rule Processing} is in fact necessary. 

In contrast, consider a deterministic DAG with no restrictions (i.e. no bounded indegree assumption) except for the size of $L_k$. Then, each vertex at level $k$ of this DAG is connected to all $L_{k-1}$ vertices at level $k-1$. The next proposition illustrates that $L_k = \Theta\big(\sqrt{\log(k)}\big)$ is the critical scaling of $L_k$ in this scenario. In particular, reconstruction is possible when $L_k = \Omega\big(\sqrt{\log(k)}\big)$ (i.e. $L_k \geq A(\delta) \sqrt{\log(k)}$ for some large constant $A(\delta)$ and all sufficiently large $k$), and reconstruction is impossible when $L_k = O\big(\sqrt{\log(k)}\big)$ (i.e. $L_k \leq B(\delta) \sqrt{\log(k)}$ for some small constant $B(\delta)$ and all sufficiently large $k$). The proof of this result is deferred to Appendix \ref{Proof of Proposition Broadcasting in Unbounded Degree DAG Model}.

\begin{proposition}[Broadcasting in Unbounded Degree DAG Model]
\label{Prop:Broadcasting in Unbounded Degree DAG Model}
Let $A(\delta)$ and $B(\delta)$ be the constants defined in \eqref{Eq: Specialized Level Size Constant} and \eqref{Eq: Converse Level Size Constant Definition} in Appendix \ref{Proof of Proposition Broadcasting in Unbounded Degree DAG Model}. Consider a deterministic DAG $\mathcal{G}$ such that for every $k \in \N\backslash\!\{0\}$, each vertex at level $k$ has one incoming edge from all $L_{k-1}$ vertices at level $k-1$. Then, for any noise level $\delta \in \big(0,\frac{1}{2}\big)$, we have:
\begin{enumerate}
\item If the number of vertices per level satisfies $L_k \geq A(\delta) \sqrt{\log(k)}$ for all sufficiently large $k$, and all Boolean processing functions in $\mathcal{G}$ are the majority rule (where ties are broken by outputting $1$), then reconstruction is possible in the sense that:
$$ \limsup_{k \rightarrow \infty}{\P(\hat{S}_{k} \neq X_0)} < \frac{1}{2} $$
where we use the majority decoder $\hat{S}_k = \I\big\{\sigma_k \geq \frac{1}{2}\big\}$ at level $k$.
\item If the number of vertices per level satisfies $L_k \leq B(\delta) \sqrt{\log(k)}$ for all sufficiently large $k$, then for all choices of Boolean processing functions (which may vary between vertices), reconstruction is impossible in the sense that: 
$$ \lim_{k \rightarrow \infty}{\left\|P_{X_k}^+ - P_{X_k}^-\right\|_{\mathsf{TV}}} = 0 \, . $$
\end{enumerate} 
\end{proposition}

The last impossibility result we present here is an important result from the reliable computation literature due to Evans and Schulman \cite{EvansSchulman1999}. Evans and Schulman studied von Neumann's noisy computation model (which we briefly discussed in subsection \ref{Results on Random DAG Models}), and established general conditions under which reconstruction is impossible in deterministic DAGs due to the decay of mutual information between $X_0$ and $X_k$. Recall that for two discrete random variables $X \in \X$ and $Y \in \Y$ (where $|\X|,|\Y| < \infty$), with joint probability mass function $P_{X,Y}$ and marginals $P_X$ and $P_Y$ respectively, the \textit{mutual information} (in bits) between them is defined as:
\begin{equation}
I(X;Y) \triangleq \sum_{x \in \X}\sum_{y \in \Y}{P_{X,Y}(x,y) \log_2\!\left(\frac{P_{X,Y}(x,y)}{P_X(x)P_Y(y)}\right)} 
\end{equation}
where $\log_2(\cdot)$ is the binary logarithm, and we assume that $0 \log_2\!\big(\frac{0}{q}\big) = 0$ for any $q \geq 0$, and $p \log_2\!\big(\frac{p}{0}\big) = \infty$ for any $p > 0$ (due to continuity considerations). We present a specialization of \cite[Lemma 2]{EvansSchulman1999} for our setting as Proposition \ref{Prop: Evans Schulman} below. This proposition portrays that if $L_k$ is sub-exponential and the parameters $\delta$ and $d$ satisfy $(1 - 2\delta)^2 d < 1$, then reconstruction is impossible in deterministic DAGs regardless of the choices of Boolean processing functions and the choice of decision rule.

\begin{proposition}[Decay of Mutual Information {\cite[Lemma 2]{EvansSchulman1999}}]
\label{Prop: Evans Schulman}
For any deterministic DAG model, we have:
$$ I(X_0;X_k) \leq L_k \left((1-2\delta)^2 d\right)^k $$
where $L_k d^k$ is the total number of paths from $X_0$ to layer $X_k$, and $(1-2\delta)^{2k}$ can be construed as the overall contraction of mutual information along each path. Therefore, if $(1 - 2\delta)^2 d < 1$ and $L_k = o\big(1/((1-2\delta)^2 d)^k\big)$, then for all choices of Boolean processing functions (which may vary between vertices), we have:
$$ \lim_{k \rightarrow \infty}{I(X_0;X_k)} = 0 $$
which implies, by Pinsker's inequality, that:
$$ \lim_{k \rightarrow \infty}{\left\|P_{X_k}^+ - P_{X_k}^-\right\|_{\mathsf{TV}}} = 0 \, . $$ 
\end{proposition}

We make some pertinent remarks about this result. Firstly, Evans and Schulman's original analysis assumes that gates are noisy as opposed to edges (in accordance with von Neumann's setup), but the re-derivation of \cite[Lemma 2]{EvansSchulman1999} in \cite[Corollary 7]{PolyanskiyWu2017} illustrates that the result also holds for our model. In fact, the \textit{site percolation} analysis in \cite[Section 3]{PolyanskiyWu2017} (which we will briefly delineate later) improves upon Evans and Schulman's estimate. Furthermore, this analysis illustrates that the bound in Proposition \ref{Prop: Evans Schulman} also holds for all choices of random Boolean processing functions.

Secondly, while Proposition \ref{Prop: Evans Schulman} holds for deterministic DAGs, we can easily extend it for random DAG models. Indeed, the random DAG model inherits the inequality in Proposition \ref{Prop: Evans Schulman} pointwise:
\begin{equation}
\label{Eq: Intermediate MI Bound given G}
I(X_0;X_k|G = \mathcal{G}) \leq L_k \left((1-2\delta)^2 d\right)^k
\end{equation}
for every realization of the random DAG $G = \mathcal{G}$, where $I(X_0;X_k|G = \mathcal{G})$ is the mutual information between $X_0$ and $X_k$ computed using the joint distribution of $X_0$ and $X_k$ given $G = \mathcal{G}$. This implies that if $L_k$ is sub-exponential and $(1 - 2\delta)^2 d < 1$, then reconstruction based on $X_k$ is impossible regardless of the choices of Boolean processing functions (which may vary between vertices and be graph dependent) and the choice of decision rule even if the decoder knows the particular random DAG realization, i.e. $\lim_{k \rightarrow \infty}{\big\|P_{X_k|G}^+ - P_{X_k|G}^-\big\|_{\mathsf{TV}}} = 0$ pointwise (which trivially implies \eqref{Eq:Strong TV Reconstruction Impossible}). Taking expectations with respect to $G$ in \eqref{Eq: Intermediate MI Bound given G}, we get:
\begin{equation}
\label{Eq: Seemingly Weaker MI Condition}
I(X_0;X_k) \leq I(X_0;X_k|G) \leq L_k \left((1-2\delta)^2 d\right)^k
\end{equation}
where $I(X_0;X_k|G)$ is the conditional mutual information (i.e. the expected value of $I(X_0;X_k|G = \mathcal{G})$ with respect to $G$), and the first inequality follows from the chain rule for mutual information and the fact that $X_0$ is independent of $G$. Since the second inequality in \eqref{Eq: Seemingly Weaker MI Condition} implies \eqref{Eq: Auxiliary Convergence Condition}, invoking the argument at the end of the proof of part 2 of Theorem \ref{Thm:Phase Transition in Random Grid with Majority Rule Processing} in section \ref{Analysis of Majority Rule Processing in Random Grid} also yields that reconstruction is impossible in the sense of \eqref{Eq:Strong TV Reconstruction Impossible} when $L_k$ is sub-exponential and $(1 - 2\delta)^2 d < 1$. Thus, $\lim_{k \rightarrow \infty}{I(X_0;X_k|G)} = 0$ is a sufficient condition for \eqref{Eq:Strong TV Reconstruction Impossible}. In contrast, the first inequality in \eqref{Eq: Seemingly Weaker MI Condition} only yields the impossibility of reconstruction in the sense of \eqref{Eq:TV Reconstruction Impossible} when $L_k$ is sub-exponential and $(1 - 2\delta)^2 d < 1$.

Thirdly, Evans and Schulman's result in Proposition \ref{Prop: Evans Schulman} provides an upper bound on the critical threshold of $\delta$ above which reconstruction of the root bit is impossible. Indeed, the condition, $(1 - 2\delta)^2 d < 1$, under which mutual information decays can be rewritten as (cf. the discussion in \cite[p.2373]{EvansSchulman1999}):
\begin{equation}
\delta_{\mathsf{ES}}(d) \triangleq \frac{1}{2} - \frac{1}{2\sqrt{d}} < \delta < \frac{1}{2}
\end{equation}
and reconstruction is impossible for deterministic or random DAGs in this regime of $\delta$ provided $L_k$ is sub-exponential. As a sanity check, we can verify that $\delta_{\mathsf{ES}}(2) = 0.14644... > 0.08856... = \delta_{\mathsf{andor}}$ in the context of Theorem \ref{Thm:Phase Transition in Random Grid with And-Or Rule Processing}, and $\delta_{\mathsf{ES}}(3) = 0.21132... > 0.16666... = \delta_{\mathsf{maj}}$ in the context of Theorem \ref{Thm:Phase Transition in Random Grid with Majority Rule Processing} with $d = 3$. Although $\delta_{\mathsf{ES}}(d)$ is a general upper bound on the critical threshold for reconstruction, in this paper, it is not particularly useful because we analyze explicit processing functions and decision rules, and derive specific bounds that characterize the corresponding thresholds.

Fourthly, it is worth comparing $\delta_{\mathsf{ES}}(d)$ (which comes from a site percolation argument, cf. \cite[Section 3]{PolyanskiyWu2017}) to an upper bound on the critical threshold for reconstruction derived from \textit{bond percolation} (similar to the percolation analysis in \cite[p.570]{Feder1989}). To this end, consider the random DAG model, and recall that the $\mathsf{BSC}(\delta)$'s along each edge generate independent bits with probability $2\delta$ (as shown in the proof of Proposition \ref{Prop: Slow Growth of Layers} in Appendix \ref{Proof of Proposition Slow Growth of Layers}). So, we can perform bond percolation so that each edge is independently ``removed'' with probability $2\delta$. It can be shown by analyzing this bond percolation process that reconstruction is impossible (in a certain sense) when $\frac{1}{2} - \frac{1}{2d} < \delta < \frac{1}{2}$. Therefore, the Evans-Schulman upper bound of $\delta_{\mathsf{ES}}(d)$ is tighter than the bond percolation upper bound: $\delta_{\mathsf{ES}}(d) < \frac{1}{2} - \frac{1}{2d}$.

Finally, we briefly delineate how the site percolation approach in \cite[Section 3]{PolyanskiyWu2017} allows us to prove that reconstruction is impossible in the random DAG model for the $(1-2\delta)^2 d = 1$ case as well. Consider a site percolation process where each vertex $X_{k,j}$ (for $k \in \N\backslash\!\{0\}$ and $j \in [L_k]$) is independently ``open'' with probability $(1-2\delta)^2$, and ``closed'' with probability $1-(1-2\delta)^2$. (Note that $X_{0,0}$ is open almost surely.) For every $k \in \N\backslash\!\{0\}$, let $p_k$ denote the probability that there is an ``open connected path'' from $X_{0}$ to $X_{k}$ (i.e. there exist $j_1 \in [L_1],\dots,j_k \in [L_k]$ such that $(X_{0,0},X_{1,j_1}),(X_{1,j_1},X_{2,j_2}),\dots,(X_{k-1,j_{k-1}},X_{k,j_k})$ are directed edges in the random DAG $G$ and $X_{1,j_1},\dots,X_{k,j_k}$ are all open). It can be deduced from \cite[Theorem 5]{PolyanskiyWu2017} that for any $k \in \N\backslash\!\{0\}$: 
\begin{equation}
\label{Eq: Polyanskiy-Wu Bound}
I(X_0;X_k|G) \leq p_k \, . 
\end{equation} 
Next, for each $k \in \N$, define the random variable:
\begin{equation}
\lambda_k \triangleq \frac{1}{L_k} \sum_{j \in [L_k]}{\I\!\left\{X_{k,j} \text{ is open and connected}\right\}} 
\end{equation}
which is the proportion of open vertices at level $k$ that are connected to the root by an open path. (Note that $\lambda_0 = 1$.) It is straightforward to verify (using Bernoulli's inequality) that for any $k \in \N\backslash\!\{0\}$:
\begin{align}
\label{Eq: Other Equality Recursion}
\E\!\left[\lambda_k|\lambda_{k-1}\right] & = (1-2\delta)^2 \!\left(1-(1-\lambda_{k-1})^d\right) \\
& \leq (1-2\delta)^2 d \lambda_{k-1} \, . 
\label{Eq: Recursion}
\end{align}
Observe that by Markov's inequality and the recursion from \eqref{Eq: Recursion}, $\E[\lambda_k] \leq (1-2\delta)^2 d \, \E[\lambda_{k-1}]$, we have:
\begin{equation}
\label{Eq: Markov argument}
p_k = \P\!\left(\lambda_k \geq \frac{1}{L_k}\right) \leq L_k \E\!\left[\lambda_k\right] \leq L_k \left((1-2\delta)^2 d\right)^k
\end{equation}
which recovers Evans and Schulman's result (Proposition \ref{Prop: Evans Schulman}) in the context of the random DAG model. Indeed, if $(1-2\delta)^2 d < 1$ and $L_k = o\big(1/((1-2\delta)^2 d)^k\big)$, then $\lim_{k \rightarrow \infty}{p_k} = 0$, and as a result, $\lim_{k \rightarrow \infty}{I(X_0;X_k|G)} = 0$ by \eqref{Eq: Polyanskiy-Wu Bound}. On the other hand, when $(1-2\delta)^2 d = 1$, taking expectations and applying Jensen's inequality to the equality in \eqref{Eq: Other Equality Recursion} produces:
\begin{equation}
\E\!\left[\lambda_k\right] \leq (1-2\delta)^2 \!\left(1-(1-\E\!\left[\lambda_{k-1}\right])^d\right) . 
\end{equation}
This implies that $\E[\lambda_k] \leq F^{-1}(k)$ for every $k \in \N$ using the estimate in \cite[Appendix A]{PolyanskiyWu2016}, where $F:[0,1] \rightarrow \R_{+}, \, F(t) = \int_{t}^{1}{\frac{1}{f(\tau)} \, d\tau}$ with $f:[0,1] \rightarrow [0,1], \, f(t) = t - (1-2\delta)^2 \big(1-(1-t)^d\big)$, and $F^{-1}:\R_{+} \rightarrow [0,1]$ is well-defined. Since $f(t) \geq \frac{d-1}{2} t^2$ for all $t \in [0,1]$, it is straightforward to show that:
\begin{equation}
\E\!\left[\lambda_k\right] \leq F^{-1}(k) \leq \frac{2}{(d-1) k} \, . 
\end{equation}
Therefore, the Markov's inequality argument in \eqref{Eq: Markov argument} illustrates that if $(1-2\delta)^2 d = 1$ and $L_k = o(k)$, then $\lim_{k \rightarrow \infty}{p_k} = 0$ and reconstruction is impossible in the random DAG model due to \eqref{Eq: Polyanskiy-Wu Bound}. Furthermore, the condition on $L_k$ can be improved to $L_k = O(k \log(k))$ using a more sophisticated Borel-Cantelli type of argument.

\section{Analysis of Majority Rule Processing in Random DAG Model}
\label{Analysis of Majority Rule Processing in Random Grid}

In this section, we prove Theorem \ref{Thm:Phase Transition in Random Grid with Majority Rule Processing}. To this end, we first make some pertinent observations. Recall that we have a random DAG model with $d \geq 3$, and all Boolean functions are the majority rule, i.e. $f_{k}(x_1,\dots,x_d) = \maj(x_1,\dots,x_d)$ for every $k \in \N\backslash\!\{0\}$. Note that when the number of $1$'s is equal to the number of $0$'s, the majority rule outputs an independent $\Ber\big(\frac{1}{2}\big)$ bit.\footnote{Although generating a random bit is a natural approach to breaking ties in the majority rule, this means that the rule is no longer purely deterministic when $d$ is even.} Suppose we are given that $\sigma_{k-1} = \sigma$ for any $k \in \N\backslash\!\{0\}$. Then, for every $j \in [L_{k}]$, $X_{k,j} = \maj(Y_1,\dots,Y_d)$ where $Y_1,\dots,Y_d$ are i.i.d. $\Ber(p)$ random variables with $p = \sigma * \delta$. Define the function $g:[0,1] \rightarrow [0,1]$ as follows: 
\begin{align}
g(\sigma) & \triangleq \E\!\left[\maj(Y_1,\dots,Y_d)\right] \\
& = \P\!\left(\sum_{i = 1}^{d}{Y_i} > \frac{d}{2}\right) + \frac{1}{2}\P\!\left(\sum_{i = 1}^{d}{Y_i} = \frac{d}{2}\right) \\
& = \left\{ 
    \begin{array}{ll}
      \displaystyle{\sum_{i = \frac{d}{2} + 1}^{d}{\binom{d}{i} (\sigma * \delta)^i (1 - \sigma * \delta)^{d - i}}} & \\
			\quad \displaystyle{+ \, \frac{1}{2}\binom{d}{\frac{d}{2}} (\sigma * \delta)^{\frac{d}{2}} (1 - \sigma * \delta)^{^{\frac{d}{2}}}} &, \enspace d \text{ even} \\
      \displaystyle{\sum_{i = \frac{d+1}{2}}^{d}{\binom{d}{i} (\underbrace{\sigma * \delta}_{p})^i (\underbrace{1 - \sigma * \delta}_{1-p})^{d - i}}} &, \enspace d \text{ odd}
    \end{array}
	\right.
\label{Eq:Binomial form of g}
\end{align}
which implies that $X_{k,j}$ are i.i.d. $\Ber(g(\sigma))$ for $j \in [L_{k}]$, and $L_k \sigma_k \sim \mathsf{binomial}(L_k,g(\sigma))$, since we have:
\begin{equation}
\label{Eq:Conditional Expectation}
\P(X_{k,j} = 1|\sigma_{k-1} = \sigma) = \E[\sigma_k|\sigma_{k-1} = \sigma] = g(\sigma) \, .
\end{equation}

To compute the first derivative of $g$, we follow the analysis in \cite[Section 2]{Mossel1998}. Recall that a Boolean function $h:\{0,1\}^d \rightarrow \{0,1\}$ is \textit{monotone non-decreasing} (respectively, \textit{non-increasing}) if its value either increases (respectively, decreases) or remains the same whenever any of its input bits is flipped from $0$ to $1$. For any such monotone function $h:\{0,1\}^d \rightarrow \{0,1\}$, the \textit{Margulis-Russo formula} states that \cite{Margulis1974,Russo1981} (alternatively, see \cite[Section 4.1]{Grimmett1997}):
\begin{equation}
\label{Eq:Margulis-Russo formula}
\begin{aligned}
& \frac{d}{dp}\E\!\left[h(Y_1,\dots,Y_d)\right] \\
& \quad \quad = \sum_{i = 1}^{d} \Big(\E\!\left[h(Y_1,\dots,Y_{i-1},1,Y_{i+1},\dots,Y_d)\right] \\
& \quad \quad \quad \quad \quad \, \, - \E\!\left[h(Y_1,\dots,Y_{i-1},0,Y_{i+1},\dots,Y_d)\right]\!\Big) \, .
\end{aligned}
\end{equation}
Hence, since $h = \maj$ is a non-decreasing function, $g^{\prime}:[0,1] \rightarrow \R_{+}$ is given by:
\begin{align}
g^{\prime}(\sigma) & = \frac{dp}{d\sigma} \frac{d}{dp} \E\!\left[h(Y_1,\dots,Y_d)\right] \nonumber \\
& = (1 - 2\delta) \sum_{i = 1}^{d} \Big(\E\!\left[h(Y_1,\dots,Y_{i-1},1,Y_{i+1},\dots,Y_d)\right] \nonumber \\
& \quad \quad \quad \quad \quad \quad - \E\!\left[h(Y_1,\dots,Y_{i-1},0,Y_{i+1},\dots,Y_d)\right]\!\Big) \nonumber \\
& = \left(1 - 2\delta\right) d \, \E\!\left[h(1,Y_{2},\dots,Y_d) - h(0,Y_{2},\dots,Y_d)\right] \nonumber \\
& = \left(1 - 2\delta\right) d \, \cdot \nonumber \\
& \quad \enspace \P\!\left(h(1,Y_{2},\dots,Y_d) = 1, h(0,Y_{2},\dots,Y_d) = 0\right) \label{Eq: Useful Margulis-Russo Calculation} \\
& = \left\{ 
    \begin{array}{ll}
      \displaystyle{\left(1 - 2\delta\right) \frac{d}{2} \, \P\!\left(\sum_{i = 2}^{d}{Y_{i}} = \frac{d}{2} - 1\right)} & \\
		  \displaystyle{\enspace \, \, + \left(1 - 2\delta\right) \frac{d}{2} \, \P\!\left(\sum_{i = 2}^{d}{Y_{i}} = \frac{d}{2}\right)} &, \enspace d \text{ even} \\
      \displaystyle{\left(1 - 2\delta\right) d \, \P\!\left(\sum_{i = 2}^{d}{Y_{i}} = \frac{d-1}{2}\right)} &, \enspace d \text{ odd}
    \end{array}
	\right. \nonumber \\
& = \left\{ 
    \begin{array}{ll}
      \displaystyle{\left(1 - 2\delta\right) \frac{d}{2} \binom{d-1}{\frac{d}{2}-1} p^{\frac{d}{2}-1} (1-p)^{\frac{d}{2}}} & \\
			\displaystyle{+ \left(1 - 2\delta\right) \frac{d}{2} \binom{d-1}{\frac{d}{2}} p^{\frac{d}{2}} (1-p)^{\frac{d}{2}-1} } &, \enspace d \text{ even} \\
      \displaystyle{\left(1 - 2\delta\right) d \, \binom{d-1}{\frac{d-1}{2}} p^{\frac{d-1}{2}} (1-p)^{\frac{d-1}{2}}} &, \enspace d \text{ odd}
    \end{array}
	\right. \nonumber \\
& = \left\{ 
    \begin{array}{ll}
      \displaystyle{\left(1 - 2\delta\right) \frac{d}{4} \binom{d}{\frac{d}{2}} (p(1-p))^{\frac{d}{2}-1} } &\!, \enspace d \text{ even} \\
      \displaystyle{\left(1 - 2\delta\right) \frac{d+1}{2} \binom{d}{\frac{d+1}{2}} (p(1-p))^{\frac{d-1}{2}}} &\!, \enspace d \text{ odd}
    \end{array}
	\right. \nonumber \\
& = \left\{ 
    \begin{array}{ll}
      \displaystyle{\left(1 - 2\delta\right) \frac{d}{4} \binom{d}{\frac{d}{2}} \cdot} & \\
			\displaystyle{\quad \quad \enspace ((\sigma * \delta)(1-\sigma * \delta))^{\frac{d}{2}-1}} &, \enspace d \text{ even} \\
      \displaystyle{\left(1 - 2\delta\right) \frac{d+1}{2} \binom{d}{\frac{d+1}{2}} \cdot} & \\
			\displaystyle{\quad \quad \enspace ((\sigma * \delta)(1-\sigma * \delta))^{\frac{d-1}{2}}} &, \enspace d \text{ odd}
    \end{array}
	\right.
\label{Eq:Derivative of g}
\end{align}
where the second equality follows from $dp/d\sigma = 1-2\delta$ and \eqref{Eq:Margulis-Russo formula}, the third equality holds because $h = \maj$ is symmetric in its input bits, the fourth equality holds because $h = \maj$ is non-decreasing, and the fifth equality follows from the definition of the majority rule. Since $p \mapsto p(1-p)$ is increasing on $\big[0,\frac{1}{2}\big]$ and decreasing on $\big[\frac{1}{2},1\big]$, and $p = \sigma * \delta$ is linear in $\sigma$ with derivative $1-2\delta > 0$ such that $p = \frac{1}{2}$ when $\sigma = \frac{1}{2}$, it is straightforward to verify from \eqref{Eq:Derivative of g} that $g^{\prime}$ is positive on $[0,1]$, increasing on $\big[0,\frac{1}{2}\big]$, and decreasing on $\big[\frac{1}{2},1\big]$. As a result, $g$ is increasing on $[0,1]$, convex on $\big[0,\frac{1}{2}\big]$, and concave on $\big[\frac{1}{2},1\big]$. Furthermore, the Lipschitz constant of $g$ over $[0,1]$, or equivalently, the maximum value of $g^{\prime}$ over $[0,1]$ is:
\begin{align} 
D(\delta,d) & \triangleq \max_{\sigma \in [0,1]}{g^{\prime}(\sigma)} = g^{\prime}\!\left(\frac{1}{2}\right) \\
& = (1 - 2\delta) \left(\frac{1}{2}\right)^{\! d-1} \ceil[\bigg]{\frac{d}{2}} \binom{d}{\ceil[\big]{\frac{d}{2}}}
\label{Eq:Lipschitz constant of g}
\end{align}
regardless of whether $d$ is even or odd.

There are two regimes of interest when we consider the contraction properties and fixed point structure of $g$. As defined in \eqref{Eq:Critical Noise Level}, let $\delta_{\mathsf{maj}}$ be the critical noise level such that the Lipschitz constant $g^{\prime}\big(\frac{1}{2}\big)$ is equal to $1$.\footnote{We can also view $\delta_{\mathsf{maj}}$ as the critical value such that the $d$-input majority gate with independent $\mathsf{BSC}(\delta)$'s at each input is an \textit{amplifier} if and only if $\delta < \delta_{\mathsf{maj}}$. We refer readers to \cite{ShuttyWoottersHayden2018} for more information about amplifiers, and in particular, the relationship between amplifiers and reliable computation.} Then, in the $\delta \in (0,\delta_{\mathsf{maj}})$ regime, the Lipschitz constant $g^{\prime}\big(\frac{1}{2}\big)$ is greater than $1$. Furthermore, since $g\big(\frac{1}{2}\big) = \frac{1}{2}$ and $g(1-\sigma) = 1 - g(\sigma)$ (which are straightforward to verify from \eqref{Eq:Binomial form of g}), the aforementioned properties of $g$ imply that $g$ has three fixed points at $\sigma = 1-\hat{\sigma},\frac{1}{2},\hat{\sigma}$, where the largest fixed point of $g$ is some $\hat{\sigma} \in \big(\frac{1}{2},1\big)$ that depends on $\delta$ (e.g. $\hat{\sigma} = \big(1 + \sqrt{(1-6\delta)/(1-2\delta)^3}\big)/2$ when $d = 3$). In contrast, in the $\delta \in \big(\delta_{\mathsf{maj}},\frac{1}{2}\big)$ regime, the Lipschitz constant $g^{\prime}\big(\frac{1}{2}\big)$ is less than $1$, and the only fixed point of $g$ is $\sigma = \frac{1}{2}$. (We also mention that when $\delta = \delta_{\mathsf{maj}}$, $g$ has only one fixed point at $\sigma = \frac{1}{2}$.) 

Using these observations, we now prove Theorem \ref{Thm:Phase Transition in Random Grid with Majority Rule Processing}.

\renewcommand{\proofname}{Proof of Theorem \ref{Thm:Phase Transition in Random Grid with Majority Rule Processing}}

\begin{proof}
We begin by constructing a useful ``monotone Markovian coupling'' that will help establish both achievability and converse directions (see \cite[Chapter 5]{LevinPeresWilmer2009} for basic definitions of Markovian couplings). Let $\{X^+_k : k \in \N\}$ and $\{X^-_k : k \in \N\}$ denote versions of the Markov chain $\{X_k : k \in \N\}$ (i.e. with the same transition kernels) initialized at $X^+_0 = 1$ and $X^-_0 = 0$, respectively. In particular, the marginal distributions of $X_k^+$ and $X_k^-$ are $P^+_{X_k}$ and $P^-_{X_k}$, respectively. The monotone Markovian coupling $\{(X^-_k,X^+_k) : k \in \N\}$ between the Markov chains $\{X^+_k : k \in \N\}$ and $\{X^-_k : k \in \N\}$ is generated as follows. First, condition on any random DAG realization $G = \G$. Recall that each edge $\mathsf{BSC}(\delta)$ of $\G$ either copies its input bit with probability $1 - 2\delta$, or produces an independent $\Ber\big(\frac{1}{2}\big)$ bit with probability $2\delta$ (as demonstrated in the proof of Proposition \ref{Prop: Slow Growth of Layers} in Appendix \ref{Proof of Proposition Slow Growth of Layers}). Next, couple $\{X^+_k : k \in \N\}$ and $\{X^-_k : k \in \N\}$ so that along any edge BSC of $\G$, say $(X_{k,j},X_{k+1,i})$, $X_{k,j}^+$ and $X_{k,j}^-$ are either both copied with probability $1 - 2\delta$, or a shared independent $\Ber\big(\frac{1}{2}\big)$ bit is produced with probability $2\delta$ that becomes the value of both $X_{k+1,i}^+$ and $X_{k+1,i}^-$. In other words, $\{X^+_k : k \in \N\}$ and $\{X^-_k : k \in \N\}$ ``run'' on the same underlying DAG $\G$ and have common BSCs. Hence, after averaging over all realizations of $G$, it is straightforward to verify that the Markovian coupling $\{(X^-_k,X^+_k) : k \in \N\}$ has the following properties:
\begin{enumerate}
\item The ``marginal'' Markov chains are $\{X^+_k : k \in \N\}$ and $\{X^-_k : k \in \N\}$.
\item For every $k \in \N$, $X_{k+1}^+$ is conditionally independent of $X_{k}^-$ given $X_k^{+}$, and $X_{k+1}^-$ is conditionally independent of $X_{k}^+$ given $X_k^{-}$. 
\item For every $k \in \N$ and every $j \in [L_{k}]$, $X_{k,j}^+ \geq X_{k,j}^-$ almost surely\textemdash this is the monotonicity property of the coupling.
\end{enumerate} 
In particular, the third property holds because $1 = X_{0,0}^+ \geq X_{0,0}^- = 0$ is true by assumption, each edge BSC preserves monotonicity (whether it copies its input or generates a new shared bit), and the majority processing functions are symmetric and monotone non-decreasing. In the sequel, probabilities of events that depend on the coupled vertex random variables $\{(X_{k,j}^-,X_{k,j}^+) : k \in \N, j \in [L_k]\}$ are defined with respect to this Markovian coupling. Note that this coupling also induces a monotone Markovian coupling $\{(\sigma^+_k,\sigma^-_k) : k \in \N\}$ between the Markov chains $\{\sigma^+_k : k \in \N\}$ and $\{\sigma^-_k : k \in \N\}$ (where $\{\sigma^+_k : k \in \N\}$ and $\{\sigma^-_k : k \in \N\}$ denote versions of the Markov chain $\{\sigma_k : k \in \N\}$ initialized at $\sigma^+_0 = 1$ and $\sigma^-_0 = 0$, respectively) such that:
\begin{enumerate}
\item The ``marginal'' Markov chains are $\{\sigma^+_k : k \in \N\}$ and $\{\sigma^-_k : k \in \N\}$.
\item For every $j > k \geq 1$, $\sigma_{j}^+$ is conditionally independent of $\sigma^-_0,\dots,\sigma^-_k,\sigma^+_0,\dots,\sigma^+_{k-1}$ given $\sigma^+_k$, and $\sigma_{j}^-$ is conditionally independent of $\sigma^+_0,\dots,\sigma^+_k,\sigma^-_0,\dots,\sigma^-_{k-1}$ given $\sigma^-_k$. 
\item For every $k \in \N$, $\sigma^+_k \geq \sigma^-_k$ almost surely.
\end{enumerate}

\textbf{Part 1:} We first prove that $\delta \in (0,\delta_{\mathsf{maj}})$ implies that $\limsup_{k \rightarrow \infty}{\P(\hat{S}_{k} \neq \sigma_0)} < \frac{1}{2}$. To this end, we start by showing that there exists $\epsilon = \epsilon(\delta,d) > 0$ (that depends on $\delta$ and $d$) such that for all $k \in \N\backslash\!\{0\}$:
\begin{equation}
\label{Eq: Stability whp}
\P\!\left(\left.\sigma_{k}^+ \geq \hat{\sigma} - \epsilon \, \right| \sigma_{k-1}^+ \geq \hat{\sigma} - \epsilon,A_{k,j}\right) \geq 1 - \exp\!\left(-2 L_k \gamma(\epsilon)^2 \right) 
\end{equation}
where $\gamma(\epsilon) \triangleq g(\hat{\sigma} - \epsilon) - (\hat{\sigma} - \epsilon) > 0$, and $A_{k,j}$ is the non-zero probability event defined as: 
$$ A_{k,j} \triangleq \left\{ 
\begin{array}{lcl}
\!\!\!\{\sigma_{j}^- \leq 1-\hat{\sigma} + \epsilon\} & \!\!\!\!\!, & \!\!\!\! 0 \leq j = k-1 \\
\!\!\!\{\sigma_{k-2}^+ \geq \hat{\sigma} - \epsilon,\dots,\sigma_{j}^+ \geq \hat{\sigma} - \epsilon\} & & \\
\!\cap \, \{\sigma_{j}^- \leq 1-\hat{\sigma} + \epsilon\} & \!\!\!\!\!, & \!\!\!\! 0 \leq j \leq k-2
\end{array} \right.
$$ 
for any $0 \leq j < k$. Since $g^{\prime}(\hat{\sigma}) < 1$ and $g(\hat{\sigma}) = \hat{\sigma}$, $g(\hat{\sigma} - \epsilon) > \hat{\sigma} - \epsilon$ for sufficiently small $\epsilon > 0$. Fix any such $\epsilon > 0$ (which depends on $\delta$ and $d$ because $g$ depends on $\delta$ and $d$) such that $\gamma(\epsilon) > 0$. Recall that $L_k \sigma_k \sim \mathsf{binomial}(L_k,g(\sigma))$ given $\sigma_{k-1} = \sigma$. This implies that for every $k \in \N\backslash\!\{0\}$ and every $0 \leq j < k$:
\begin{align*}
& \P\!\left(\left.\sigma_k^+ < g\!\left(\sigma_{k-1}^+\right) - \gamma(\epsilon) \, \right|\sigma_{k-1}^+ = \sigma,A_{k,j}\right) \\
& \quad \quad \quad \quad \quad \quad = \P(\sigma_k < g(\sigma_{k-1}) - \gamma(\epsilon)|\sigma_{k-1} = \sigma) \\
& \quad \quad \quad \quad \quad \quad \leq \exp\!\left(-2 L_k \gamma(\epsilon)^2 \right) 
\end{align*}
where the equality follows from property 2 of our Markovian coupling $\{(\sigma^+_k,\sigma^-_k) : k \in \N\}$, and the inequality follows from \eqref{Eq:Conditional Expectation} and Hoeffding's inequality \cite[Theorem 1]{Hoeffding1963}. As a result, we have:
\begin{align*}
& \sum_{\sigma \geq \hat{\sigma} - \epsilon} \! \Big(\P\!\left(\left.\sigma_{k-1}^+ = \sigma \, \right| A_{k,j}\right) \cdot \\
& \quad \quad \quad \, \, \P\!\left(\left.\sigma_k^+ < g\!\left(\sigma_{k-1}^+\right) - \gamma(\epsilon) \, \right|\sigma_{k-1}^+ = \sigma,A_{k,j}\right) \!\Big) \\
& \leq \exp\!\left(-2 L_k \gamma(\epsilon)^2 \right) \sum_{\sigma \geq \hat{\sigma} - \epsilon}{\P\!\left(\left.\sigma_{k-1}^+ = \sigma \, \right| A_{k,j}\right)} 
\end{align*}
or equivalently:
\begin{align*}
& \P\!\left(\left.\sigma_k^+ < g\!\left(\sigma_{k-1}^+\right) - \gamma(\epsilon), \sigma_{k-1}^+ \geq \hat{\sigma} - \epsilon \, \right|A_{k,j}\right) \\
& \quad \quad \quad \leq \exp\!\left(-2 L_k \gamma(\epsilon)^2 \right) \P\!\left(\left.\sigma_{k-1}^+ \geq \hat{\sigma} - \epsilon \, \right| A_{k,j}\right) \\
\Leftrightarrow \quad & \P\!\left(\left.\sigma_k^+ < g\!\left(\sigma_{k-1}^+\right) - \gamma(\epsilon) \, \right| \sigma_{k-1}^+ \geq \hat{\sigma} - \epsilon, A_{k,j}\right) \\
& \quad \quad \quad \leq \exp\!\left(-2 L_k \gamma(\epsilon)^2 \right) .
\end{align*}
Finally, notice that $\sigma_k^+ < \hat{\sigma} - \epsilon = g(\hat{\sigma} - \epsilon) - \gamma(\epsilon)$ implies that $\sigma_k^+ < g(\sigma_{k-1}^+) - \gamma(\epsilon)$ when $\sigma_{k-1}^+ \geq \hat{\sigma} - \epsilon$ (since $g$ is non-decreasing and $g(\sigma_{k-1}^+) \geq g(\hat{\sigma} - \epsilon)$). This produces: 
$$ \P\!\left(\left.\sigma_k^+ < \hat{\sigma} - \epsilon \, \right| \sigma_{k-1}^+ \geq \hat{\sigma} - \epsilon, A_{k,j}\right) \leq \exp\!\left(-2 L_k \gamma(\epsilon)^2 \right) $$
which in turn establishes \eqref{Eq: Stability whp}.

Now fix any $\tau > 0$, and choose a sufficiently large value $K = K(\epsilon,\tau) \in \N$ (that depends on $\epsilon$ and $\tau$) such that:
\begin{equation}
\label{Eq: Bound on Tail of Sum}
\sum_{m = K+1}^{\infty}{\exp\!\left(-2 L_m \gamma(\epsilon)^2 \right)} \leq \tau \, . 
\end{equation}
Note that such $K$ exists because $\sum_{m = 1}^{\infty}{1/m^2} = \pi^2 /6 < +\infty$, and for all sufficiently large $m$ (depending on $\delta$ and $d$), we have: 
\begin{equation}
\label{Eq:Relax Condition on L_m}
\exp\!\left(-2 L_m \gamma(\epsilon)^2\right) \leq \frac{1}{m^2} \quad \Leftrightarrow \quad L_m \geq \frac{\log(m)}{\gamma(\epsilon)^2} \, . 
\end{equation}
In \eqref{Eq:Relax Condition on L_m}, we use the assumption that $L_m \geq C(\delta,d) \log(m)$ for all sufficiently large $m$ (depending on $\delta$ and $d$), where we define the constant $C(\delta,d)$ as:
\begin{equation}
\label{Eq: Large Constant before Log}
C(\delta,d) \triangleq \frac{1}{\gamma(\epsilon(\delta,d))^2} > 0 \, .
\end{equation}
Using the continuity of probability measures, observe that:
\begin{align*}
& \P\!\left(\left. \bigcap_{k > K}{\left\{\sigma_k^+ \geq \hat{\sigma} - \epsilon\right\}} \, \right| \sigma_K^+ \geq \hat{\sigma} - \epsilon, \sigma_K^- \leq 1-\hat{\sigma} + \epsilon \right) \\
& \quad \quad \quad \quad \quad = \prod_{k > K}{\P\!\left(\left. \sigma_k^+ \geq \hat{\sigma} - \epsilon \, \right| \sigma_{k-1}^+ \geq \hat{\sigma} - \epsilon,A_{k,K}\right)} \\
& \quad \quad \quad \quad \quad \geq \prod_{k > K}{1 - \exp\!\left(-2 L_k \gamma(\epsilon)^2\right)} \\
& \quad \quad \quad \quad \quad \geq 1 - \sum_{k > K}{\exp\!\left(-2 L_k \gamma(\epsilon)^2\right)} \\
& \quad \quad \quad \quad \quad \geq 1 - \tau
\end{align*}
where the first inequality follows from \eqref{Eq: Stability whp}, the second inequality is straightforward to establish using induction, and the final inequality follows from \eqref{Eq: Bound on Tail of Sum}. Therefore, we have for any $k > K$:
\begin{equation}
\label{Eq: Stability with Extra Conditioning}
\P\!\left(\left. \sigma_k^+ \geq \hat{\sigma} - \epsilon \, \right| \sigma_K^+ \geq \hat{\sigma} - \epsilon, \sigma_K^- \leq 1-\hat{\sigma} + \epsilon \right) \geq 1 - \tau \, . 
\end{equation}
Likewise, we can also prove mutatis mutandis that for any $k > K$:
\begin{equation}
\label{Eq: Stability with Extra Conditioning 2}
\P\!\left(\left. \sigma_k^- \leq 1-\hat{\sigma} + \epsilon \, \right| \sigma_K^+ \geq \hat{\sigma} - \epsilon, \sigma_K^- \leq 1-\hat{\sigma} + \epsilon \right) \geq 1 - \tau 
\end{equation}
where the choices of $\epsilon$, $\tau$, and $K$ in \eqref{Eq: Stability with Extra Conditioning 2} are the same as those in \eqref{Eq: Stability with Extra Conditioning} without loss of generality.

We need to show that $\limsup_{k \rightarrow \infty}{\P(\hat{S}_{k} \neq \sigma_0)} < \frac{1}{2}$, or equivalently, that there exists $\lambda > 0$ such that for all sufficiently large $k \in \N$:
\begin{align*}
\frac{1 - \lambda}{2} & \geq \P\!\left(\hat{S}_k \neq \sigma_0\right) = \frac{1}{2} \,\P\!\left(\left.\hat{S}_k \neq \sigma_0 \, \right| \sigma_0 = 1\right) \\
& \quad \quad \quad \quad \quad \quad \quad \enspace \,\, + \frac{1}{2} \, \P\!\left(\left.\hat{S}_k \neq \sigma_0 \, \right| \sigma_0 = 0\right) \\
\Leftrightarrow \quad 1 - \lambda & \geq \P\!\left(\left.\sigma_k < \frac{1}{2} \, \right| \sigma_0 = 1\right) + \P\!\left(\left.\sigma_k \geq \frac{1}{2} \, \right| \sigma_0 = 0\right) \\
\Leftrightarrow \quad \lambda & \leq \P\!\left(\sigma_k^+ \geq \frac{1}{2}\right) - \P\!\left(\sigma_k^- \geq \frac{1}{2}\right) .
\end{align*}
To this end, let $E = \big\{\sigma_K^+ \geq \hat{\sigma} - \epsilon, \, \sigma_K^- \leq 1-\hat{\sigma} + \epsilon\big\}$, and observe that for all $k > K$:
\begin{align*}
& \P\!\left(\sigma_k^+ \geq \frac{1}{2}\right) - \P\!\left(\sigma_k^- \geq \frac{1}{2}\right) \\
& = \E\!\left[\I\!\left\{\sigma_k^+ \geq \frac{1}{2}\right\} - \I\!\left\{\sigma_k^- \geq \frac{1}{2}\right\} \right] \\
& \geq \E\!\left[\left(\I\!\left\{\sigma_k^+ \geq \frac{1}{2}\right\} - \I\!\left\{\sigma_k^- \geq \frac{1}{2}\right\}\right)\I\!\left\{E\right\}\right] \\
& = \E\!\left[\left. \I\!\left\{\sigma_k^+ \geq \frac{1}{2}\right\} - \I\!\left\{\sigma_k^- \geq \frac{1}{2}\right\} \right| E \right] \P(E) \\
& = \left(\P\!\left(\left. \sigma_k^+ \geq \frac{1}{2} \, \right| E \right) - \P\!\left(\left. \sigma_k^- \geq \frac{1}{2} \, \right| E \right)\right) \P(E) \\
& \geq \left(\P\!\left(\left. \sigma_k^+ \geq \hat{\sigma} - \epsilon \, \right| E \right) - \P\!\left(\left. \sigma_k^- > 1-\hat{\sigma} + \epsilon \, \right| E \right)\right) \P(E) \\
& \geq (1 - 2\tau) \P(E) \triangleq \lambda > 0
\end{align*}
where the first inequality holds because $\I\big\{\sigma_{k}^+ \geq \frac{1}{2}\big\} - \I\big\{\sigma_{k}^- \geq \frac{1}{2}\big\} \geq 0$ almost surely due to the monotonicity (property 3) of the Markovian coupling $\{(\sigma^+_k,\sigma^-_k) : k \in \N\}$, the second inequality holds because $1 - \hat{\sigma} + \epsilon < \frac{1}{2} < \hat{\sigma} - \epsilon$ (since $\epsilon > 0$ is small), and the final inequality follows from \eqref{Eq: Stability with Extra Conditioning} and \eqref{Eq: Stability with Extra Conditioning 2}. This completes the proof for the $\delta \in (0,\delta_{\mathsf{maj}})$ regime. 

\textbf{Part 2:} We next prove that $\delta \in \big(\delta_{\mathsf{maj}},\frac{1}{2}\big)$ implies \eqref{Eq:Strong TV Reconstruction Impossible}. First, notice that conditioned on any realization of the random DAG $G$, we have $X^+_{k,j} \geq X^-_{k,j}$ almost surely for every $k \in \N$ and $j \in [L_k]$ (by construction of our coupling). Hence, conditioned on $G$, we obtain:
\begin{align*}
\left\|P_{X_k|G}^+ - P_{X_k|G}^-\right\|_{\mathsf{TV}} & \leq \P\!\left(X^+_k \neq X^-_k \middle| G\right) \\
& = \P\!\left(\exists j \in [L_k], \, X^+_{k,j} \neq X^-_{k,j} \middle| G\right) \\
& \leq \sum_{j = 0}^{L_k - 1}{\P\!\left(X^+_{k,j} \neq X^-_{k,j} \middle| G\right)} \\
& = \E\!\left[\sum_{j = 0}^{L_k - 1}{X^+_{k,j} - X^-_{k,j}}\middle| G\right] \\
& = L_k \, \E\!\left[\sigma^+_{k} - \sigma^-_{k} \middle| G\right] 
\end{align*}
where the first inequality follows from Dobrushin's maximal coupling representation of TV distance \cite[Chapter 4.2]{LevinPeresWilmer2009}, the third inequality follows from the union bound, and the fourth equality holds because $\P\big(X^+_{k,j} \neq X^-_{k,j} \big| G\big) = \P\big(X^+_{k,j} - X^-_{k,j} = 1 \big| G\big) = \E\big[X^+_{k,j} - X^-_{k,j} \big| G\big]$ due to the monotonicity of our coupling. Then, taking expectations with respect to $G$ yields:
\begin{equation} 
\label{Eq: Initial TV Bound}
\E\!\left[\left\|P_{X_k|G}^+ - P_{X_k|G}^-\right\|_{\mathsf{TV}}\right] \leq L_k \, \E\!\left[\sigma^+_k - \sigma^-_k \right] .
\end{equation}

We can bound $\E\big[\sigma^+_k - \sigma^-_k\big]$ as follows. Firstly, we use the Lipschitz continuity of $g$ (with Lipschitz constant $D(\delta,d)$) and the monotonicity of our coupling to get:
\begin{align*}
0 \leq \E\!\left[\left.\sigma^+_k - \sigma^-_k\right|\sigma^+_{k-1},\sigma^-_{k-1}\right] & = g\!\left(\sigma^+_{k-1}\right) - g\!\left(\sigma^-_{k-1}\right) \\
& \leq D(\delta,d) \left(\sigma^+_{k-1} - \sigma^-_{k-1}\right) . 
\end{align*}
Then, we can take expectations with respect to $\big(\sigma^+_{k-1},\sigma^-_{k-1}\big)$ on both sides of this inequality (and use the tower property on the left hand side) to obtain:
$$ 0 \leq \E\!\left[\sigma^+_k - \sigma^-_k\right] \leq D(\delta,d) \, \E\!\left[\sigma^+_{k-1} - \sigma^-_{k-1}\right]  . $$
Therefore, we recursively have:
$$ 0 \leq \E\!\left[\sigma^+_k - \sigma^-_k\right] \leq D(\delta,d)^{k} $$
where we use the fact that $\E\big[\sigma^+_0 - \sigma^-_0\big] = 1$. Using \eqref{Eq: Initial TV Bound} with this bound, we get:
$$ \E\!\left[\left\|P_{X_k|G}^+ - P_{X_k|G}^-\right\|_{\mathsf{TV}}\right] \leq L_k D(\delta,d)^{k} $$
where letting $k \rightarrow \infty$ yields:
\begin{equation}
\label{Eq: Auxiliary Convergence Condition}
\lim_{k \rightarrow \infty}{\E\!\left[\left\|P_{X_k|G}^+ - P_{X_k|G}^-\right\|_{\mathsf{TV}}\right]} = 0 
\end{equation}
because $L_k = o(D(\delta,d)^{-k})$ by assumption. (It is worth mentioning that although $L_k = o(D(\delta,d)^{-k})$ in this regime, it can diverge to infinity because the Lipschitz constant $D(\delta,d) < 1$.) 

Finally, observe that $\big\|P_{X_k|G}^+ - P_{X_k|G}^-\big\|_{\mathsf{TV}} \in [0,1]$ forms a non-increasing sequence in $k$ for every realization of the random DAG $G$ (since the sequence $\{X_k : k \in \N\}$ forms a Markov chain given $G$, and the data processing inequality for TV distance yields the desired monotonicity). Hence, the pointwise limit (over realizations of $G$) random variable, $\lim_{k \rightarrow \infty}{\big\|P_{X_k|G}^+ - P_{X_k|G}^-\big\|_{\mathsf{TV}}} \in [0,1]$, has mean $\E\big[\lim_{k \rightarrow \infty}{\big\|P_{X_k|G}^+ - P_{X_k|G}^-\big\|_{\mathsf{TV}}}\big] = 0$ due to \eqref{Eq: Auxiliary Convergence Condition} and the bounded convergence theorem. Since a non-negative random variable that has zero mean must be equal to zero almost surely, we have \eqref{Eq:Strong TV Reconstruction Impossible}:
$$ \lim_{k \rightarrow \infty}{\left\|P_{X_k|G}^+ - P_{X_k|G}^-\right\|_{\mathsf{TV}}} = 0 \quad G\text{-}a.s. $$
This completes the proof.
\end{proof}

\renewcommand{\proofname}{Proof}

Finally, the next proposition portrays that the Markov chain $\{\sigma_k : k \in \N\}$ converges almost surely when $\delta \in \big(\delta_{\mathsf{maj}},\frac{1}{2}\big)$, $L_k = \omega(\log(k))$, and all processing functions are majority.

\begin{proposition}[Majority Random DAG Model Almost Sure Convergence]
\label{Prop: Majority Grid Almost Sure Convergence}
If $\delta \in \big(\delta_{\mathsf{maj}},\frac{1}{2}\big)$ and $L_k = \omega(\log(k))$, then $\lim_{k \rightarrow \infty}{\sigma_k} = \frac{1}{2}$ almost surely.
\end{proposition}

Proposition \ref{Prop: Majority Grid Almost Sure Convergence} is proved in Appendix \ref{Proof of Proposition Majority Grid Almost Sure Convergence}. It can be construed as a ``weak'' impossibility result since it demonstrates that the average number of $1$'s tends to $\frac{1}{2}$ in the $\delta \in \big(\delta_{\mathsf{maj}},\frac{1}{2}\big)$ regime regardless of the initial state of the Markov chain $\{\sigma_k : k \in \N\}$.

\section{Analysis of AND-OR Rule Processing in Random DAG Model}
\label{Analysis of And-Or Rule Processing in Random Grid}

In this section, we prove Theorem \ref{Thm:Phase Transition in Random Grid with And-Or Rule Processing}. As before, we begin by making some pertinent observations. Recall that we have a random DAG model with $d = 2$, and all Boolean functions at even levels are the AND rule, and all Boolean functions at odd levels are the OR rule, i.e. $f_{k}(x_1,x_2) = x_1 \land x_2$ for every $k \in 2\N\backslash\!\{0\}$, and $f_{k}(x_1,x_2) = x_1 \lor x_2$ for every $k \in \N\backslash 2\N$. Suppose we are given that $\sigma_{k-1} = \sigma$ for any $k \in \N\backslash\!\{0\}$. Then, for every $j \in [L_{k}]$:
\begin{equation}
X_{k,j} = \left\{ 
\begin{array}{lcl}
\Ber(\sigma * \delta) \land \Ber(\sigma * \delta) & , & k \text{ even} \\ 
\Ber(\sigma * \delta) \lor \Ber(\sigma * \delta) & , & k \text{ odd}
\end{array} 
\right.
\end{equation}
for two i.i.d. Bernoulli random variables. Since we have:
\begin{align}
\P(X_{k,j} = 1|\sigma_{k-1} = \sigma) & = \left\{
\begin{array}{lcl}
(\sigma * \delta)^2 & , & k \text{ even} \\
1 - (1 - \sigma * \delta)^2 & , & k \text{ odd}
\end{array}
\right. \\
& = \E[\sigma_k|\sigma_{k-1} = \sigma] \, ,
\end{align}
$X_{k,j}$ are i.i.d. $\Ber(g_{k \, (\mathsf{mod} \, 2)}(\sigma))$ for $j \in [L_{k}]$, and $L_k \sigma_k \sim \mathsf{binomial}(L_k,g_{k \, (\mathsf{mod} \, 2)}(\sigma))$, where we define $g_0:[0,1] \rightarrow [0,1]$ as $g_0(\sigma) \triangleq (\sigma * \delta)^2$, and $g_1:[0,1] \rightarrow [0,1]$ as $g_1(\sigma) \triangleq 1 - (1 - \sigma * \delta)^2 = 2(\sigma * \delta) - (\sigma * \delta)^2$. The derivatives of $g_0$ and $g_1$ are:
\begin{align}
\label{Eq: g_0 Derivative}
g_0^{\prime}(\sigma) & = 2(1-2\delta)(\sigma * \delta) \geq 0 \, , \\
g_1^{\prime}(\sigma) & = 2(1 - 2\delta)(1-\sigma * \delta) \geq 0 \, . 
\label{Eq: g_1 Derivative}
\end{align}
Consider the composition of $g_0$ and $g_1$, denoted $g \triangleq g_0 \circ g_1:[0,1] \rightarrow [0,1]$, $g(\sigma) = \big(\big(2(\sigma * \delta) - (\sigma * \delta)^2\big) * \delta\big)^{2}$, which has derivative $g^{\prime}:[0,1] \rightarrow \R_{+}$ given by:
\begin{align}
 g^{\prime}(\sigma) & = g_0^{\prime}(g_1(\sigma)) g_1^{\prime}(\sigma) \nonumber \\
& = 4(1-2\delta)^2 (g_1(\sigma) * \delta) (1-\sigma * \delta) \geq 0 \, .
\end{align}
This is a cubic function of $\sigma$ with maximum value:
\begin{align}
D(\delta) & \triangleq \max_{\sigma \in [0,1]}{g^{\prime}(\sigma)} \\
& = \left\{
\begin{array}{lcl}
\!\!\! g^{\prime}\!\left(\frac{1-\delta}{1-2\delta} - \sqrt{\frac{1-\delta}{3(1-2\delta)^3}}\right) & , & \delta \in \left(0,\frac{9-\sqrt{33}}{12}\right] \\
\!\!\! g^{\prime}(0) & , & \delta \in \left(\frac{9-\sqrt{33}}{12},\frac{1}{2}\right)
\end{array} \right. \\
& = \left\{
\begin{array}{lcl}
\!\!\! \left(\frac{4(1-\delta)(1-2\delta)}{3}\right)^{\frac{3}{2}} & \!\!\!\!\!, & \!\!\!\!\delta \!\in\! \left(0,\frac{9-\sqrt{33}}{12}\right] \\
\!\!\! 4\delta(1-\delta)^2 (1-2\delta)^2 (3-2\delta) < 1 & \!\!\!\!\!, & \!\!\!\!\delta \!\in\! \left(\frac{9-\sqrt{33}}{12},\frac{1}{2}\right)
\end{array} \right.
\label{Eq: Lipschitz constant}
\end{align}
which follows from standard calculus and algebraic manipulations, and Wolfram Mathematica computations. Hence, $D(\delta)$ in \eqref{Eq: Lipschitz constant} is the Lipschitz constant of $g$ over $[0,1]$. Since $4(1-\delta)(1-2\delta)/3 \in (0,1) \Leftrightarrow \delta \in ((3-\sqrt{7})/4,(9 - \sqrt{33})/12]$, $D(\delta) < 1$ if and only if $\delta \in ((3-\sqrt{7})/4,1/2)$. Moreover, $D(\delta) > 1$ if and only if $\delta \in (0,(3-\sqrt{7})/4)$ (and $D(\delta) = 1$ when $\delta = (3-\sqrt{7})/4$).

We next summarize the fixed point structure of $g$. Solving the equation $g(\sigma) = \sigma$ in Wolfram Mathematica produces:
\begin{align}
\sigma & = \frac{1-6\delta + 4\delta^2 \pm \sqrt{1-12\delta + 8\delta^2}}{2(1-2\delta)^2} \\
\text{or} \quad \sigma & = \frac{3-6\delta + 4\delta^2 \pm \sqrt{5-12\delta + 8\delta^2}}{2(1-2\delta)^2}
\end{align}
where the first pair is real when $\delta \in [0,(3-\sqrt{7})/4]$, and the second pair is always real. From these solutions, it is straightforward to verify that the only fixed points of $g$ in the interval $[0,1]$ are:
\begin{align}
t_0 & \triangleq \frac{2(1-\delta)(1-2\delta) - 1 - \sqrt{4(1-\delta)(1-2\delta) - 3}}{2(1-2\delta)^2} \\
t_1 & \triangleq \frac{2(1-\delta)(1-2\delta) - 1 + \sqrt{4(1-\delta)(1-2\delta) - 3}}{2(1-2\delta)^2} \\
t & \triangleq \frac{2(1-\delta)(1-2\delta) + 1 - \sqrt{4(1-\delta)(1-2\delta) + 1}}{2(1-2\delta)^2} 
\label{Eq: Middle Fixed Point}
\end{align}
where $t_0$ and $t_1$ are valid when $\delta \in [0,(3-\sqrt{7})/4]$, the fixed points satisfy $t_0 = t_1 = t$ when $\delta = (3-\sqrt{7})/4$, and $t_0 = 0, \, t_1 = 1$ when $\delta = 0$. Furthermore, observe that for $\delta \in (0,(3-\sqrt{7})/4)$:
\begin{align}
t_1 - t & = \frac{\sqrt{a} + \sqrt{a + 4} - 2}{2(1-2\delta)^2} > 0 \\
\text{and} \quad t - t_0 & = \frac{\sqrt{a} - \sqrt{a+4} + 2}{2(1-2\delta)^2} > 0 
\end{align}
where $a = 4(1-\delta)(1-2\delta) - 3 > 0$, $t_1 - t > 0$ because $x \mapsto \sqrt{x}$ is strictly increasing ($\Rightarrow \sqrt{a} + \sqrt{a + 4} > 2$), and $t - t_0 > 0$ because $x \mapsto \sqrt{x}$ is strictly subadditive ($\Rightarrow \sqrt{a} + 2 > \sqrt{a + 4}$). Hence, $0 < t_0 < t < t_1 < 1$ when $\delta \in (0,(3-\sqrt{7})/4)$.

Therefore, there are again two regimes of interest. Define the critical threshold $\delta_{\mathsf{andor}}$ as in \eqref{Eq: NAND threshold}. In the regime $\delta \in (0,\delta_{\mathsf{andor}})$, $g$ has three fixed points $0 < t_0 < t < t_1 < 1$, and $D(\delta) > 1$. In contrast, in the regime $\delta \in \big(\delta_{\mathsf{andor}},\frac{1}{2}\big)$, $g$ has only one fixed point at $t \in (0,1)$, and $D(\delta) < 1$. 

We now prove Theorem \ref{Thm:Phase Transition in Random Grid with And-Or Rule Processing}. (The proof closely resembles the proof of Theorem \ref{Thm:Phase Transition in Random Grid with Majority Rule Processing} in section \ref{Analysis of Majority Rule Processing in Random Grid}.)

\renewcommand{\proofname}{Proof of Theorem \ref{Thm:Phase Transition in Random Grid with And-Or Rule Processing}}

\begin{proof}
As in the proof of Theorem \ref{Thm:Phase Transition in Random Grid with Majority Rule Processing}, we begin by constructing a monotone Markovian coupling $\{(X^-_k,X^+_k) : k \in \N\}$ between the Markov chains $\{X^+_k : k \in \N\}$ and $\{X^-_k : k \in \N\}$ (which are versions of the Markov chain $\{X_k : k \in \N\}$ initialized at $X^+_0 = 1$ and $X^-_0 = 0$, respectively), and this coupling induces a monotone Markovian coupling $\{(\sigma^+_k,\sigma^-_k) : k \in \N\}$ between the Markov chains $\{\sigma^+_k : k \in \N\}$ and $\{\sigma^-_k : k \in \N\}$ (which are versions of the Markov chain $\{\sigma_k : k \in \N\}$ initialized at $\sigma^+_0 = 1$ and $\sigma^-_0 = 0$, respectively). This monotone Markovian coupling satisfies the following properties:
\begin{enumerate}
\item The ``marginal'' Markov chains are $\{X^+_k : k \in \N\}$ and $\{X^-_k : k \in \N\}$.
\item For every $k \in \N$, $X_{k+1}^+$ is conditionally independent of $X_{k}^-$ given $X_k^{+}$, and $X_{k+1}^-$ is conditionally independent of $X_{k}^+$ given $X_k^{-}$.
\item For every $j > k \geq 1$, $\sigma_{j}^+$ is conditionally independent of $\sigma^-_0,\dots,\sigma^-_{k},\sigma^+_0,\dots,\sigma^+_{k-1}$ given $\sigma^+_{k}$, and $\sigma_{j}^-$ is conditionally independent of $\sigma^+_0,\dots,\sigma^+_{k},\sigma^-_0,\dots,\sigma^-_{k-1}$ given $\sigma^-_{k}$.  
\item For every $k \in \N$ and every $j \in [L_{k}]$, $X_{k,j}^+ \geq X_{k,j}^-$ almost surely. 
\item Due to the previous property, $\sigma^+_{k} \geq \sigma^-_{k}$ almost surely for every $k \in \N$.
\end{enumerate} 
As before, the fourth property above holds because $1 = X_{0,0}^+ \geq X_{0,0}^- = 0$ is true by assumption, each edge BSC preserves monotonicity, and the AND and OR processing functions are symmetric and monotone non-decreasing. 

\textbf{Part 1:} We first prove that $\delta \in (0,\delta_{\mathsf{andor}})$ implies that $\limsup_{k \rightarrow \infty}{\P(\hat{T}_{2k} \neq \sigma_0)} < \frac{1}{2}$. To this end, we start by establishing that there exists $\epsilon = \epsilon(\delta) > 0$ (that depends on $\delta$) such that for all $k \in \N\backslash\!\{0\}$:
\begin{equation}
\label{Eq: Stability whp 2}
\begin{aligned}
& \P\!\left(\left.\sigma_{2k}^+ \geq t_1 - \epsilon \, \right| \sigma_{2k-2}^+ \geq t_1 - \epsilon,A_{k,j}\right) \\
& \quad \quad \quad \quad \quad \geq 1 - 4 \exp\!\left(-\frac{(L_{2k} \wedge L_{2k-1}) \gamma(\epsilon)^2}{8}\right)
\end{aligned}
\end{equation}
where $\wedge$ denotes the minimum operation (not to be confused with the AND operation), $\gamma(\epsilon) \triangleq g(t_1 - \epsilon) - (t_1 - \epsilon) > 0$, and $A_{k,j}$ is the non-zero probability event defined as:
$$ A_{k,j} \!\triangleq\! \left\{ 
\begin{array}{lcl}
\!\!\!\{\sigma_{2j}^- \leq t_0 + \epsilon\} & \!\!\!\!\!\!\!\!, & \!\!\!\!\!\!\! 0 \leq j = k-1 \\
\!\!\!\{\sigma_{2k-4}^+ \geq t_1 - \epsilon,\sigma_{2k-6}^+ \geq t_1 - \epsilon, & & \\
\!\!\dots,\sigma_{2j}^+ \geq t_1 - \epsilon\} \!\cap\! \{\sigma_{2j}^- \leq t_0 + \epsilon\} & \!\!\!\!\!\!\!\!, & \!\!\!\!\!\!\! 0 \leq j \leq k-2
\end{array} \right.
$$ 
for any $0 \leq j < k$. Since $g^{\prime}(t_1) = 4\delta(3-2\delta) < 1$ and $g(t_1) = t_1$, $g(t_1 - \epsilon) > t_1 - \epsilon$ for sufficiently small $\epsilon > 0$. Fix any such $\epsilon > 0$ (which depends on $\delta$ because $g$ depends on $\delta$) such that $\gamma(\epsilon) > 0$. Observe that for every $k \in \N\backslash\!\{0\}$ and $\xi > 0$, we have:
\begin{align}
& \P(|\sigma_{2k} - g(\sigma_{2k-2})| > \xi \, | \, \sigma_{2k-2} = \sigma) \nonumber \\
& \leq \P\Big(|\sigma_{2k} - g_0(\sigma_{2k-1})| \nonumber \\
& \quad \quad \, \, + |g_0(\sigma_{2k-1}) - g_0(g_1(\sigma_{2k-2}))| > \xi \, \Big| \, \sigma_{2k-2} = \sigma\Big) \nonumber \\
& \leq \P\Big(|\sigma_{2k} - g_0(\sigma_{2k-1})| \nonumber \\
& \quad \quad \, \, + g_0^{\prime}(1) |\sigma_{2k-1} - g_1(\sigma_{2k-2})| > \xi \, \Big| \, \sigma_{2k-2} = \sigma\Big) \nonumber \\
& \leq \P\bigg(\!\left\{|\sigma_{2k} - g_0(\sigma_{2k-1})| > \frac{\xi}{2}\right\} \nonumber \\
& \quad \quad \enspace \cup \left\{g_0^{\prime}(1)|\sigma_{2k-1} - g_1(\sigma_{2k-2})| > \frac{\xi}{2}\right\} \bigg| \, \sigma_{2k-2} = \sigma\bigg) \nonumber \\
& \leq \P\!\left(|\sigma_{2k} - g_0(\sigma_{2k-1})| > \frac{\xi}{2} \, \middle| \, \sigma_{2k-2} = \sigma\right) \nonumber \\
& \quad \, + \P\!\left(|\sigma_{2k-1} - g_1(\sigma_{2k-2})| > \frac{\xi}{2 g_0^{\prime}(1)} \, \middle| \, \sigma_{2k-2} = \sigma\right) \nonumber \\
& \leq \E\!\left[\P\!\left(|\sigma_{2k} - g_0(\sigma_{2k-1})| > \frac{\xi}{2} \, \middle| \, \sigma_{2k-1}\right) \middle| \, \sigma_{2k-2} = \sigma\right] \nonumber \\
& \quad \, + 2 \exp\!\left(-\frac{L_{2k-1} \xi^2}{8(1-\delta)^2 (1-2\delta)^2} \right) \nonumber \\
& \leq 2 \exp\!\left(-\frac{L_{2k} \xi^2}{2}\right) + 2 \exp\!\left(-\frac{L_{2k-1} \xi^2}{8(1-\delta)^2 (1-2\delta)^2} \right) \nonumber \\
& \leq 4 \exp\!\left(-\frac{(L_{2k} \wedge L_{2k-1}) \xi^2}{8}\right)
\label{Eq: Hoeffding Consequence 2}
\end{align} 
where the first inequality follows from the triangle inequality and the fact that $g = g_0 \circ g_1$, the second inequality holds because the Lipschitz constant of $g_0$ on $[0,1]$ is $\max_{\sigma \in [0,1]}{g_0^{\prime}(\sigma)} = g_0^{\prime}(1) = 2(1-\delta)(1-2\delta)$ using \eqref{Eq: g_0 Derivative}, the fourth inequality follows from the union bound, the fifth and sixth inequalities follow from the Markov property and Hoeffding's inequality (as well as the fact that $L_k \sigma_k \sim \mathsf{binomial}(L_k,g_{k \, (\mathsf{mod} \, 2)}(\sigma))$ given $\sigma_{k-1} = \sigma$), and the final inequality holds because $(1-\delta)^2 (1-2\delta)^2 \leq 1$. Hence, for any $k \in \N\backslash\!\{0\}$ and any $0 \leq j < k$, we have:
\begin{align*}
& \P\!\left(\sigma_{2k}^+ < g\!\left(\sigma_{2k-2}^+\right) - \gamma(\epsilon)\, \middle| \,\sigma_{2k-2}^+ = \sigma,A_{k,j}\right) \\
& \quad \quad \quad \quad = \P\!\left(\sigma_{2k} < g\!\left(\sigma_{2k-2}\right) - \gamma(\epsilon)\, \middle| \,\sigma_{2k-2} = \sigma\right) \\
& \quad \quad \quad \quad \leq \P\!\left(|\sigma_{2k} - g(\sigma_{2k-2})| > \gamma(\epsilon) \, \middle| \, \sigma_{2k-2} = \sigma\right) \\
& \quad \quad \quad \quad \leq 4 \exp\!\left(-\frac{(L_{2k} \wedge L_{2k-1}) \gamma(\epsilon)^2}{8}\right)
\end{align*} 
where the first equality follows from property 3 of the Markovian coupling, and the final inequality follows from \eqref{Eq: Hoeffding Consequence 2}. As shown in the proof of Theorem \ref{Thm:Phase Transition in Random Grid with Majority Rule Processing}, this produces:
\begin{align*}
& \P\!\left(\sigma_{2k}^+ < g\!\left(\sigma_{2k-2}^+\right) - \gamma(\epsilon) \, \middle| \, \sigma_{2k-2}^+ \geq t_1 - \epsilon, A_{k,j}\right) \\
& \quad \quad \quad \quad \quad \quad \quad \leq 4 \exp\!\left(-\frac{(L_{2k} \wedge L_{2k-1}) \gamma(\epsilon)^2}{8}\right) 
\end{align*}
which implies that:
\begin{align*}
& \P\!\left(\sigma_{2k}^+ < t_1 - \epsilon \, \middle| \, \sigma_{2k-2}^+ \geq t_1 - \epsilon, A_{k,j}\right) \\
& \quad \quad \quad \quad \quad \quad \quad \leq 4 \exp\!\left(-\frac{(L_{2k} \wedge L_{2k-1}) \gamma(\epsilon)^2}{8}\right)
\end{align*}
because $\sigma_{2k}^+ < t_1 - \epsilon = g(t_1 - \epsilon) - \gamma(\epsilon)$ implies that $\sigma_{2k}^+ < g(\sigma_{2k-2}^+) - \gamma(\epsilon)$ when $\sigma_{2k-2}^+ \geq t_1 - \epsilon$ (since $g$ is non-decreasing and $g(\sigma_{2k-2}^+) \geq g(t_1 - \epsilon)$). This proves \eqref{Eq: Stability whp 2}.

Now fix any $\tau > 0$, and choose a sufficiently large even integer $K = K(\epsilon,\tau) \in 2\N$ (that depends on $\epsilon$ and $\tau$) such that:
\begin{equation}
\label{Eq: Bound on Tail of Sum 2}
4 \sum_{m = \frac{K}{2}+1}^{\infty}{\exp\!\left(-\frac{(L_{2m} \wedge L_{2m-1}) \gamma(\epsilon)^2}{8}\right)} \leq \tau \, . 
\end{equation}
Note that such $K$ exists because $\sum_{m = 1}^{\infty}{1/(2m-1)^2} \leq 1 + \sum_{m = 2}^{\infty}{1/(2m-2)^2} = 1 + (\pi^2 /24) < +\infty$, and for sufficiently large $m$ (depending on $\delta$), we have: 
\begin{align}
\exp\!\left(-\frac{(L_{2m} \wedge L_{2m-1}) \gamma(\epsilon)^2}{8}\right) & \leq \frac{1}{(2m-1)^2} \nonumber \\
\Leftrightarrow \quad L_{2m} \wedge L_{2m-1} & \geq \frac{16 \log(2m-1)}{\gamma(\epsilon)^2} \, .
\label{Eq:Relax Condition on L_m 2}
\end{align}
As before, in \eqref{Eq:Relax Condition on L_m 2}, we utilize the assumption that $L_m \geq C(\delta) \log(m)$ for all sufficiently large $m$ (depending on $\delta$), where we define the constant $C(\delta)$ as:
\begin{equation}
\label{Eq: Large Constant before Log 2}
C(\delta) \triangleq \frac{16}{\gamma(\epsilon(\delta))^2} > 0 \, .
\end{equation}
Using the continuity of probability measures, observe that:
\begin{align*}
& \P\!\left(\left. \bigcap_{k > \frac{K}{2}}{\left\{\sigma_{2k}^+ \geq t_1 - \epsilon\right\}} \, \right| \sigma_{K}^+ \geq t_1 - \epsilon, \sigma_{K}^- \leq t_0 + \epsilon \right) \\
& \quad \quad \quad \quad = \prod_{k > \frac{K}{2}}{\P\!\left(\sigma_{2k}^+ \geq t_1 - \epsilon  \left| \, \sigma_{2k-2}^+ \geq t_1 - \epsilon,A_{k,\frac{K}{2}} \right.\right)} \\
& \quad \quad \quad \quad \geq \prod_{k > \frac{K}{2}}{1 - 4 \exp\!\left(- \frac{(L_{2k} \wedge L_{2k-1}) \gamma(\epsilon)^2}{8}\right)} \\
& \quad \quad \quad \quad \geq 1 - 4 \sum_{k > \frac{K}{2}}{\exp\!\left(- \frac{(L_{2k} \wedge L_{2k-1}) \gamma(\epsilon)^2}{8}\right)} \\
& \quad \quad \quad \quad \geq 1 - \tau
\end{align*}
where the first inequality follows from \eqref{Eq: Stability whp 2}, and the final inequality follows from \eqref{Eq: Bound on Tail of Sum 2}. Therefore, we have for any $k > \frac{K}{2}$:
\begin{equation}
\label{Eq: And-Or Stability with Extra Conditioning}
\P\!\left(\left. \sigma_{2k}^+ \geq t_1 - \epsilon \, \right| \sigma_K^+ \geq t_1 - \epsilon, \sigma_K^- \leq t_0 + \epsilon \right) \geq 1 - \tau \, . 
\end{equation}
Likewise, we can also prove mutatis mutandis that for any $k > \frac{K}{2}$:
\begin{equation}
\label{Eq: And-Or Stability with Extra Conditioning 2}
\P\!\left(\left. \sigma_{2k}^- \leq t_0 + \epsilon \, \right| \sigma_K^+ \geq t_1 - \epsilon, \sigma_K^- \leq t_0 + \epsilon \right) \geq 1 - \tau 
\end{equation}
where $\epsilon$, $\tau$, and $K$ in \eqref{Eq: And-Or Stability with Extra Conditioning 2} can be chosen to be the same as those in \eqref{Eq: And-Or Stability with Extra Conditioning} without loss of generality.

Finally, we let $E = \big\{\sigma_K^+ \geq t_1 - \epsilon, \, \sigma_K^- \leq t_0 + \epsilon\big\}$, and observe that for all $k > \frac{K}{2}$:
\begin{align*}
& \P\!\left(\sigma_{2k}^+ \geq t\right) - \P\!\left(\sigma_{2k}^- \geq t\right) \\
& \geq \E\!\left[\left(\I\!\left\{\sigma_{2k}^+ \geq t\right\} - \I\!\left\{\sigma_{2k}^- \geq t\right\}\right)\I\!\left\{E\right\}\right] \\
& = \left(\P\!\left(\left. \sigma_{2k}^+ \geq t \, \right| E \right) - \P\!\left(\left. \sigma_{2k}^- \geq t \, \right| E \right)\right) \P(E) \\
& \geq \left(\P\!\left(\left. \sigma_{2k}^+ \geq t_1 - \epsilon \, \right| E \right) - \P\!\left(\left. \sigma_{2k}^- > t_0 + \epsilon \, \right| E \right)\right) \P(E) \\
& \geq (1 - 2\tau) \P(E) > 0
\end{align*}
where the first inequality holds because $\I\big\{\sigma_{2k}^+ \geq t\big\} - \I\big\{\sigma_{2k}^- \geq t\big\} \geq 0$ almost surely due to the monotonicity (property 5) of our Markovian coupling, the second inequality holds because $t_0 + \epsilon < t < t_1 - \epsilon$ (since $\epsilon > 0$ is small), and the final inequality follows from \eqref{Eq: And-Or Stability with Extra Conditioning} and \eqref{Eq: And-Or Stability with Extra Conditioning 2}. As argued in the proof of Theorem \ref{Thm:Phase Transition in Random Grid with Majority Rule Processing}, this illustrates that $\limsup_{k \rightarrow \infty}{\P(\hat{T}_{2k} \neq \sigma_0)} < \frac{1}{2}$. 

\textbf{Part 2:} We next prove that $\delta \in \big(\delta_{\mathsf{andor}},\frac{1}{2}\big)$ implies \eqref{Eq:Strong TV Reconstruction Impossible}. Following the proof of Theorem \ref{Thm:Phase Transition in Random Grid with Majority Rule Processing}, we can show that:
\begin{equation} 
\label{Eq: Initial TV Bound 2}
\E\!\left[\left\|P_{X_{2k}|G}^+ - P_{X_{2k}|G}^-\right\|_{\mathsf{TV}}\right] \leq L_{2k} \, \E\!\left[\sigma^+_{2k} - \sigma^-_{2k}\right] .
\end{equation}
In order to bound $\E\big[\sigma^+_{2k} - \sigma^-_{2k}\big]$, we proceed as follows. Firstly, for any $k \in \N\backslash\!\{0\}$, we have:
\begin{align}
& \E\!\left[\left.\sigma^+_{2k} - \sigma^-_{2k}\right|\sigma^+_{2k-2},\sigma^-_{2k-2}\right] \nonumber \\
& \quad \quad = \E\!\left[\left.\E\!\left[\left.\sigma^+_{2k} - \sigma^-_{2k}\right|\sigma^+_{2k-1},\sigma^-_{2k-1}\right]\right|\sigma^+_{2k-2},\sigma^-_{2k-2}\right] \nonumber \\
& \quad \quad = \E\!\left[\left.g_0\!\left(\sigma^+_{2k-1}\right) - g_0\!\left(\sigma^-_{2k-1}\right)\right|\sigma^+_{2k-2},\sigma^-_{2k-2}\right] 
\label{Eq: g0 Difference}
\end{align}
where the first equality follows from the tower and Markov properties, and the second equality holds because $L_{2k} \sigma_{2k} \sim \mathsf{binomial}(L_{2k},g_{0}(\sigma))$ given $\sigma_{2k-1} = \sigma$. Then, recalling that $g_0(\sigma) = (\sigma * \delta)^2 = (1-2\delta)^2 \sigma^2 + 2 \delta (1-2\delta) \sigma + \delta^2$, we can compute:
\begin{align}
& \E\!\left[\left.g_0\!\left(\sigma^+_{2k-1}\right)\right|\sigma^+_{2k-2},\sigma^-_{2k-2}\right] \nonumber \\
& = \E\!\left[\left.g_0\!\left(\sigma^+_{2k-1}\right)\right|\sigma^+_{2k-2}\right] \nonumber \\
& = \E\!\left[\left.(1-2\delta)^2 \sigma^{+\quad 2}_{2k-1} + 2 \delta (1-2\delta) \sigma^+_{2k-1} + \delta^2 \right|\sigma^+_{2k-2}\right] \nonumber \\
& = (1-2\delta)^2 \! \left(\VAR\!\left(\left.\sigma^{+}_{2k-1} \right|\sigma^+_{2k-2}\right) + \E\!\left[\left.\sigma^{+}_{2k-1} \right|\sigma^+_{2k-2}\right]^2\right) \nonumber \\
& \quad \, + 2 \delta (1-2\delta) \E\!\left[\left.\sigma^{+}_{2k-1} \right|\sigma^+_{2k-2}\right] + \delta^2 \nonumber \\
& = (1-2\delta)^2 g_1\!\left(\sigma^+_{2k-2}\right)^2 + 2 \delta (1-2\delta) g_1\!\left(\sigma^+_{2k-2}\right) + \delta^2 \nonumber \\
& \quad \, + (1-2\delta)^2 \, \frac{g_1\!\left(\sigma^+_{2k-2}\right) \! \left(1 - g_1\!\left(\sigma^+_{2k-2}\right)\right)}{L_{2k-1}} \nonumber \\
& = g\!\left(\sigma^+_{2k-2}\right) + (1-2\delta)^2 \, \frac{g_1\!\left(\sigma^+_{2k-2}\right) \! \left(1 - g_1\!\left(\sigma^+_{2k-2}\right)\right)}{L_{2k-1}} 
\label{Eq: g0 Expectation}
\end{align}
where the first equality uses property 3 of the monotone Markovian coupling, and the fourth equality uses the fact that $L_{2k-1} \sigma_{2k-1} \sim \mathsf{binomial}(L_{2k-1},g_{1}(\sigma))$ given $\sigma_{2k-2} = \sigma$. Using \eqref{Eq: g0 Difference} and \eqref{Eq: g0 Expectation}, we get:
\begin{align}
& \E\!\left[\left.\sigma^+_{2k} - \sigma^-_{2k}\right|\sigma^+_{2k-2},\sigma^-_{2k-2}\right] \nonumber \\
& \quad \quad = g\!\left(\sigma^+_{2k-2}\right) - g\!\left(\sigma^-_{2k-2}\right) \nonumber \\
& \quad \quad \quad + (1-2\delta)^2 \Bigg(\frac{g_1\!\left(\sigma^+_{2k-2}\right) \! \left(1 - g_1\!\left(\sigma^+_{2k-2}\right)\right)}{L_{2k-1}} \nonumber \\
& \quad \quad \quad \quad \quad \quad \quad \quad \quad - \frac{g_1\!\left(\sigma^-_{2k-2}\right) \! \left(1 - g_1\!\left(\sigma^-_{2k-2}\right)\right)}{L_{2k-1}}\Bigg) \nonumber \\
& \quad \quad = g\!\left(\sigma^+_{2k-2}\right) - g\!\left(\sigma^-_{2k-2}\right) \nonumber \\
& \quad \quad \quad + (1-2\delta)^2 \Bigg(\frac{g_1\!\left(\sigma^+_{2k-2}\right) - g_1\!\left(\sigma^-_{2k-2}\right)}{L_{2k-1}} \nonumber \\
& \quad \quad \quad \quad \quad \quad \quad \quad \quad - \frac{g_1\!\left(\sigma^+_{2k-2}\right)^2 - g_1\!\left(\sigma^-_{2k-2}\right)^2}{L_{2k-1}}\Bigg) \nonumber \\
& \quad \quad \leq g\!\left(\sigma^+_{2k-2}\right) - g\!\left(\sigma^-_{2k-2}\right) \nonumber \\
& \quad \quad \quad + (1-2\delta)^2 \Bigg(\frac{g_1\!\left(\sigma^+_{2k-2}\right) - g_1\!\left(\sigma^-_{2k-2}\right)}{L_{2k-1}}\Bigg) \nonumber \\
& \quad \quad \leq \left(D(\delta) + \frac{2(1 - \delta)(1-2\delta)^3}{L_{2k-1}}\right) \! \left(\sigma^+_{2k-2} - \sigma^-_{2k-2}\right) \nonumber \\
& \quad \quad \leq \left(D(\delta) + \frac{2}{L_{2k-1}}\right) \! \left(\sigma^+_{2k-2} - \sigma^-_{2k-2}\right)
\label{Eq: Lipschitz expectation 2}
\end{align}
where the first inequality holds because $g_1(\sigma^+_{2k-2})^2 - g_1(\sigma^-_{2k-2})^2 \geq 0$ almost surely (since $g_1$ is non-negative and non-decreasing by \eqref{Eq: g_1 Derivative}, and $\sigma^+_{2k-2} \geq \sigma^-_{2k-2}$ almost surely by property 5 of the monotone Markovian coupling), the second inequality holds because $\sigma^+_{2k-2} \geq \sigma^-_{2k-2}$ almost surely and $g$ and $g_1$ have Lipschitz constants $D(\delta)$ and $\max_{\sigma \in [0,1]}{g_1^{\prime}(\sigma)} = 2(1 - \delta)(1 - 2\delta)$ respectively, and the final inequality holds because $(1 - \delta)(1-2\delta)^3 \leq 1$. Then, as in the proof of Theorem \ref{Thm:Phase Transition in Random Grid with Majority Rule Processing}, we can take expectations in \eqref{Eq: Lipschitz expectation 2} to obtain:
$$ 0 \leq \E\!\left[\sigma^+_{2k} - \sigma^-_{2k}\right] \leq \left(D(\delta) + \frac{2}{L_{2k-1}}\right) \E\!\left[\sigma^+_{2k-2} - \sigma^-_{2k-2}\right] $$
which recursively produces:
$$ 0 \leq \E\!\left[\sigma^+_{2k} - \sigma^-_{2k}\right] \leq \prod_{i = 1}^{k}{\left(D(\delta) + \frac{2}{L_{2i-1}}\right)} $$
where we use the fact that $\E\big[\sigma^+_{0} - \sigma^-_{0}\big] = 1$. 

Next, using \eqref{Eq: Initial TV Bound 2} with this bound, we get:
\begin{equation}
\label{Eq: Final TV Bound}
\E\!\left[\left\|P_{X_{2k}|G}^+ - P_{X_{2k}|G}^-\right\|_{\mathsf{TV}}\right] \leq L_{2k} \, \prod_{i = 1}^{k}{\left(D(\delta) + \frac{2}{L_{2i-1}}\right)} . 
\end{equation}
Recall that $L_k = o\big(E(\delta)^{-\frac{k}{2}}\big)$ and $\liminf_{k \rightarrow \infty}{L_k} > \frac{2}{E(\delta) - D(\delta)}$ for some $E(\delta) \in (D(\delta),1)$ (that depends on $\delta$). Hence, there exists $K = K(\delta) \in \N$ (that depends on $\delta$) such that for all $i > K$, $L_{2i-1} \geq \frac{2}{E(\delta) - D(\delta)}$. This means that we can further upper bound \eqref{Eq: Final TV Bound} as follows:
\begin{align*}
\forall k > K, \enspace & \E\!\left[\left\|P_{X_{2k}|G}^+ - P_{X_{2k}|G}^-\right\|_{\mathsf{TV}}\right] \\
& \quad \quad \leq L_{2k} \, E(\delta)^{k - K} \prod_{i = 1}^{K}{\left(D(\delta) + \frac{2}{L_{2i-1}}\right)} 
\end{align*}
and letting $k \rightarrow \infty$ produces:
\begin{equation}
\label{Eq: Auxiliary Convergence Condition 2}
\lim_{k \rightarrow \infty}{\E\!\left[\left\|P_{X_{2k}|G}^+ - P_{X_{2k}|G}^-\right\|_{\mathsf{TV}}\right]} = 0 \, . 
\end{equation}

Finally, as shown in the proof of Theorem \ref{Thm:Phase Transition in Random Grid with Majority Rule Processing}, $\big\|P_{X_{2k}|G}^+ - P_{X_{2k}|G}^-\big\|_{\mathsf{TV}} \in [0,1]$ forms a non-increasing sequence in $k$ for every realization of the random DAG $G$, and the pointwise limit random variable, $\lim_{k \rightarrow \infty}{\big\|P_{X_{2k}|G}^+ - P_{X_{2k}|G}^-\big\|_{\mathsf{TV}}} \in [0,1]$, has mean $\E\big[\lim_{k \rightarrow \infty}{\big\|P_{X_{2k}|G}^+ - P_{X_{2k}|G}^-\big\|_{\mathsf{TV}}}\big] = 0$ due to \eqref{Eq: Auxiliary Convergence Condition 2} and the bounded convergence theorem. Therefore, we must have:
$$ \lim_{k \rightarrow \infty}{\left\|P_{X_{2k}|G}^+ - P_{X_{2k}|G}^-\right\|_{\mathsf{TV}}} = 0 \quad G\text{-}a.s. $$
Moreover, since $\big\|P_{X_{k}|G}^+ - P_{X_{k}|G}^-\big\|_{\mathsf{TV}} \in [0,1]$ forms a non-increasing sequence in $k$, we have:
$$ \lim_{k \rightarrow \infty}{\left\|P_{X_{k}|G}^+ - P_{X_{k}|G}^-\right\|_{\mathsf{TV}}} = \lim_{k \rightarrow \infty}{\left\|P_{X_{2k}|G}^+ - P_{X_{2k}|G}^-\right\|_{\mathsf{TV}}} $$ 
for every realization of the random DAG $G$. Hence, we obtain \eqref{Eq:Strong TV Reconstruction Impossible}, which completes the proof.
\end{proof}

\renewcommand{\proofname}{Proof}

We remark that when $\delta \in \big(\delta_{\mathsf{andor}},\frac{1}{2}\big)$ and the condition $\liminf_{k \rightarrow \infty}{L_k} > \frac{2}{E(\delta) - D(\delta)}$ cannot be satisfied by any $E(\delta)$, if $L_k$ satisfies the condition of Proposition \ref{Prop: Slow Growth of Layers} (in subsection \ref{Further Discussion}), then part 2 of Proposition \ref{Prop: Slow Growth of Layers} still yields the desired converse result. Finally, the next proposition demonstrates that the Markov chain $\{\sigma_{2k} : k \in \N\}$ converges almost surely when $\delta \in \big(\delta_{\mathsf{andor}},\frac{1}{2}\big)$, $L_k = \omega(\log(k))$, all processing functions at even levels are the AND rule, and all processing functions at odd levels are the OR rule. 

\begin{proposition}[AND-OR Random DAG Model Almost Sure Convergence]
\label{Prop: And-Or Grid Almost Sure Convergence}
If $\delta \in \big(\delta_{\mathsf{andor}},\frac{1}{2}\big)$ and $L_k = \omega(\log(k))$, then $\lim_{k \rightarrow \infty}{\sigma_{2k}} = t$ almost surely.
\end{proposition}

Proposition \ref{Prop: And-Or Grid Almost Sure Convergence} is proved in Appendix \ref{Proof of Proposition And-Or Grid Almost Sure Convergence}, and much like Proposition \ref{Prop: Majority Grid Almost Sure Convergence}, it can also be construed as a ``weak'' impossibility result.

\section{Deterministic Quasi-Polynomial Time and Randomized Polylogarithmic Time Constructions of DAGs where Broadcasting is Possible}
\label{Deterministic Quasi-Polynomial Time and Randomized Polylogarithmic Time Constructions of DAGs where Broadcasting is Possible}

In this section, we prove Theorem \ref{Thm: Reconstruction in Expander DAGs} and Proposition \ref{Prop: DAG Construction using Expander Graphs} by constructing deterministic bounded degree DAGs with $L_k = \Theta(\log(k))$ where broadcasting is possible. As mentioned in subsection \ref{Explicit Construction of DAGs where Broadcasting is Possible}, our construction is based on $d$-regular bipartite lossless $(d^{-6/5},d-2d^{4/5})$-expander graphs. So, we first verify that such graphs actually exist. Recall that we represent a $d$-regular bipartite graph as $B = (U,V,E)$, where $U$ and $V$ are disjoint sets of vertices and $E$ is the set of undirected edges. The next proposition is a specialization of \cite[Proposition 1, Appendix II]{SipserSpielman1996} which illustrates that randomly generated regular bipartite graphs are good expanders with high probability.

\begin{proposition}[Random Expander Graph {\cite[Lemma 1]{Pinsker1973}, \cite[Proposition 1, Appendix II]{SipserSpielman1996}}]
\label{Prop: Random Expander Graph}
Fix any fraction $\alpha \in (0,1)$ and any degree $d \in \N\backslash\!\{0\}$. Then, for every sufficiently large $n$ (depending on $\alpha$ and $d$), the randomly generated $d$-regular bipartite graph $\mathsf{B} = (U,V,\mathsf{E})$ with $|U| = |V| = n$ satisfies:
\begin{align*}
& \P\Big(\forall S \subseteq U \text{ with } |S| = \alpha n, \\
& \quad \, |\Gamma(S)| \geq n \!\left(1-(1-\alpha)^d - \sqrt{2 d \alpha H(\alpha)}\right)\!\Big) \\
& \quad \quad \quad \quad \quad \quad \quad \quad \quad \quad > 1 - \binom{n}{n\alpha} \exp(-n H(\alpha)) \\
& \quad \quad \quad \quad \quad \quad \quad \quad \quad \quad \geq 1 - \frac{e}{2 \pi \sqrt{\alpha (1- \alpha) n}}
\end{align*}
where $H(\alpha) \triangleq -\alpha \log(\alpha) - (1-\alpha)\log(1-\alpha)$ is the binary Shannon entropy function.
\end{proposition}

We note that the probability measure $\P$ in Proposition \ref{Prop: Random Expander Graph} is defined by the random $d$-regular bipartite graph $\mathsf{B}$, whose vertices $U$ and $V$ are fixed and edges $\mathsf{E}$ are random. In particular, $\mathsf{B}$ is generated as follows (cf. \textit{configuration model} in \cite[Section 2.4]{Bollobas2001}): 
\begin{enumerate}
\item Fix a complete bipartite graph $\hat{B} = (\hat{U},\hat{V},\hat{E})$ such that $|\hat{U}| = |\hat{V}| = d n$.
\item Randomly and uniformly select a perfect matching $\mathsf{M} \subseteq \hat{E}$ in $\hat{B}$. 
\item Group sets of $d$ consecutive vertices in $\hat{U}$, respectively $\hat{V}$, to generate a set of $n$ super-vertices $U$, respectively $V$. 
\item This yields a random $d$-regular bipartite graph $\mathsf{B} = (U,V,\mathsf{E})$, where every edge in $\mathsf{E}$ is an edge between super-vertices in $\mathsf{M}$. 
\end{enumerate}
Note that we allow for the possibility that two vertices in $\mathsf{B}$ have multiple edges between them. The first inequality in Proposition \ref{Prop: Random Expander Graph} is proved in \cite[Appendix II]{SipserSpielman1996}. On the other hand, the second inequality in Proposition \ref{Prop: Random Expander Graph} is a straightforward consequence of estimating the binomial coefficient using precise Stirling's formula bounds, cf. \cite[Chapter II, Section 9, Equation (9.15)]{Feller1968}: 
\begin{equation}
\forall n \in \N\backslash\!\{0\}, \enspace \sqrt{2 \pi n} \left(\frac{n}{e}\right)^n \leq n! \leq e \sqrt{n} \left(\frac{n}{e}\right)^n .
\end{equation}
The second inequality portrays that the probability in Proposition \ref{Prop: Random Expander Graph} tends to $1$ as $n \rightarrow \infty$. Moreover, strictly speaking, $\alpha n$ must be an integer, but we will neglect this detail throughout our exposition for simplicity (as in subsection \ref{Explicit Construction of DAGs where Broadcasting is Possible}). We next use this proposition to establish the existence of pertinent regular bipartite lossless expander graphs. 

\begin{corollary}[Lossless Expander Graph]
\label{Cor: Lossless Expander Graph}
Fix any $\epsilon \in (0,1)$ and any degree $d \geq \big(\frac{2}{\epsilon}\big)^{\!5}$. Then, for every sufficient large $n$ (depending on $d$), the randomly generated $d$-regular bipartite graph $\mathsf{B} = (U,V,\mathsf{E})$ with $|U| = |V| = n$ satisfies:
\begin{align*}
& \P\!\left(\forall S \subseteq U \text{ with } |S| = \frac{n}{d^{6/5}}, \, |\Gamma(S)| \geq (1-\epsilon) d |S| \right) \\
& \quad \quad \quad \quad \quad \quad \quad \quad \quad > 1 - \frac{e}{2 \pi \sqrt{d^{-6/5} \!\left(1 - d^{-6/5}\right) n}} \, . 
\end{align*}
Hence, for every sufficient large $n$ (depending on $d$), there exists a $d$-regular bipartite lossless $(d^{-6/5},(1-\epsilon)d)$-expander graph $B = (U,V,E)$ with $|U| = |V| = n$ such that for every subset of vertices $S \subseteq U$, we have:
$$ |S| = \frac{n}{d^{6/5}} \quad \Rightarrow \quad |\Gamma(S)| \geq (1-\epsilon) d |S| = (1-\epsilon) \frac{n}{d^{1/5}} \, . $$
\end{corollary}

Corollary \ref{Cor: Lossless Expander Graph} is proved in Appendix \ref{Proof of Corollary Lossless Expander Graph}. We remark that explicit constructions of bipartite lossless expander graphs $B$ where only the vertices in $U$ are $d$-regular can be found in the literature, cf. \cite{Capalboetal2002}, but we require the vertices in $V$ to be $d$-regular in our construction.  

As we discussed in subsection \ref{Explicit Construction of DAGs where Broadcasting is Possible}, $d$-regular bipartite lossless expander graphs can be concatenated to produce a DAG where broadcasting is possible. To formally establish this, we first argue that a single $d$-regular bipartite lossless expander graph, when perceived as two successive layers of a deterministic DAG, exhibits a ``one-step broadcasting'' property. Fix any crossover probability $\delta \in \big(0,\frac{1}{2}\big)$, and choose any sufficiently large \textit{odd degree} $d = d(\delta)$ (that depends on $\delta$) such that \eqref{Eq: Relation between d and delta} (reproduced below) holds: 
\begin{equation}
\frac{8}{d^{1/5}} + d^{6/5} \exp\!\left(-\frac{(1 - 2\delta)^2 (d - 4)^2}{8 d}\right) \leq \frac{1}{2}
\end{equation}
where the left hand side tends to $0$ as $d \rightarrow \infty$ for fixed $\delta$, and the minimum value of $d$ satisfying this inequality increases as $\delta \rightarrow \frac{1}{2}^{-}$. Then, Corollary \ref{Cor: Lossless Expander Graph} demonstrates that for any sufficiently large $n$ (depending on $d$), there exists a $d$-regular bipartite lossless $(d^{-6/5},d-2d^{4/5})$-expander graph $B = (U,V,E)$ with $|U| = |V| = n$ such that the expansion property in \eqref{Eq: Expansion Property 2} (reproduced below) holds:
\begin{equation}
|S| = \frac{n}{d^{6/5}} \enspace \Rightarrow \enspace |\Gamma(S)| \geq (1-\epsilon) \frac{n}{d^{1/5}} \text{ with } \epsilon = \frac{2}{d^{1/5}}
\end{equation}
for every subset of vertices $S \subseteq U$. Note that in the statements of Lemma \ref{Lemma: One-Step Broadcasting} (see below), Theorem \ref{Thm: Reconstruction in Expander DAGs}, and Proposition \ref{Prop: DAG Construction using Expander Graphs}, we assume the existence of such $d$-regular bipartite lossless $(d^{-6/5},d-2d^{4/5})$-expander graphs without proof due to Corollary \ref{Cor: Lossless Expander Graph}. Let us assume that the undirected edges in $E$ are actually all directed from $U$ to $V$, and construe $B$ as two consecutive levels of a deterministic DAG upon which we are broadcasting (as in subsection \ref{Explicit Construction of DAGs where Broadcasting is Possible}). In particular, let the Bernoulli random variable corresponding to any vertex $v \in U \cup V$ be denoted by $X_v$, and suppose each (directed) edge of $B$ is an independent $\mathsf{BSC}(\delta)$ as before. Furthermore, let the Boolean processing function at each vertex in $V$ be the majority rule, which is always well-defined as $d$ is odd. This defines a Bayesian network on $B$, and the ensuing lemma establishes the feasibility of ``one-step broadcasting'' down this Bayesian network.

\begin{lemma}[One-Step Broadcasting in Expander DAG]
\label{Lemma: One-Step Broadcasting}
For any fixed noise level $\delta \in \big(0,\frac{1}{2}\big)$, any sufficiently large odd degree $d = d(\delta) \geq 5$ (that depends on $\delta$) satisfying \eqref{Eq: Relation between d and delta}, and any sufficiently large $n$ (depending on $d$), consider any fixed $d$-regular bipartite lossless $(d^{-6/5},d-2d^{4/5})$-expander graph $B = (U,V,E)$ with $|U| = |V| = n$ such that the edges have independent $\mathsf{BSC}(\delta)$ noise and the vertices in $V$ use majority Boolean processing functions (as outlined above). Then, for every input distribution on $\{X_u : u \in U\}$, we have:
$$ \P\!\left(\sum_{v \in V}{X_v} > \frac{n}{d^{6/5}} \, \middle| \, \sum_{u \in U}{X_u} \leq \frac{n}{d^{6/5}}\right) \leq \exp\!\left(- \frac{n}{2 d^{12/5}}\right) . $$
\end{lemma} 

\begin{proof}
We begin with some useful definitions. For any vertex $v \in V$, let $\mathsf{pa}(v)$ denote the multiset of vertices in $U$ that are parents of $v$. (Note that $\mathsf{pa}(v)$ is a multiset because there may be multiple edges between two vertices, and $|\mathsf{pa}(v)| = d$.) Let $S \triangleq \{u \in U : X_u = 1\} \subseteq U$ denote the subset of vertices in $U$ that take value $1$, which implies that $|S| = \sum_{u \in U}{X_u}$. Furthermore, for any vertex $v \in V$, let $N_v \triangleq \sum_{u \in \mathsf{pa}(v)}{X_u}$ denote the number of parents of $v$ in $S$ that have value $1$ (counting with repetition). Finally, let $T \triangleq \{v \in V : N_v \geq t\} \subseteq V$ denote the subset of vertices in $V$ with at least $t \in \N\backslash\!\{0,1\}$ parents in $S$. We will assign an appropriate value to $t$ below.

Suppose $|S| = \sum_{u \in U}{X_u} \leq n/d^{6/5}$ (which is the event we condition upon in the lemma statement). Consider the case where $|S| = n/d^{6/5}$. Then, applying the expansion property in \eqref{Eq: Expansion Property 2} yields (the ``vertex counting'' bound):
\begin{equation}
\label{Eq: Vertex Counting}
|\Gamma(S)| = |T| + |\Gamma(S)\backslash T| \geq (1-\epsilon) \frac{n}{d^{1/5}}
\end{equation}
where $T \subseteq \Gamma(S)$ by definition of $T$, and $\epsilon = 2 d^{-1/5}$. Moreover, we also have the ``edge counting'' bound:
\begin{equation}
\label{Eq: Edge Counting}
t|T| + |\Gamma(S)\backslash T| \leq d |S| = \frac{n}{d^{1/5}}
\end{equation}
since each vertex in $T$ has at least $t$ edges from $S$, each vertex in $\Gamma(S)\backslash T$ has at least $1$ edge from $S$, and the total number of outgoing edges from $S$ is $d |S|$. Combining \eqref{Eq: Vertex Counting} and \eqref{Eq: Edge Counting} produces:
$$ (1-\epsilon) \frac{n}{d^{1/5}} - |T| \leq |\Gamma(S)\backslash T| \leq \frac{n}{d^{1/5}} - t |T| $$
which implies that:
\begin{equation}
\label{Eq: Bound on |T|}
|T| \leq \frac{n \epsilon}{d^{1/5} (t-1)} = \frac{2 n}{d^{2/5} (t-1)} \, . 
\end{equation}
On the other hand, in the the case where $|S| < n/d^{6/5}$, if we flip the values of vertices in $U\backslash S$ to $1$ and subsequently increase the cardinality of $S$, then the cardinality of $T$ also increases or remains the same. Hence, if $|S| = \sum_{u \in U}{X_u} \leq n/d^{6/5}$, then \eqref{Eq: Bound on |T|} also holds.

Now, for any input distribution on $\{X_u : u \in U\}$, observe that:
\begin{align}
& \P\!\left(\sum_{v \in V}{X_v} > \frac{n}{d^{6/5}} \, \middle| \, \sum_{u \in U}{X_u} \leq \frac{n}{d^{6/5}}\right) \nonumber \\
& = \P\!\left(\sum_{v \in T}{X_v} + \sum_{v \in V\backslash T}{X_v} > \frac{n}{d^{6/5}} \, \middle| \, |S| \leq \frac{n}{d^{6/5}}\right) \nonumber \\
& \leq \P\!\left(\sum_{v \in V\backslash T}{X_v} > \frac{n}{d^{6/5}} - |T| \, \middle| \, |S| \leq \frac{n}{d^{6/5}}\right) \nonumber \\
& \leq \P\!\left(\sum_{v \in V\backslash T}{X_v} > \frac{n}{d^{6/5}} - \frac{2 n}{d^{2/5} (t-1)} \, \middle| \, |S| \leq \frac{n}{d^{6/5}}\right) \nonumber \\
& = \E\vast[\P\vast(\!\left.\sum_{v \in V\backslash T}{X_v} > \frac{n}{d^{6/5}} - \frac{2 n}{d^{2/5} (t-1)} \, \right| V\backslash T, \nonumber \\
& \quad \quad \quad \quad \enspace \left.\left\{N_v : v \in V\backslash T\right\}\!\vast) \middle| \, |S| \leq \frac{n}{d^{6/5}}\right.\!\vast] 
\label{Eq: Expander Key Conditional Independence Step}
\\
& \leq \E\vast[\P\vast(\!\left.\sum_{v \in V\backslash T}{X_v} > \frac{n}{d^{6/5}} - \frac{2 n}{d^{2/5} (t-1)} \, \right| V\backslash T, \nonumber \\
& \quad \quad \quad \quad \enspace \left.\left\{\forall v \in V \backslash T, \, N_v = t - 1\right\}\!\vast) \middle| \, |S| \leq \frac{n}{d^{6/5}}\right.\!\vast] \nonumber \\ 
& = \E\bigg[\P\bigg(\mathsf{binomial}(|V\backslash T|, \P(X_v = 1 | N_v = t-1)) \nonumber \\
& \quad \quad \quad \enspace \, \left. > \frac{n}{d^{6/5}} - \frac{2 n}{d^{2/5} (t-1)} \, \middle| \, V\backslash T \bigg) \right| |S| \leq \frac{n}{d^{6/5}}\bigg] \nonumber \\
& \leq \P\bigg(\mathsf{binomial}(n, \P(X_v = 1 | N_v = t-1)) \nonumber \\
& \quad \quad \enspace > \frac{n}{d^{6/5}} - \frac{2 n}{d^{2/5} (t-1)} \bigg) \nonumber \\
& \leq \P\Bigg(\mathsf{binomial}\!\left(n, \exp\!\left(-2 d (1 - 2\delta)^2 \! \left(\frac{1}{2} - \frac{t-1}{d}\right)^{\! 2}\right)\right) \nonumber \\
& \quad \quad \enspace > \frac{n}{d^{6/5}} - \frac{2 n}{d^{2/5} (t-1)} \Bigg)
\label{Eq: Bound used in second Hoeffding application}
\end{align}
where the steps hold due to the following reasons:
\begin{enumerate}
\item In the first equality, $T$ and $V\backslash T$ are random sets.
\item The second inequality holds because $X_v \in \{0,1\}$ for all $v \in T$.
\item The third inequality follows from \eqref{Eq: Bound on |T|}.
\item The fourth equality uses the fact that $\{X_v : v \in V\backslash T\}$ are conditionally independent of the event $\{|S| \leq n/d^{6/5}\}$ given $V\backslash T$ and $\{N_v : v \in V\backslash T\}$, and the conditional expectation in the fourth equality is over the random set $V\backslash T$ and the random variables $\{N_v : v \in V\backslash T\}$.
\item The fifth inequality holds because $N_v \leq t - 1$ for every $v \in V \backslash T$, and a straightforward monotone coupling argument shows that the distribution $P_{X_v|N_v = t - 1}$ stochastically dominates the distribution $P_{X_v|N_v = k}$ for any $k < t-1$. Furthermore, the conditional expectation in the fifth inequality is only over the random set $V\backslash T$.
\item The sixth equality holds because $\{X_v : v \in V \backslash T\}$ are conditionally i.i.d. given $V \backslash T$ and the event $\{\forall v \in V \backslash T, \, N_v = t - 1\}$.
\item The seventh inequality holds because $|V \backslash T| \leq n$, and a simple monotone coupling argument establishes that a $\mathsf{binomial}(n, \P(X_v = 1 | N_v = t-1))$ random variable stochastically dominates a $\mathsf{binomial}(|V\backslash T|, \P(X_v = 1 | \allowbreak N_v = t-1))$ random variable.
\item The eighth inequality holds because a $\mathsf{binomial}(n,p)$ random variable stochastically dominates a $\mathsf{binomial}(n,\allowbreak q)$ random variable when $p \geq q$ (again via a monotone coupling argument), and Hoeffding's inequality yields:
\begin{align}
& \P(X_v = 1 | N_v = t-1) \nonumber \\
& \quad \quad = \P\!\left(\sum_{i = 1}^{t-1}{Z_i} + \sum_{j = 1}^{d - t + 1}{Y_j} > \frac{d}{2}\right) \nonumber \\
& \quad \quad \leq \exp\!\left(-2 d (1 - 2\delta)^2 \!\left(\frac{1}{2} - \frac{t-1}{d}\right)^{\! 2}\right)
\label{Eq: First Assumption}
\end{align}
where $Z_i$ are i.i.d. $\Ber(1-\delta)$, $Y_j$ are i.i.d. $\Ber(\delta)$, $\{Z_i : i \in \{1,\dots, t-1\}\}$ and $\{Y_j : j \in \{1,\dots,d-t+1\}\}$ are independent, we assume that $\frac{t-1}{d} < \frac{1}{2}$, and we use the fact that $X_v$ is the majority of its parents' values after passing them through independent $\mathsf{BSC}(\delta)$'s.
\end{enumerate}
Finally, applying Hoeffding's inequality once more to \eqref{Eq: Bound used in second Hoeffding application} yields:
\begin{equation}
\label{Eq: Intermediate one-step bound}
\begin{aligned}
& \P\!\left(\sum_{v \in V}{X_v} > \frac{n}{d^{6/5}} \, \middle| \, \sum_{u \in U}{X_u} \leq \frac{n}{d^{6/5}}\right) \\
& \leq \exp\!\vast(\!\!-2 n \Bigg(\frac{1}{d^{6/5}} - \frac{2}{d^{2/5} (t-1)} \\
& \quad \quad \quad \quad \quad \quad \, \, - \exp\!\left(-2 d (1 - 2\delta)^2 \!\left(\frac{1}{2} - \frac{t-1}{d}\right)^{\! 2}\right)\!\Bigg)^{\! 2} \vast)
\end{aligned}
\end{equation}
where we assume that:
\begin{equation}
\label{Eq: Second Assumption}
\frac{1}{d^{6/5}} - \frac{2}{d^{2/5} (t-1)} > \exp\!\left(-2 d (1 - 2\delta)^2 \!\left(\frac{1}{2} - \frac{t-1}{d}\right)^{\! 2}\right) . 
\end{equation}
 
Next, let $t = 1 + \ceil[\big]{\frac{d}{4}}$ so that:\footnote{The choice of $t$ is arbitrary and we could have chosen any $t$ such that $0 < \frac{t-1}{d} < \frac{1}{2}$.}
$$ \frac{1}{4} \leq \frac{t-1}{d} \leq \frac{1}{4} + \frac{1}{d} \, . $$
Since we have assumed in the lemma statement that $d \geq 5$, the upper bound on $\frac{t-1}{d}$ illustrates that $\frac{t-1}{d} < \frac{1}{2}$, which ensures that \eqref{Eq: First Assumption} is valid. Furthermore, using both the upper and lower bounds on $\frac{t-1}{d}$, notice that \eqref{Eq: Second Assumption} is also valid if we have:
\begin{align*}
\frac{1}{d^{6/5}} - \frac{8}{d^{7/5}} & > \exp\!\left(-\frac{(1 - 2\delta)^2 (d - 4)^2}{8 d}\right) \\
\Leftrightarrow \quad 1 & > \frac{8}{d^{1/5}} + d^{6/5} \exp\!\left(-\frac{(1 - 2\delta)^2 (d - 4)^2}{8 d}\right) 
\end{align*}
which is true by our assumption in \eqref{Eq: Relation between d and delta}. In fact, a simple computation shows that:
\begin{align*}
& \frac{1}{d^{6/5}} - \frac{2}{d^{2/5} (t-1)} - \exp\!\left(-2 d (1 - 2\delta)^2 \!\left(\frac{1}{2} - \frac{t-1}{d}\right)^{\! 2}\right) \\
& \quad \quad \enspace \geq \frac{1}{d^{6/5}} - \frac{8}{d^{7/5}} - \exp\!\left(-\frac{(1 - 2\delta)^2 (d - 4)^2}{8 d}\right) \\
& \quad \quad \enspace \geq \frac{1}{2 d^{6/5}} 
\end{align*}
where the second inequality is equivalent to \eqref{Eq: Relation between d and delta}. Therefore, we have from \eqref{Eq: Intermediate one-step bound}:
$$ \P\!\left(\sum_{v \in V}{X_v} > \frac{n}{d^{6/5}} \, \middle| \, \sum_{u \in U}{X_u} \leq \frac{n}{d^{6/5}}\right) \leq \exp\!\left(- \frac{n}{2 d^{12/5}} \right) $$
which completes the proof.
\end{proof}

Intuitively, Lemma \ref{Lemma: One-Step Broadcasting} parallels \eqref{Eq: Stability whp} in the proof of Theorem \ref{Thm:Phase Transition in Random Grid with Majority Rule Processing} in section \ref{Analysis of Majority Rule Processing in Random Grid}. The lemma portrays that if the proportion of $1$'s is small in a given layer, then it remains small in the next layer with high probability when the edges between the layers are defined by a regular bipartite lossless expander graph. We next prove Theorem \ref{Thm: Reconstruction in Expander DAGs} by showing using Lemma \ref{Lemma: One-Step Broadcasting} that the root bit of an expander-based DAG (outlined in subsection \ref{Explicit Construction of DAGs where Broadcasting is Possible}) can be reconstructed using the majority decision rule $\hat{S}_k = \I\big\{\sigma_k \geq \frac{1}{2}\big\}$.

\renewcommand{\proofname}{Proof of Theorem \ref{Thm: Reconstruction in Expander DAGs}}

\begin{proof}
Fix any $\delta \in \big(0,\frac{1}{2}\big)$, any sufficiently large odd $d = d(\delta) \geq 5$ satisfying \eqref{Eq: Relation between d and delta}, and any sufficiently large constant $N = N(\delta) \in \N$ such that $M = \exp(N/(4 d^{12/5})) \geq 2$ and for every $n \geq N$, there exists a $d$-regular bipartite lossless $(d^{-6/5},d - 2d^{4/5})$-expander graph $B_n = (U_n,V_n,E_n)$ with $|U_n| = |V_n| = n$ that satisfies \eqref{Eq: Expansion Property 2} for every subset $S \subseteq U_n$. We will argue that broadcasting is possible for the Bayesian network defined on the bounded degree DAG with level sizes $\{L_k : k \in \N\}$ given by \eqref{Eq: Expander Level Sizes} and edge configuration described in the statement of Theorem \ref{Thm: Reconstruction in Expander DAGs}. Note that it is straightforward to verify that $L_k = \Theta(\log(k))$ (for fixed $\delta$). 

As before, we follow the proof of Theorem \ref{Thm:Phase Transition in Random Grid with Majority Rule Processing} in section \ref{Analysis of Majority Rule Processing in Random Grid}. So, we first construct a monotone Markovian coupling $\{(X^-_k,X^+_k) : k \in \N\}$ between the Markov chains $\{X^+_k : k \in \N\}$ and $\{X^-_k : k \in \N\}$ (which denote versions of the Markov chain $\{X_k : k \in \N\}$ initialized at $X^+_0 = 1$ and $X^-_0 = 0$, respectively) such that along any edge BSC of the deterministic expander-based DAG, say $(X_{k,j},X_{k+1,i})$, $X_{k,j}^+$ and $X_{k,j}^-$ are either both copied with probability $1 - 2\delta$, or a shared independent $\Ber\big(\frac{1}{2}\big)$ bit is produced with probability $2\delta$ that becomes the value of both $X_{k+1,i}^+$ and $X_{k+1,i}^-$. This coupling satisfies the three properties delineated at the outset of the proof of Theorem \ref{Thm:Phase Transition in Random Grid with Majority Rule Processing} in section \ref{Analysis of Majority Rule Processing in Random Grid}. Furthermore, let $\sigma^+_k$ and $\sigma_k^-$ for $k \in \N$ be random variables with distributions $P_{\sigma_k|\sigma_0 = 1}$ and $P_{\sigma_k|\sigma_0 = 0}$, respectively (which means that $\sigma^+_0 = 1$ and $\sigma^-_0 = 0$).

Notice that Lemma \ref{Lemma: One-Step Broadcasting} implies the following result:
\begin{equation}
\label{Eq: One Step Bound 1}
\P\!\left(\sigma_k^- \leq \frac{1}{d^{6/5}} \, \middle| \, \sigma_{k-1}^- \leq \frac{1}{d^{6/5}}\right) \geq 1 - \exp\!\left(- \frac{L_{k-1}}{2 d^{12/5}}\right)
\end{equation}
for every pair of consecutive levels $k-1$ and $k$ such that $L_{k} = L_{k-1}$. Moreover, for every pair of consecutive levels $k-1$ and $k$ such that $L_{k} = 2 L_{k-1}$, we have:
\begin{align*}
& \P\!\left(\sigma_k^- > \frac{1}{d^{6/5}} \, \middle| \, \sigma_{k-1}^- \leq \frac{1}{d^{6/5}}\right) \\
& = \P\vast(\frac{1}{L_{k-1}}\sum_{i = 0}^{L_{k-1} - 1}{X_{k,i}^{-}} \\
& \quad \quad \enspace \, + \frac{1}{L_{k-1}}\sum_{j = L_{k-1}}^{L_{k} - 1}{X_{k,j}^{-}} > \frac{2}{d^{6/5}} \left| \, \sigma_{k-1}^- \leq \frac{1}{d^{6/5}}\vast)\right. \\
& = \P\vast(\!\left\{\frac{1}{L_{k-1}}\sum_{i = 0}^{L_{k-1} - 1}{X_{k,i}^{-}} > \frac{1}{d^{6/5}}\right\} \\
& \quad \quad \enspace \, \, \cup \!\left.\left\{ \frac{1}{L_{k-1}}\sum_{j = L_{k-1}}^{L_{k} - 1}{X_{k,j}^{-}} > \frac{1}{d^{6/5}}\right\} \right| \sigma_{k-1}^- \leq \frac{1}{d^{6/5}}\vast) \\
& \leq \P\!\left(\frac{1}{L_{k-1}}\sum_{i = 0}^{L_{k-1} - 1}{X_{k,i}^{-}} > \frac{1}{d^{6/5}} \, \middle| \, \sigma_{k-1}^- \leq \frac{1}{d^{6/5}}\right) \\
& \quad \, + \P\!\left(\frac{1}{L_{k-1}}\sum_{j = L_{k-1}}^{L_{k} - 1}{X_{k,j}^{-}} > \frac{1}{d^{6/5}}\, \middle| \, \sigma_{k-1}^- \leq \frac{1}{d^{6/5}}\right) \\
& \leq 2\exp\!\left(- \frac{L_{k-1}}{2 d^{12/5}}\right)
\end{align*}
where the first inequality follows from the union bound, and the final inequality follows from Lemma \ref{Lemma: One-Step Broadcasting} and the construction of our DAG (recall that two separate $d$-regular bipartite lossless $(d^{-6/5},d-2d^{4/5})$-expander graphs make up the edges between $X_{k-1}$ and $X_{k}^1$, and between $X_{k-1}$ and $X_{k}^2$, respectively). This implies that:
\begin{equation}
\label{Eq: One Step Bound 2}
\P\!\left(\sigma_k^- \leq \frac{1}{d^{6/5}} \, \middle| \, \sigma_{k-1}^- \leq \frac{1}{d^{6/5}}\right) \geq 1 - 2\exp\!\left(- \frac{L_{k-1}}{2 d^{12/5}}\right)
\end{equation}
for every pair of consecutive levels $k-1$ and $k$ such that $L_{k} = 2 L_{k-1}$, as well as for every pair of consecutive levels $k-1$ and $k$ such that $L_{k} = L_{k-1}$ (by slackening the bound in \eqref{Eq: One Step Bound 1}). Hence, the bound in \eqref{Eq: One Step Bound 2} holds for all levels $k \geq 2$.

Now fix any $\tau > 0$, and choose a sufficiently large value $K = K(\delta,\tau) \in \N$ (that depends on $\delta$ and $\tau$) such that:
\begin{equation}
\label{Eq: Bound on Tail of Sum 3}
2 \sum_{k = K+1}^{\infty}{\exp\!\left(- \frac{L_{k-1}}{2 d^{12/5}}\right)} \leq \tau \, . 
\end{equation}
Note that such $K$ exists because $2\sum_{k = 1}^{\infty}{1/k^2} = \pi^2 /3 < +\infty$, and for every $m \in \N$ and every $M^{\floor{2^{m-1}}} < k \leq M^{2^m}$,\footnote{Note that the floor function in $M^{\floor{2^{m-1}}}$ ensures that $M^{\floor{2^{m-1}}} = 1$ when $m = 0$.} we have: 
$$ \exp\!\left(- \frac{L_{k}}{2 d^{12/5}}\right) \leq \frac{1}{k^2} \quad \Leftrightarrow \quad k \leq \exp\!\left(\frac{2^m N}{4 d^{12/5}}\right) = M^{2^m} $$
where the right hand side holds due to the construction of our deterministic DAG (see \eqref{Eq: Expander Level Sizes}). Using the continuity of probability measures, observe that:
\begin{align*}
& \P\!\left(\bigcap_{k > K}\!{\left\{\sigma_k^- \leq \frac{1}{d^{6/5}}\right\}} \, \middle| \, \sigma_K^+ \geq 1 - \frac{1}{d^{6/5}}, \, \sigma_K^- \leq \frac{1}{d^{6/5}} \right) \\
& \quad \quad \quad \quad \quad \quad = \prod_{k > K}{\P\!\left(\sigma_k^- \leq \frac{1}{d^{6/5}} \, \middle| \, \sigma_{k-1}^- \leq \frac{1}{d^{6/5}}, \, A_{k}\right)} \\
& \quad \quad \quad \quad \quad \quad \geq \prod_{k > K}{1 - 2\exp\!\left(- \frac{L_{k-1}}{2 d^{12/5}}\right)} \\
& \quad \quad \quad \quad \quad \quad \geq 1 - 2 \sum_{k > K}{\exp\!\left(- \frac{L_{k-1}}{2 d^{12/5}}\right)} \\
& \quad \quad \quad \quad \quad \quad \geq 1 - \tau
\end{align*}
where $A_{k}$ for $k > K$ is the non-zero probability event defined as:
$$ A_{k} \triangleq \left\{ 
\begin{array}{lcl}
\!\!\!\left\{\sigma_{K}^+ \geq 1-\frac{1}{d^{6/5}}\right\} & \!\!\!, & \!\! k = K+1 \\
\!\!\!\left\{\sigma_{k-2}^- \leq \frac{1}{d^{6/5}},\dots,\sigma_{K}^- \leq \frac{1}{d^{6/5}}\right\} & & \\
\!\cap \left\{\sigma_{K}^+ \geq 1-\frac{1}{d^{6/5}}\right\} & \!\!\!, & \!\! k \geq K + 2
\end{array} \right. , $$
the first inequality follows from \eqref{Eq: One Step Bound 2}, and the final inequality follows from \eqref{Eq: Bound on Tail of Sum 3}. When using \eqref{Eq: One Step Bound 2} in the calculation above, we can neglect the effect of the conditioning event $A_k$, because a careful perusal of the proof of Lemma \ref{Lemma: One-Step Broadcasting} (which yields \eqref{Eq: One Step Bound 2} as a consequence) shows that \eqref{Eq: One Step Bound 2} continues to hold when we condition on events like $A_k$. Indeed, in step \eqref{Eq: Expander Key Conditional Independence Step} of the proof, the random variables $\{X_v : v \in V\backslash T\}$ are conditionally independent of the \textit{$\sigma$-algebra generated by random variables in previous layers of the DAG} given $V\backslash T$ and $\{N_v : v \in V\backslash T\}$. Moreover, this observation extends appropriately to the current Markovian coupling setting. We have omitted these details from Lemma \ref{Lemma: One-Step Broadcasting} for the sake of clarity. Therefore, we have for any $k > K$:
\begin{equation}
\label{Eq: Expander Stability with Extra Conditioning}
\P\!\left(\sigma_k^- \leq \frac{1}{d^{6/5}} \, \middle| \, \sigma_K^+ \geq 1 - \frac{1}{d^{6/5}}, \, \sigma_K^- \leq \frac{1}{d^{6/5}} \right) \geq 1 - \tau \, . 
\end{equation}
Likewise, due to the symmetry of the role of $0$'s and $1$'s in our deterministic DAG model, we can also prove mutatis mutandis that for any $k > K$:
\begin{equation}
\label{Eq: Expander Stability with Extra Conditioning 2}
\P\!\left(\sigma_k^+ \geq 1-\frac{1}{d^{6/5}} \, \middle| \, \sigma_K^+ \geq 1-\frac{1}{d^{6/5}}, \, \sigma_K^- \leq \frac{1}{d^{6/5}} \right) \geq 1 - \tau
\end{equation}
where $\tau$ and $K$ in \eqref{Eq: Expander Stability with Extra Conditioning 2} can be chosen to be the same as those in \eqref{Eq: Expander Stability with Extra Conditioning} without loss of generality.

Finally, define the event $E = \big\{\sigma_K^+ \geq 1-\frac{1}{d^{6/5}}, \, \sigma_K^- \leq \frac{1}{d^{6/5}}\big\}$, and observe that for all $k > K$:
\begin{align*}
& \P\!\left(\sigma_k^+ \geq \frac{1}{2}\right) - \P\!\left(\sigma_k^- \geq \frac{1}{2}\right) \\
& \geq \E\!\left[\left(\I\!\left\{\sigma_k^+ \geq \frac{1}{2}\right\} - \I\!\left\{\sigma_k^- \geq \frac{1}{2}\right\}\right)\I\!\left\{E\right\}\right] \\ 
& = \left(\P\!\left(\sigma_k^+ \geq \frac{1}{2} \, \middle| \, E \right) - \P\!\left(\sigma_k^- \geq \frac{1}{2} \, \middle| \, E \right)\right) \P(E) \\
& \geq \left(\P\!\left(\sigma_k^+ \geq 1-\frac{1}{d^{6/5}} \, \middle| \, E \right) - \P\!\left(\sigma_k^- > \frac{1}{d^{6/5}} \, \middle| \, E \right)\right) \P(E) \\
& \geq (1 - 2\tau) \P(E) > 0
\end{align*}
where the first inequality holds because $\I\big\{\sigma_{k}^+ \geq \frac{1}{2}\big\} - \I\big\{\sigma_{k}^- \geq \frac{1}{2}\big\} \geq 0$ almost surely due to the monotonicity of our Markovian coupling, the second inequality holds because $\frac{1}{d^{6/5}} < \frac{1}{2} < 1-\frac{1}{d^{6/5}}$ (since $d \geq 5$), and the final inequality follows from \eqref{Eq: Expander Stability with Extra Conditioning} and \eqref{Eq: Expander Stability with Extra Conditioning 2}. As argued in the proof of Theorem \ref{Thm:Phase Transition in Random Grid with Majority Rule Processing} in section \ref{Analysis of Majority Rule Processing in Random Grid}, this illustrates that $\limsup_{k \rightarrow \infty}{\P(\hat{S}_{k} \neq X_0)} < \frac{1}{2}$, which completes the proof.  
\end{proof}

We finally prove Proposition \ref{Prop: DAG Construction using Expander Graphs} by delineating two simple algorithms for constructing the constituent expander graphs of the deterministic bounded degree DAGs with $L_k = \Theta(\log(k))$ (from Theorem \ref{Thm: Reconstruction in Expander DAGs}) that admit reconstruction: a \textit{deterministic quasi-polynomial time} algorithm and a \textit{randomized polylogarithmic time} algorithm.

\renewcommand{\proofname}{Proof of Proposition \ref{Prop: DAG Construction using Expander Graphs}}

\begin{proof}
Fix any noise level $\delta \in \big(0,\frac{1}{2}\big)$ and consider the deterministic DAG from Theorem \ref{Thm: Reconstruction in Expander DAGs} with sufficiently large odd degree $d = d(\delta) \geq 5$ satisfying \eqref{Eq: Relation between d and delta}, level sizes given by \eqref{Eq: Expander Level Sizes} for sufficiently large $N = N(\delta) \in \N$ and $M = \exp(N/(4 d^{12/5})) \geq 2$, and edge configuration given by E-1, E-2, and E-3. The remainder of the proof is split into two parts. We first present a deterministic quasi-polynomial time algorithm to generate the constituent expander graphs of the aforementioned deterministic DAG, and then present a randomized polylogarithmic time algorithm for the same task.

\textbf{Part 1: Deterministic Quasi-Polynomial Time Algorithm}

We will require two useful facts:
\begin{enumerate}
\item For fixed sets of labeled vertices $U_n$ and $V_n$ with $|U_n| = |V_n| = n$, the total number of $d$-regular bipartite graphs $B_n = (U_n,V_n,E_n)$ is given by the multinomial coefficient:
$$ \binom{nd}{d,d,\dots,d} = \frac{(nd)!}{(d!)^n} \leq n^{nd} = \exp(d n \log(n)) $$
where we allow multiple edges between two vertices, and the inequality follows from e.g. \cite[Lemma 2.2]{CsiszarShields2004}. To see this, first attach $d$ edges to each vertex in $U_n$, and then successively count the number of ways to choose $d$ edges for each vertex in $V_n$.\footnote{Since the vertices in $U_n$ and $V_n$ are labeled, the total number of non-isomorphic $d$-regular bipartite graphs is smaller than $(nd)!/(d!)^n$. However, it is larger than $(nd)!/((2n)! (d!)^n)$, and the quasi-polynomial nature of our running time does not change with a more careful calculation of the number of non-isomorphic $d$-regular bipartite graphs.}  
\item Checking whether a given $d$-regular bipartite graph $B_n = (U_n,V_n,E_n)$ with $|U_n| = |V_n| = n$ satisfies \eqref{Eq: Expansion Property 2} for all subsets $S \subseteq U_n$ using brute force takes $O(n^2 \exp(n H(d^{-6/5})))$ time. To see this, note that there are $\binom{n}{n d^{-6/5}} \leq \exp(n H(d^{-6/5}))$ (cf. \cite[Lemma 2.2]{CsiszarShields2004}) subsets $S \subseteq U_n$ with $|S| = n d^{-6/5}$, and verifying \eqref{Eq: Expansion Property 2} takes $O(n^2)$ time for each such subset $S$.
\end{enumerate}

Consider any level $M^{2^{m-1}} < r \leq M^{2^m}$ with some associated $m \in \N\backslash\!\{0\}$. We show that the distinct expander graphs making up levels $0,\dots,r$ of the deterministic DAG in Theorem \ref{Thm: Reconstruction in Expander DAGs} can be constructed in quasi-polynomial time in $r$. In particular, we need to generate $m+1$ $d$-regular bipartite lossless $(d^{-6/5},d - 2d^{4/5})$-expander graphs $B_{N}, B_{2 N},\dots, B_{2^{m} N}$. So, for each $i \in \{0,\dots,m\}$, we generate $B_{2^{i} N}$ by exhaustively enumerating over the all possible $d$-regular bipartite graphs with $|U_{2^{i} N}| = |V_{2^{i} N}| = 2^{i} N$ until we find one that satisfies the desired expansion condition. (Note that such expander graphs are guaranteed to exist due to Corollary \ref{Cor: Lossless Expander Graph}.) Using the aforementioned facts 1 and 2, generating all $m+1$ desired graphs takes running time:
\begin{equation}
\begin{aligned}
& O\!\left((m+1) L_r^2 \exp(L_r H(d^{-6/5})) \exp(d L_r \log(L_r))\right) \\
& \qquad \qquad = O(\exp(\Theta(\log(r) \log\log(r))))
\end{aligned}
\end{equation}
where we also use the facts that $L_r = 2^m N = \Theta(\log(r))$ and $m = \Theta(\log\log(r))$ because $M^{2^{m-1}} < r \leq M^{2^m}$.\footnote{In our descriptions and analyses of the two algorithms, the Bachmann-Landau big-$O$ and big-$\Theta$ asymptotic notation conceal constants that depend on the fixed crossover probability parameter $\delta$.} Therefore, we can construct the constituent expander graphs in levels $0,\dots,r$ of our DAG in quasi-polynomial time with brute force. Note that we neglect details of how intermediate graphs are represented in our analysis. Moreover, we are not concerned with optimizing the quasi-polynomial running time.

\textbf{Part 2: Randomized Polylogarithmic Time Algorithm} 

We will require another useful fact:
\begin{enumerate}
\item[3)] A random $d$-regular bipartite graph $\mathsf{B} = (U_{n},V_{n},\mathsf{E})$ with $|U_{n}| = |V_{n}| = n$ can be generated according to the distribution described after Proposition \ref{Prop: Random Expander Graph} in $O(n)$ time. To see this, as outlined after Proposition \ref{Prop: Random Expander Graph}, we must first generate a uniform random perfect matching in a complete bipartite graph $\hat{B} = (\hat{U}_{dn},\hat{V}_{dn},\hat{E})$ such that $|\hat{U}_{dn}| = |\hat{V}_{dn}| = dn$. Observe that the edges in a perfect matching can be written as a permutation of the sequence $(1,2,\dots,dn)$, because each index and its corresponding value in the (permuted) sequence encodes an edge. So, perfect matchings in $\hat{B}$ are in bijective correspondence with permutations of the sequence $(1,2,\dots,dn)$. Therefore, we can generate a uniform random perfect matching by generating a uniform random permutation of $(1,2,\dots,dn)$ in $O(dn)$, or equivalently $O(n)$, time using the \textit{Fisher-Yates-Durstenfeld-Knuth shuffle}, cf. \cite[Section 3.4.2, p.145]{Knuth1997} and the references therein. (Note that we do not take the running time of the random number generation process into account.) All that remains is to create super-vertices, which can also be done in $O(n)$ time. 
\end{enumerate}

Suppose that the constant $N = N(\delta)$ also satisfies the additional condition:
\begin{equation}
\label{Eq: Additional Assumption}
N > \frac{e^2}{\left(6-4\sqrt{2}\right)\! \pi^2 d^{-6/5} \!\left(1 - d^{-6/5}\right)} 
\end{equation}
where $N$ still depends only on $\delta$ (through the dependence of $d$ on $\delta$). Consider any level $M^{2^{m-1}} < r \leq M^{2^m}$ with some associated $m \in \N\backslash\!\{0\}$. We present a \textit{Monte Carlo algorithm} that constructs the distinct expander graphs making up levels $0,\dots,r$ of the deterministic DAG in Theorem \ref{Thm: Reconstruction in Expander DAGs} with strictly positive success probability (that depends on $\delta$ but not on $r$) in polylogarithmic time in $r$. As in the previous algorithm, we ideally want to output $m+1$ $d$-regular bipartite lossless $(d^{-6/5},d - 2d^{4/5})$-expander graphs $B_{N}, B_{2 N},\dots, B_{2^{m} N}$. So, using the aforementioned fact 3, for each $i \in \{0,\dots,m\}$, we can generate a random $d$-regular bipartite graph $\mathsf{B} = (U_{2^{i} N},V_{2^{i} N},\mathsf{E})$ with $|U_{2^{i} N}| = |V_{2^{i} N}| = 2^{i} N$ according to the distribution in Corollary \ref{Cor: Lossless Expander Graph} in at most $O(2^{m} N)$ time. The total running time of the algorithm is thus:
\begin{equation}
O((m+1) 2^{m} N) = O(\log(r) \log\log(r))
\end{equation}
since $2^m N = \Theta(\log(r))$ and $m = \Theta(\log\log(r))$ as before. Furthermore, by Corollary \ref{Cor: Lossless Expander Graph}, the outputted random graphs satisfy \eqref{Eq: Expansion Property 2} for all relevant subsets of vertices with probability at least:
\begin{align}
& \prod_{i = 0}^{m}{\!\left(1 - \frac{e}{2\pi\sqrt{d^{-6/5} \! \left(1 - d^{-6/5}\right) 2^i N}}\right)} \nonumber \\
& \quad \geq 1 - \frac{e}{2\pi\sqrt{d^{-6/5} \! \left(1 - d^{-6/5}\right) N}}\sum_{i = 0}^{m}{\left(\frac{1}{\sqrt{2}}\right)^{\! i}} \nonumber \\
& \quad \geq 1 - \frac{e}{2\pi\sqrt{d^{-6/5} \! \left(1 - d^{-6/5}\right) N}}\sum_{i = 0}^{\infty}{\left(\frac{1}{\sqrt{2}}\right)^{\! i}} \nonumber\\
& \quad = 1 - \frac{e}{\left(2 - \sqrt{2}\right)\! \pi \sqrt{d^{-6/5} \! \left(1 - d^{-6/5}\right) N}} > 0
\label{Eq: Success Probability of Monte Carlo Algorithm}
\end{align}
where the first inequality is easily proved by induction, and the quantity in the final equality is strictly positive by assumption \eqref{Eq: Additional Assumption}. Hence, our Monte Carlo algorithm constructs the constituent expander graphs in levels $0,\dots,r$ of our DAG with strictly positive success probability in polylogarithmic time. Once again, note that we neglect details of how intermediate graphs are represented in our analysis. Moreover, we do not account for the running time of actually printing out levels $0,\dots,r$ of the DAG.

Finally, it is worth mentioning that the aforementioned fact 2 conveys that testing whether the $m+1$ $d$-regular random bipartite graphs our Monte Carlo algorithm generates are lossless $(d^{-6/5},d - 2d^{4/5})$-expander graphs takes polynomial running time:
\begin{equation}
\begin{aligned}
& O\!\left((m+1) 2^{2m} N^2 \exp(2^{m} N H(d^{-6/5}))\right) \\
& \qquad \qquad = O\!\left(\log\log(r) \log(r)^2 r^{8 d^{12/5} H(d^{-6/5})}\right)
\end{aligned}
\end{equation} 
where we use the fact that $2^m N < 2 N \log(r) / \! \log(M) = 8 d^{12/5} \log(r)$ since $r > M^{2^{m-1}}$ and $\log(M) = N / (4d^{12/5})$. Therefore, by repeatedly running our Monte Carlo algorithm until a valid set of $m+1$ $d$-regular bipartite lossless $(d^{-6/5},d - 2d^{4/5})$-expander graphs is produced, we obtain a \textit{Las Vegas algorithm} that runs in expected polynomial time $O\big(\log\log(r) \log(r)^2 r^{8 d^{12/5} H(d^{-6/5})}\big)$. 
\end{proof}

\renewcommand{\proofname}{Proof}

\section{Conclusion}
\label{Conclusion}

To conclude, we recapitulate the main contributions of this work. For random DAG models with indegree $d \geq 3$, we considered the intuitively reasonable setting where all Boolean processing functions are the majority rule. We proved in Theorem \ref{Thm:Phase Transition in Random Grid with Majority Rule Processing} that reconstruction of the root bit for this model is possible using the majority decision rule when $\delta < \delta_{\mathsf{maj}}$ and $L_k = \Omega(\log(k))$, and impossible using the ML decision rule in all but a zero measure set of DAGs when $\delta > \delta_{\mathsf{maj}}$ and $L_k$ is sub-exponential. On the other hand, when the indegree $d = 2$ so that the choices of Boolean processing functions are unclear, we derived a similar phase transition in Theorem \ref{Thm:Phase Transition in Random Grid with And-Or Rule Processing} for random DAG models with AND processing functions at all even levels and OR processing functions at all odd levels. These main results on random DAG models established the existence of deterministic DAGs where broadcasting is possible via the probabilistic method. For example, we conveyed in Corollary \ref{Cor: Existence of Grids where Reconstruction is Possible} that for any indegree $d \geq 3$, any noise level $\delta < \delta_{\mathsf{maj}}$, and $L_k = \Theta(\log(k))$, there exists a deterministic DAG with all majority processing functions such that reconstruction of the root bit is possible. In fact, Proposition \ref{Prop: Slow Growth of Layers} showed that the scaling $L_k = \Theta(\log(k))$ is optimal for such DAGs where broadcasting is possible. Finally, for any $\delta \in \big(0,\frac{1}{2}\big)$ and any sufficiently large bounded indegrees and outdegrees, we constructed explicit deterministic DAGs with $L_k = \Theta(\log(k))$ and all majority processing functions such that broadcasting is possible in Theorem \ref{Thm: Reconstruction in Expander DAGs} and Proposition \ref{Prop: DAG Construction using Expander Graphs}. Our construction utilized regular bipartite lossless expander graphs between successive layers of the DAGs, and we showed that the constituent expander graphs can be generated in either deterministic quasi-polynomial time or randomized polylogarithmic time in the number of levels. 

We close our discussion with a brief list of open problems that could serve as compelling directions for future research:
\begin{enumerate}
\item We conjectured in subsection \ref{Results on Random DAG Models} that in the random DAG model with $L_k = O(\log(k))$ and fixed $d \geq 3$, reconstruction is impossible for all choices of Boolean processing functions when $\delta \geq \delta_{\mathsf{maj}}$. Naturally, the analogous question for $d = 2$ is also open. Based on the reliable computation literature (see the discussion in subsection \ref{Results on Random DAG Models}), we can conjecture that majority processing functions are optimal for odd $d \geq 3$, and alternating levels of AND and OR processing is optimal for $d = 2$, but it is not obvious which processing functions are optimal for general even $d \geq 4$.
\item We provided some evidence for the previous conjecture in the odd $d \geq 3$ case in part 2 of Proposition \ref{Prop: Single Vertex Reconstruction}. A potentially simpler open question is to extend the proof of part 2 of Proposition \ref{Prop: Single Vertex Reconstruction} in Appendix \ref{Proof of Proposition Single Vertex Reconstruction} to show the impossibility of reconstruction using two (or more) vertices in the odd $d \geq 3$ case regardless of the choices of Boolean processing functions. 
\item It is unknown whether a result similar to part 2 of Proposition \ref{Prop: Single Vertex Reconstruction} holds for even $d \geq 2$. For the $d = 2$ setting, a promising direction is to try and exploit the \textit{potential function contraction} approach in \cite{Unger2008} instead of the TV distance contraction approach in \cite{HajekWeller1991,EvansSchulman2003}. 
\item As mentioned in subsection \ref{Explicit Construction of DAGs where Broadcasting is Possible}, it is an open problem to find a deterministic polynomial time algorithm to construct deterministic DAGs with sufficiently large $d$ and $L_k = \Theta(\log(k))$ given some $\delta$ for which broadcasting is possible. Indeed, the deterministic algorithm in Proposition \ref{Prop: DAG Construction using Expander Graphs} takes quasi-polynomial time.
\item As indicated above, for fixed $\delta$, Theorem \ref{Thm: Reconstruction in Expander DAGs} and Proposition \ref{Prop: DAG Construction using Expander Graphs} can only construct deterministic DAGs with sufficiently large $d$ such that broadcasting is possible. However, Corollary \ref{Cor: Existence of Grids where Reconstruction is Possible} elucidates that such deterministic DAGs exist for every $d \geq 3$ as long as $\delta < \delta_{\mathsf{maj}}$. It is an open problem to efficiently construct deterministic DAGs with $L_k = \Theta(\log(k))$ for arbitrary $d \geq 3$ and $\delta < \delta_{\mathsf{maj}}$, or $d = 2$ and $\delta < \delta_{\mathsf{andor}}$, such that broadcasting is possible.
\end{enumerate}

\appendices

\section{Proof of Proposition \ref{Prop: Single Vertex Reconstruction}}
\label{Proof of Proposition Single Vertex Reconstruction}

\begin{proof}
In this proof, we assume familiarity with the development and notation in section \ref{Analysis of Majority Rule Processing in Random Grid} and the proof of Theorem \ref{Thm:Phase Transition in Random Grid with Majority Rule Processing}. 

\textbf{Part 1:} We first prove part 1. Observe that for any $k \in \N\backslash\!\{0\}$:
\begin{align}
\P\!\left(X_{k,0} \neq X_{0,0}\right) & = \frac{1}{2}\P\!\left(X_{k,0}^+ = 0\right) + \frac{1}{2}\P\!\left(X_{k,0}^- = 1\right) \nonumber \\
& = \frac{1}{2}\E\!\left[\P\!\left(X_{k,0}^+ = 0 \middle| \sigma_{k}^{+}\right)\right] \nonumber \\
& \quad \, + \frac{1}{2}\E\!\left[\P\!\left(X_{k,0}^- = 1 \middle| \sigma_{k}^{-}\right)\right]  \nonumber \\
& = \frac{1}{2} \E\!\left[1 - \sigma_{k}^{+}\right] + \frac{1}{2} \E\!\left[\sigma_{k}^{-}\right]  \nonumber \\
& = \frac{1}{2} \left(1 - \E\!\left[\sigma_{k}^{+} - \sigma_{k}^{-}\right]\right)
\label{Eq: Equivalent Form of Error Probability}
\end{align}
where the third equality holds because $X_{k,0} \sim \Ber(\sigma)$ given $\sigma_k = \sigma$. To see this, recall the relation \eqref{Eq: Exchangeability} from subsection \ref{Random Grid Model}. Using this relation, it is straightforward to verify that $X_k$ is conditionally independent of $X_{0,0}$ given $\sigma_k$. Moreover, the conditional distribution $P_{X_k|\sigma_k}$ can be computed using \eqref{Eq: Exchangeability}, and this yields the desired conditional distribution $P_{X_{k,0}|\sigma_{k}}$ mentioned above. (We omit these calculations because it is intuitively obvious that random bits at level $k$ can be generated by first generating $\sigma_k$, then setting a uniformly and randomly chosen subset of vertices in $X_k$ of size $L_k \sigma_k$ to be $1$, and finally setting the remaining vertices in $X_k$ to be $0$.) 

Due to \eqref{Eq: Equivalent Form of Error Probability}, it suffices to prove that:
\begin{equation}
\label{Eq: STP Exchangeability}
\liminf_{k \rightarrow \infty}{\E\!\left[\sigma_{k}^{+} - \sigma_{k}^{-}\right]} > 0 \, . 
\end{equation}
To this end, recall from the proof of Theorem \ref{Thm:Phase Transition in Random Grid with Majority Rule Processing} that for any sufficiently small $\epsilon = \epsilon(\delta,d) > 0$ (that depends on $\delta$ and $d$) and any $\tau > 0$, there exists $K = K(\epsilon,\tau) \in \N$ (that depends on $\epsilon$ and $\tau$) such that for all $k > K$, \eqref{Eq: Stability with Extra Conditioning} and \eqref{Eq: Stability with Extra Conditioning 2} (which are reproduced below) hold:
\begin{align}
\label{Eq: Stability with Extra Conditioning Restated}
\P\!\left( A \middle| E \right) & \geq 1 - \tau \\
\P\!\left( B \middle| E \right) & \geq 1 - \tau 
\label{Eq: Stability with Extra Conditioning 2 Restated}
\end{align} 
where the events are:
\begin{align*}
A & = \left\{\sigma_k^+ \geq \hat{\sigma} - \epsilon\right\} , \\
B & = \left\{\sigma_k^- \leq 1-\hat{\sigma} + \epsilon\right\} , \\
E & = \left\{\sigma_K^+ \geq \hat{\sigma} - \epsilon, \, \sigma_K^- \leq 1-\hat{\sigma} + \epsilon\right\} .
\end{align*}
Now notice that for all $k > K$:
\begin{align*}
\E\!\left[\sigma_{k}^{+} - \sigma_{k}^{-}\right] & = \E\!\left[\sigma_{k}^{+} - \sigma_{k}^{-} \middle| E\right] \P(E) \\
& \quad \, + \E\!\left[\sigma_{k}^{+} - \sigma_{k}^{-} \middle| E^{c}\right] \P(E^{c}) \\
& \geq \E\!\left[\sigma_{k}^{+} - \sigma_{k}^{-} \middle| E\right] \P(E) \\
& = \P(E) \Big(\E\!\left[\sigma_{k}^{+} \middle| E, A\right] \P(A|E) \\
& \quad \quad \quad \enspace \, \, + \E\!\left[\sigma_{k}^{+} \middle| E, A^{c}\right] \P(A^{c}|E) \\
& \quad \quad \quad \enspace \, \, - \E\!\left[\sigma_{k}^{-} \middle| E\right] \!\Big) \\
& \geq \P(E) \Big(\E\!\left[\sigma_{k}^{+} \middle| E, A\right] \P(A|E) - \E\!\left[\sigma_{k}^{-} \middle| E\right] \!\Big) \\
& = \P(E) \Big(\E\!\left[\sigma_{k}^{+} \middle| E, A\right] \P(A|E) \\
& \quad \quad \quad \enspace \, \, - \E\!\left[\sigma_{k}^{-} \middle| E, B\right] \P(B|E) \\
& \quad \quad \quad \enspace \, \, - \E\!\left[\sigma_{k}^{-} \middle| E, B^{c}\right] \P(B^{c}|E) \Big) \\
& \geq \P(E) \Big(\E\!\left[\sigma_{k}^{+} \middle| E, A\right] \P(A|E) - \E\!\left[\sigma_{k}^{-} \middle| E, B\right] \\
& \quad \quad \quad \enspace \, \, - \P(B^{c}|E) \Big) \\
& \geq \P(E) \Big(\E\!\left[\sigma_{k}^{+} \middle| E, A\right] (1 - \tau) \\
& \quad \quad \quad \enspace \, \, - \E\!\left[\sigma_{k}^{-} \middle| E, B\right] - \tau \Big) \\
& \geq \P(E) \Big((\hat{\sigma} - \epsilon)(1 - \tau) - (1-\hat{\sigma} + \epsilon) - \tau \Big) \\
& = \P(E) \Big(\hat{\sigma} - (1 - \hat{\sigma}) - 2\epsilon - \tau(1 +  \hat{\sigma} - \epsilon) \Big) \\
& > 0
\end{align*}
where the second line holds because $\sigma_k^+ \geq \sigma_k^-$ almost surely (monotonicity), the fourth line holds because $\sigma_k^+ \geq 0$ almost surely, the sixth line holds because $\sigma_k^- \leq 1$ almost surely, the seventh line follows from \eqref{Eq: Stability with Extra Conditioning Restated} and \eqref{Eq: Stability with Extra Conditioning 2 Restated}, the eighth line follows from the definitions of $A$ and $B$, and the quantity in the ninth line does not depend on $k$ and is strictly positive for sufficiently small $\epsilon$ and $\tau$ (which now depends on $\delta$ and $d$) because $\hat{\sigma} > 1 - \hat{\sigma}$. Therefore, we have \eqref{Eq: STP Exchangeability}, which completes the proof of part 1.

\textbf{Part 2:} We next prove part 2. We begin with a few seemingly unrelated observations that will actually be quite useful later. Recall that $R_k = \inf_{n \geq k}{L_n}$ for every $k \in \N$ and $R_k = O\big(d^{2k}\big)$. Hence, there exists a constant $\alpha = \alpha(\delta,d) > 0$ (that depends on $\delta$ and $d$) such that for all sufficiently large $k$ (depending on $\delta$ and $d$), we have:
\begin{equation}
\label{Eq: Big O Consequence}
R_k \leq \alpha d^{2k} \, . 
\end{equation}
Let $\beta = \frac{\log(\alpha)}{6 \log(d)}$, and define the sequence $\{m(k) \in \N\}$ (indexed by $k$) as:
\begin{equation}
\label{Eq: Sequence m Definition}
m = m(k) \triangleq \floor[\Bigg]{\frac{\log\!\left(R_{\floor{(2 k/3) - \beta}}\right)}{4\log(d)}} 
\end{equation}
where $\frac{2 k}{3} \geq \beta$ for all sufficiently large $k$ (depending on $\delta$ and $d$) so that the sequence is eventually well-defined. This sequence satisfies the following conditions:
\begin{align}
\label{Eq: Limit of m}
\lim_{k \rightarrow \infty}{m(k)} & = \infty \, , \\
\lim_{k \rightarrow \infty}{\frac{d^{2m}}{R_{k-m}}} & = 0 \, .
\label{Eq: Domination by R}
\end{align}
The first limit in \eqref{Eq: Limit of m} holds because $\lim_{k \rightarrow \infty}{R_k} = \liminf_{k \rightarrow \infty}{L_k} = \infty$ (by assumption), and the second limit in \eqref{Eq: Domination by R} is true because for all sufficiently large $k$ (depending on $\delta$ and $d$):
\begin{align*}
\frac{d^{2m}}{R_{k-m}} & \leq \frac{\sqrt{R_{\floor{(2 k/3)-\beta}}}}{R_{k-m}} \\
& \leq \frac{\sqrt{R_{\floor{(2 k/3)-\beta}}}}{R_{\floor{(2k/3)-\beta}}} \\
& = \frac{1}{\sqrt{R_{\floor{(2 k/3)-\beta}}}} 
\end{align*}
where the first inequality follows from \eqref{Eq: Sequence m Definition}, and the second inequality holds because $\{R_k : k \in \N\}$ is non-decreasing, and $m \leq \log\!\big(\alpha d^{(4k/3) - 2\beta}\big)/(4\log(d)) = \frac{k}{3} + \beta$ for all sufficiently large $k$ using \eqref{Eq: Big O Consequence} and \eqref{Eq: Sequence m Definition}.

We next establish that a small portion of the random DAG $G$ above the vertex $X_{k,0}$ is a directed tree with high probability. To this end, for any sufficiently large $k \in \N$ (depending on $\delta$ and $d$) such that $k - m \geq 0$, let $G_{k}$ denote the (random) induced subgraph of the random DAG $G$ consisting of all vertices in levels $k-m,\dots,k$ that have a path to $X_{k,0}$, where $m = m(k)$ is defined in \eqref{Eq: Sequence m Definition}. (Note that $X_{k,0}$ always has a path to itself.) Moreover, define the event:
$$ T_{k} \triangleq \left\{G_{k} \text{ is a directed tree}\right\} . $$
Now, for any sufficiently large $k$ (depending on $\delta$ and $d$) such that $d^{2r} \leq R_{k - r} \leq L_{k - r}$ for every $r \in \{1,\dots,m\}$ (which is feasible due to \eqref{Eq: Domination by R}, and ensures that the ensuing steps are valid), notice that:
\begin{align}
\P(T_k) & = \prod_{r = 1}^{m}{\prod_{s = 0}^{d^{r}-1}{\!\left(1 - \frac{s}{L_{k-r}}\right)}} \nonumber \\
& \geq \prod_{r = 1}^{m}{\!\left(1 - \frac{1}{L_{k-r}}\sum_{s = 0}^{d^{r}-1}{s}\right)} \nonumber \\
& = \prod_{r = 1}^{m}{\!\left(1 - \frac{d^{r} (d^{r}-1)}{2 L_{k-r}}\right)} \nonumber \\
& \geq 1 - \frac{1}{2}\sum_{r = 1}^{m}{\frac{d^{r} (d^{r}-1)}{L_{k-r}}} \nonumber \\
& \geq 1 - \frac{1}{2 R_{k-m}}\sum_{r = 1}^{m}{d^{2r}} \nonumber \\
& = 1 - \frac{1}{2 R_{k-m}} \left(\frac{d^2 (d^{2m} - 1)}{d^2 - 1}\right) \nonumber \\
& \geq 1 - \left(\frac{d^2}{2 (d^2 - 1)}\right) \frac{d^{2m}}{R_{k-m}}
\label{Eq: Tree whp}
\end{align}
where the first equality holds because the edges of $G$ are chosen randomly and independently and we must ensure that the parents of every vertex in $G_k$ are distinct, the second and fourth inequalities are straightforward to prove by induction, and the third and sixth equalities follow from arithmetic and geometric series computations, respectively. The bound in \eqref{Eq: Tree whp} conveys that $\lim_{k \rightarrow \infty}{\P(T_k)} = 1$ due to \eqref{Eq: Domination by R}, i.e. $G_k$ is a directed tree with high probability for large $k$.

We introduce some useful notation for the remainder of this proof. First, condition on any realization of the random DAG $G$ such that the event $T_k$ occurs (for sufficiently large $k$ such that \eqref{Eq: Tree whp} holds). This also fixes the choices of Boolean processing functions at the vertices (which may vary between vertices and be graph dependent). For any vertex $X_{n,j}$ in the tree $G_k$ with $n < k$, let $\tilde{X}_{n,j}$ denote the output of the edge $\mathsf{BSC}(\delta)$ with input $X_{n,j}$ in $G_k$. (Hence, $\tilde{X}_{n,j}$ is the input of a Boolean processing function at a single vertex in level $n+1$ of $G_k$.) On the other hand, let $\tilde{X}_{k,0}$ be the output of an independent $\mathsf{BSC}(\delta)$ channel (which is not necessarily in $G$) with input $X_{k,0}$. Since $G_k$ is a tree, the random variables $\{\tilde{X}_{n,j} : X_{n,j} \text{ is a vertex of } G_k\}$ describe the values at the gates of a \textit{noisy formula} $\tilde{G}_k$, where the Boolean functions in $G_k$ correspond to $d$-input $\delta$-noisy gates in $\tilde{G}_k$ (and we think of the independent BSC errors as occurring at the gates rather than the edges). Next, in addition to conditioning on $G$ and $T_k$, we also condition on one of two realizations $X_{k-m} = x_0$ or $X_{k-m} = x_1$ for any $x_0,x_1 \in \{0,1\}^{L_{k-m}}$. In particular, corresponding to any binary random variable $Y \in \{0,1\}$ in $\tilde{G}_{k}$, define the following $2$-tuple in $[0,1]^2$, cf. \cite{HajekWeller1991,EvansSchulman2003}:
\begin{equation}
\label{Eq: Lambda Definition}
\begin{aligned}
\lambda^{Y} & \triangleq \Big(\P(Y \neq 0 | X_{k-m} = x_0, G, T_k) \, ,\\
& \quad \quad \P(Y \neq 1 | X_{k-m} = x_1, G, T_k)\Big) . 
\end{aligned}
\end{equation}
Lastly, for any constant $a \in [0,1]$, let $S(a) \subseteq [0,1]^2$ be the convex hull of the points $\{(a,a),(1-a,1-a),(0,1),(1,0)\}$, cf. \cite{HajekWeller1991,EvansSchulman2003}. With these definitions, we can state a version of the pivotal lemma in \cite[Lemma 2]{EvansSchulman2003}, which was proved in the $d = 3$ case in \cite{HajekWeller1991}.

\begin{lemma}[TV Distance Contraction in Noisy Formulae {\cite[Lemma 2]{EvansSchulman2003}}]
\label{Lemma: TV Distance Contraction in Noisy Formulae}
If $d \geq 3$ is odd and $\delta \geq \delta_{\mathsf{maj}}$, then for every possible $d$-input $\delta$-noisy gate in $\tilde{G}_k$ with inputs $Y_1,\dots,Y_d \in \{0,1\}$ and output $Y \in \{0,1\}$, we have:
$$ \lambda^{Y_1},\dots,\lambda^{Y_d} \in S(a) \text{ with } a \in \left[0,\frac{1}{2}\right] \quad \Rightarrow \quad \lambda^{Y} \in S(f(a)) $$
where the function $f:[0,1] \rightarrow [0,1]$ is defined in \eqref{Eq: General Von Neumann Recursion}.
\end{lemma}  

We remark that Lemma \ref{Lemma: TV Distance Contraction in Noisy Formulae} differs from \cite[Lemma 2]{EvansSchulman2003} in the definition of the $2$-tuple $\lambda^{Y}$ for any binary random variable $Y$ in the noisy formula. Since \cite[Lemma 2]{EvansSchulman2003} is used to yield the impossibility results on reliable computation discussed in subsection \ref{Results on Random DAG Models}, \cite[Section III]{EvansSchulman2003} defines $\lambda^{Y}$ for this purpose as $\lambda^{Y} = (\P(Y \neq X|X = 0),\P(Y \neq X|X = 1))$, where $X$ is a single relevant binary input random variable of the noisy formula (and all other inputs are fixed). In contrast, we define $\lambda^{Y}$ in \eqref{Eq: Lambda Definition} by conditioning on any two realizations of the random variables $X_{k-m}$. This ensures that the inputs, say $\tilde{X}_{n,j_1},\dots,\tilde{X}_{n,j_d}$ for some $k-m \leq n < k$ and $j_1,\dots,j_d \in [L_n]$, of every $d$-input $\delta$-noisy gate in the noisy formula $\tilde{G}_k$ are conditionally independent given $X_{k-m}$, which is a crucial property required by the proof of \cite[Lemma 2]{EvansSchulman2003}. We omit the proof of Lemma \ref{Lemma: TV Distance Contraction in Noisy Formulae} because it is virtually identical to the proof of \cite[Lemma 2]{EvansSchulman2003} in \cite[Sections IV and V]{EvansSchulman2003}. (The reader can verify that every step in the proofs in \cite[Sections IV and V]{EvansSchulman2003} continues to hold with our definition of $\lambda^Y$.)

Lemma \ref{Lemma: TV Distance Contraction in Noisy Formulae} indeed demonstrates a strong data processing inequality style of contraction for TV distance, cf. \cite[Equation (1)]{Unger2007}. To see this, observe that $(x,y) \in S(a)$ with $a \in \big[0,\frac{1}{2}\big]$ if and only if $a \leq a x + (1-a)y \leq 1-a$ and $a \leq a y + (1-a) x \leq 1-a$. This implies that $a \leq \frac{x+y}{2} \leq 1-a$, and hence, $|1-x-y| \leq 1-2a$. Furthermore, for any binary random variable $Y$ in $\tilde{G}_k$, we have using \eqref{Eq: TV Distance Definition}:
\begin{equation}
\label{Eq: TV Distance Simple Formula}
\begin{aligned}
& \left\|P_{Y|G,T_k,X_{k-m} = x_1} - P_{Y|G,T_k,X_{k-m} = x_0}\right\|_{\mathsf{TV}} \\
& \quad \quad \quad \quad \quad = \left|1-\P(Y \neq 0 | X_{k-m} = x_0, G, T_k)\right. \\
& \quad \quad \quad \quad \quad \quad \quad \, - \left.\P(Y \neq 1 | X_{k-m} = x_1, G, T_k)\right| 
\end{aligned}
\end{equation}
where $P_{Y|G,T_k,X_{k-m} = x}$ denotes the conditional distribution of $Y$ given $\{X_{k-m} = x, G, T_k\}$ for any $x \in \{0,1\}^{L_{k-m}}$. Thus, if $\lambda^{Y} \in S(a)$ with $a \in \big[0,\frac{1}{2}\big]$, then we get:
$$ \left\|P_{Y|G,T_k,X_{k-m} = x_1} - P_{Y|G,T_k,X_{k-m} = x_0}\right\|_{\mathsf{TV}} \leq 1-2a \, . $$
Now notice that $\lambda^{\tilde{X}_{k-m,j}} \in S(0)$ for every random variable $\tilde{X}_{k-m,j}$ in $\tilde{G}_{k}$, where $j \in [L_{k-m}]$. As a result, a straightforward induction argument using Lemma \ref{Lemma: TV Distance Contraction in Noisy Formulae} (much like that in the proof in \cite[Section III]{EvansSchulman2003}) yields $\lambda^{\tilde{X}_{k,0}} \in S(f^{(m)}(0))$. This implies that:\footnote{The inequality in \eqref{Eq: SDPI for TV Distance in Noisy Formulae} can be perceived as a repeated application of a \emph{tensorized} universal upper bound on the \emph{Dobrushin curve} of any $d$-input $\delta$-noisy gate, cf. \cite[Section II-A]{PolyanskiyWu2016}.}
\begin{align}
& \left\|P_{\tilde{X}_{k,0}|G,T_k,X_{k-m} = x_1} - P_{\tilde{X}_{k,0}|G,T_k,X_{k-m} = x_0}\right\|_{\mathsf{TV}} \nonumber \\
& \quad \quad \quad \quad \quad \quad \quad \quad \quad \leq 1 - 2 f^{(m)}(0) \nonumber \\
& \quad \quad \quad \quad \quad \quad \quad \quad \quad = 1 - 2\!\left(\delta * g^{(m-1)}(0)\right)
\label{Eq: SDPI for TV Distance in Noisy Formulae}
\end{align} 
where the function $g : [0,1] \rightarrow [0,1]$ is given in \eqref{Eq:Binomial form of g} in section \ref{Analysis of Majority Rule Processing in Random Grid}, and the equality follows from \eqref{Eq: Equivalence of Recursions}. Moreover, since for any $y \in \{0,1\}$ and any $x \in \{0,1\}^{L_{k-m}}$:
\begin{align*}
& \P\!\left(\tilde{X}_{k,0} \neq y \middle| G, T_k, X_{k-m} = x\right) \\
& \quad \quad \quad \quad \quad = \delta * \P(X_{k,0} \neq y | G, T_k, X_{k-m} = x) \, ,
\end{align*}
a simple calculation using \eqref{Eq: TV Distance Simple Formula} shows that:
\begin{align*}
& \left\|P_{\tilde{X}_{k,0}|G,T_k,X_{k-m} = x_1} - P_{\tilde{X}_{k,0}|G,T_k,X_{k-m} = x_0}\right\|_{\mathsf{TV}} \\
& = (1-2\delta) \left\|P_{X_{k,0}|G,T_k,X_{k-m} = x_1} - P_{X_{k,0}|G,T_k,X_{k-m} = x_0}\right\|_{\mathsf{TV}} 
\end{align*}
which, using \eqref{Eq: SDPI for TV Distance in Noisy Formulae}, produces:
\begin{equation}
\label{Eq: SDPI for TV Distance in Noisy Formulae 2}
\begin{aligned}
& \left\|P_{X_{k,0}|G,T_k,X_{k-m} = x_1} - P_{X_{k,0}|G,T_k,X_{k-m} = x_0}\right\|_{\mathsf{TV}} \\
& \quad \quad \quad \quad \quad \quad \quad \quad \quad \leq \frac{1 - 2\!\left(\delta * g^{(m-1)}(0)\right)}{1-2\delta}
\end{aligned}
\end{equation}
for any $x_0,x_1 \in \{0,1\}^{L_{k-m}}$. The inequality in \eqref{Eq: SDPI for TV Distance in Noisy Formulae 2} conveys a contraction of the TV distance on the left hand side. Since $g$ has only one fixed point at $\frac{1}{2}$ when $\delta \geq \delta_{\mathsf{maj}}$ (see section \ref{Analysis of Majority Rule Processing in Random Grid}), and \eqref{Eq: Limit of m} holds, the \textit{fixed point theorem} (see e.g. \cite[Chapter 5, Exercise 22(c)]{Rudin1976}) gives us $\lim_{k \rightarrow \infty}{g^{(m-1)}(0)} = \frac{1}{2}$, where $g^{(m-1)}(0)$ increases to $\frac{1}{2}$. Hence, the upper bound in \eqref{Eq: SDPI for TV Distance in Noisy Formulae 2} decreases to $0$ as $k \rightarrow \infty$. Furthermore, note that \eqref{Eq: SDPI for TV Distance in Noisy Formulae 2} holds for all choices of Boolean processing functions (which may vary between vertices and be graph dependent), because Lemma \ref{Lemma: TV Distance Contraction in Noisy Formulae} is agnostic to the particular gates used in $\tilde{G}_k$.

Finally, recall that the \textit{Dobrushin contraction coefficient} of any Markov transition kernel $P_{Z|W}$ with input alphabet $\mathcal{W}$ and output alphabet $\mathcal{Z}$, such that $2 \leq |\mathcal{W}|,|\mathcal{Z}| < +\infty$, is defined as \cite{Dobrushin1956}:
\begin{align}
\label{Eq: Dobrushin Coefficient Definition}
\etaTV\!\left(P_{Z|W}\right) & \triangleq \sup_{\substack{P_W,Q_W : \\ P_W \neq Q_W}}{\frac{\left\|P_Z - Q_Z\right\|_{\mathsf{TV}}}{\left\|P_W - Q_W\right\|_{\mathsf{TV}}}} \\
& = \max_{w,w^{\prime} \in \mathcal{W}}{\left\|P_{Z|W = w} - P_{Z|W = w^{\prime}}\right\|_{\mathsf{TV}}} \in [0,1]
\label{Eq: Two Point Characterization}
\end{align}
where the supremum in the first equality is over all pairs of distinct probability distributions $P_W$ and $Q_W$ on $\mathcal{W}$, $P_Z$ and $Q_Z$ denote the output distributions on $\mathcal{Z}$ induced by passing $P_W$ and $Q_W$ through $P_{Z|W}$ respectively, the second equality is Dobrushin's two-point characterization of $\eta_{\mathsf{TV}}$ \cite{Dobrushin1956}, and for any $w \in \mathcal{W}$, $P_{Z|W = w}$ denotes the $w$th conditional distribution on $\mathcal{Z}$ in $P_{Z|W}$. For any fixed realization of the random DAG $G$ such that $T_k$ occurs (for sufficiently large $k$ such that \eqref{Eq: Tree whp} holds), observe that:
\begin{align}
& \left\|P_{X_{k,0}|G,T_k,X_{0,0} = 1} - P_{X_{k,0}|G,T_k,X_{0,0} = 0}\right\|_{\mathsf{TV}} \nonumber \\
& \enspace = \etaTV\!\left(P_{X_{k,0}|G,T_k,X_{0}}\right) \nonumber \\
& \enspace \leq \etaTV\!\left(P_{X_{k,0}|G,T_k,X_{k-m}}\right) \etaTV\!\left(P_{X_{k-m}|G,T_k,X_{0}}\right) \nonumber \\
& \enspace \leq \max_{x_0,x_1}{\left\|P_{X_{k,0}|G,T_k,X_{k-m} = x_1} - P_{X_{k,0}|G,T_k,X_{k-m} = x_0}\right\|_{\mathsf{TV}}} \nonumber \\
& \enspace \leq \frac{1 - 2\!\left(\delta * g^{(m-1)}(0)\right)}{1-2\delta}
\label{Eq: Contraction Based TV Bound}
\end{align}
where $P_{X_{k,0}|G,T_k,X_{0}}$, $P_{X_{k,0}|G,T_k,X_{k-m}}$, and $P_{X_{k-m}|G,T_k,X_{0}}$ are transition kernels from $X_0$ to $X_{k,0}$, from $X_{k-m}$ to $X_{k,0}$, and from $X_0$ to $X_{k-m}$ respectively, the first equality follows from \eqref{Eq: Two Point Characterization} where $P_{X_{k,0}|G,T_k,X_{0,0} = y}$ denotes the conditional distribution of $X_{k,0}$ given $\{X_{0,0} = y,G,T_k\}$ for any $y \in \{0,1\}$, the second inequality holds because $X_{0} \rightarrow X_{k-m} \rightarrow X_{k,0}$ forms a Markov chain (given $G$ and $T_k$) and $\etaTV$ is sub-multiplicative in its input stochastic matrix (this follows easily from \eqref{Eq: Dobrushin Coefficient Definition}\textemdash see e.g. \cite[Lemma 4.3]{Seneta1981}), the third inequality follows from \eqref{Eq: Two Point Characterization} and the maximum here is over all $x_0,x_1 \in \{0,1\}^{L_{k-m}}$, and the last inequality follows from \eqref{Eq: SDPI for TV Distance in Noisy Formulae 2}. Taking conditional expectations with respect to $G$ given $T_k$ in \eqref{Eq: Contraction Based TV Bound} yields:
$$ \E\!\left[\left\|P_{X_{k,0}|G}^{+} - P_{X_{k,0}|G}^{-}\right\|_{\mathsf{TV}} \middle| T_k\right] \leq \frac{1 - 2\!\left(\delta * g^{(m-1)}(0)\right)}{1-2\delta} $$
where $P_{X_{k,0}|G}^{+}$ and $P_{X_{k,0}|G}^{-}$ inside the conditional expectation correspond to the conditional probability distributions $P_{X_{k,0}|G,T_k,X_{0,0} = 1}$ and $P_{X_{k,0}|G,T_k,X_{0,0} = 0}$, respectively (as we condition on $T_k$). Therefore, we have:
\begin{align*}
& \E\!\left[\left\|P_{X_{k,0}|G}^{+} - P_{X_{k,0}|G}^{-}\right\|_{\mathsf{TV}} \right] \\
& \quad = \E\!\left[\left\|P_{X_{k,0}|G}^{+} - P_{X_{k,0}|G}^{-}\right\|_{\mathsf{TV}} \middle| T_k\right] \P(T_k) \\
& \quad \quad \, + \E\!\left[\left\|P_{X_{k,0}|G}^{+} - P_{X_{k,0}|G}^{-}\right\|_{\mathsf{TV}} \middle| T_k^c\right] (1 - \P(T_k)) \\
& \quad \leq \frac{1 - 2\!\left(\delta * g^{(m-1)}(0)\right)}{1-2\delta} + \left(\frac{d^2}{2 (d^2 - 1)}\right) \frac{d^{2m}}{R_{k-m}}
\end{align*}
using the tower property, the fact that TV distance is bounded by $1$, and \eqref{Eq: Tree whp}. Letting $k \rightarrow \infty$ establishes the desired result:
\begin{align*}
& \lim_{k \rightarrow \infty}{\E\!\left[\left\|P_{X_{k,0}|G}^{+} - P_{X_{k,0}|G}^{-}\right\|_{\mathsf{TV}} \right]} \\
& \leq \frac{\displaystyle{1 - 2\!\left(\delta * \lim_{k \rightarrow \infty}{g^{(m-1)}(0)}\right)}}{1-2\delta} + \left(\frac{d^2}{2 (d^2 - 1)}\right) \lim_{k \rightarrow \infty}{\frac{d^{2m}}{R_{k-m}}} \\
& = 0 
\end{align*}
because $\lim_{k \rightarrow \infty}{g^{(m-1)}(0)} = \frac{1}{2}$ (as noted earlier) and \eqref{Eq: Domination by R} holds. This completes the proof.
\end{proof}

\section{Proof of Corollary \ref{Cor: Existence of Grids where Reconstruction is Possible}} 
\label{Proof of Corollary Existence of Grids where Reconstruction is Possible}

\begin{proof}
This follows from applying the probabilistic method. Fix any $d \geq 3$, any $\delta \in (0,\delta_{\mathsf{maj}})$, and any sequence of level sizes satisfying $L_k \geq C(\delta,d)\log(k)$ for all sufficiently large $k$. We know from Theorem \ref{Thm:Phase Transition in Random Grid with Majority Rule Processing} that for the random DAG model with these parameters and majority processing functions, there exist $\epsilon = \epsilon(\delta,d) > 0$ and $K = K(\delta,d) \in \N$ (which depend on $\delta$ and $d$) such that:
$$ \forall k \geq K, \enspace \P\!\left(\hat{S}_{k} \neq X_{0}\right) \leq \frac{1}{2} - 2\epsilon \, . $$
Now, for $k \in \N$, define:
$$ P_k(G) \triangleq \P\!\left(h_{\mathsf{ML}}^k(X_k,G) \neq X_{0} \middle| G\right) $$ 
as the conditional probability that the ML decision rule based on the full $k$-layer state $X_k$ makes an error given the random DAG $G$, and let $E_k$ for $k \in \N$ be the set of all deterministic DAGs $\mathcal{G}$ with indegree $d$ and level sizes $\{L_m: m \in \N\}$ such that $P_k(\mathcal{G}) \leq \frac{1}{2} - \epsilon$. Observe that for every $k \geq K$:
\begin{align*}
\frac{1}{2} - 2\epsilon & \geq \P\!\left(\hat{S}_{k} \neq X_{0}\right) \\
& = \E\!\left[\P\!\left(\hat{S}_{k} \neq X_{0}\middle|G\right)\right] \\
& \geq \E\!\left[P_k(G)\right] \\
& = \E\!\left[P_k(G)\middle|G \in E_k\right] \P\!\left(G \in E_k\right) \\
& \quad \, + \E\!\left[P_k(G)\middle|G \not\in E_k\right] \P\!\left(G \not\in E_k\right) \\
& \geq \E\!\left[P_k(G)\middle|G \not\in E_k\right] \P\!\left(G \not\in E_k\right) \\
& \geq \left(\frac{1}{2} - \epsilon\right) \P\!\left(G \not\in E_k\right)
\end{align*}
where the second and fourth lines follow from the law of total expectation, the third line holds because the ML decision rule minimizes the probability of error, the fifth line holds because the first term in the previous line is non-negative, and the final line holds because $G \not\in E_k$ implies that $P_k(G) > \frac{1}{2} - \epsilon$. Then, we have for every $k \geq K$:
$$ \P\!\left(G \in E_k\right) \geq \frac{2\epsilon}{1 - 2\epsilon} > 0 \, . $$
Since $\{E_k : k \in \N\}$ form a non-increasing sequence of sets (because $P_k(G)$ is non-decreasing in $k$), we get via continuity:
$$ \P\!\left(G \in \bigcap_{k \in \N}{E_k}\right) = \lim_{k \rightarrow \infty}{\P\!\left(G \in E_k\right)} \geq \frac{2\epsilon}{1 - 2\epsilon} > 0 $$
which means that there exists a deterministic DAG $\mathcal{G}$ with indegree $d$, noise level $\delta$, level sizes $\{L_k : k \in \N\}$, and majority processing functions such that $P_k(\mathcal{G}) \leq \frac{1}{2} - \epsilon$ for all $k \in \N$. This completes the proof.
\end{proof}

\section{Proof of Proposition \ref{Prop: Slow Growth of Layers}} 
\label{Proof of Proposition Slow Growth of Layers}

\begin{proof} ~\newline
\indent
\textbf{Part 1:} We first prove part 1, where we are given a fixed deterministic DAG $\mathcal{G}$. Observe that the BSC along each edge of this DAG produces its output bit by either copying its input bit exactly with probability $1-2\delta$, or generating an independent $\Ber\big(\frac{1}{2}\big)$ output bit with probability $2\delta$. This is because the BSC transition matrix can be decomposed as:\footnote{This simple, but useful, idea is a specialization of more sophisticated Fortuin-Kasteleyn random cluster representations of Ising models \cite{Grimmett1997}, and has been exploited in the contexts of broadcasting on trees \cite[p.412]{Evansetal2000}, and reliable computation \cite[p.570]{Feder1989}.}
\begin{equation}
\def\arraystretch{1.1} 
\left[\begin{array}{cc} 1-\delta & \delta \\ \delta & 1-\delta \end{array}\right] = (1-2\delta)\left[\begin{array}{cc} 1 & 0 \\ 0 & 1 \end{array}\right] + (2\delta) \left[\begin{array}{cc} \frac{1}{2} & \frac{1}{2} \\ \frac{1}{2} & \frac{1}{2} \end{array}\right] .
\def\arraystretch{1} 
\end{equation}
Now consider the events:
$$ A_k \triangleq \left\{\parbox[]{17em}{all $d L_k$ edges from level $k-1$ to level $k$ generate independent output bits}\right\} $$
for $k \in \N\backslash\!\{0\}$, which have probabilities $\P(A_k) = (2\delta)^{d L_k}$ since the BSCs on the edges are independent. These events are mutually independent (once again because the BSCs on the edges are independent). Since the condition on $L_k$ in the proposition statement is equivalent to:
$$ (2\delta)^{d L_k} \geq \frac{1}{k} \enspace \text{for all sufficiently large } k \, , $$
we must have:
$$ \sum_{k = 1}^{\infty}{\P(A_k)} = \sum_{k = 1}^{\infty}{(2\delta)^{d L_k}} = +\infty \, . $$ 
The second Borel-Cantelli lemma then tells us that infinitely many of the events $\{A_k : k \in \N\backslash\!\{0\}\}$ occur almost surely, i.e. $\P\big(\bigcap_{m = 1}^{\infty} \bigcup_{k = m}^{\infty} A_{k}\big) = 1$. In particular, if we define $B_m \triangleq \bigcup_{k = 1}^{m} A_{k}$ for $m \in \N\backslash\!\{0\}$, then by continuity:
\begin{equation}
\label{Eq: Forget whp}
\lim_{m \rightarrow \infty}{\P\!\left(B_m\right)} = \P\!\left(\bigcup_{k = 1}^{\infty} A_{k}\right) = 1 \, .
\end{equation}
Finally, observe that:
\begin{align}
& \lim_{m \rightarrow \infty} \P\!\left(h_{\mathsf{ML}}^m(X_m,\mathcal{G}) \neq X_0\right) \nonumber \\
& \quad \quad = \lim_{m \rightarrow \infty} \P\!\left( h_{\mathsf{ML}}^m(X_m,\mathcal{G}) \neq X_0 \middle| B_m\right) \P(B_m) \nonumber \\
& \quad \quad \quad \quad \quad \enspace + \P\!\left(h_{\mathsf{ML}}^m(X_m,\mathcal{G}) \neq X_0 \middle| B_m^{c}\right) \P(B_m^{c}) \nonumber \\
& \quad \quad = \lim_{m \rightarrow \infty} \P\!\left( h_{\mathsf{ML}}^m(X_m,\mathcal{G}) \neq X_0 \middle| B_m\right) \nonumber \\
& \quad \quad = \lim_{m \rightarrow \infty} \frac{1}{2} \, \P\!\left( h_{\mathsf{ML}}^m(X_m,\mathcal{G}) = 1 \middle| B_m\right) \nonumber \\
& \quad \quad \quad \quad \quad \enspace + \frac{1}{2} \, \P\!\left( h_{\mathsf{ML}}^m(X_m,\mathcal{G}) = 0 \middle| B_m\right) \nonumber \\
& \quad \quad = \frac{1}{2}
\label{Eq: ML Decoder Fails}
\end{align}
where $h_{\mathsf{ML}}^m(\cdot,\mathcal{G}):\{0,1\}^{L_m} \rightarrow \{0,1\}$ denotes the ML decision rule at level $m$ based on $X_m$ (given knowledge of the DAG $\mathcal{G}$), the second equality uses \eqref{Eq: Forget whp}, and the third equality holds because $X_{0,0} \sim \Ber\big(\frac{1}{2}\big)$ is independent of $B_m$, and $X_m$ is conditionally independent of $X_0$ given $B_m$. The condition in \eqref{Eq: ML Decoder Fails} is equivalent to the TV distance condition in part 1 of the proposition statement; this proves part 1.

\textbf{Part 2:} To prove part 2, notice that part 1 immediately yields:
$$ \lim_{k \rightarrow \infty}{\left\|P_{X_{k}|G}^+ - P_{X_{k}|G}^-\right\|_{\mathsf{TV}}} = 0 \quad \textit{pointwise} $$
which completes the proof. 
\end{proof}

\section{Proof of Proposition \ref{Prop:Broadcasting in Unbounded Degree DAG Model}}
\label{Proof of Proposition Broadcasting in Unbounded Degree DAG Model}

\begin{proof} ~\newline
\indent
\textbf{Part 1:} Fix any noise level $\delta \in \big(0,\frac{1}{2}\big)$ and any constant $\epsilon \in \big(0,\frac{1}{4}\big)$. Furthermore, tentatively suppose that $L_k \geq A(\epsilon,\delta) \sqrt{\log(k)}$ for all sufficiently large $k$, where the constant $A(\epsilon,\delta)$ is defined as:
\begin{equation}
\label{Eq: Generalized Level Size Constant Definition}
A(\epsilon,\delta) \triangleq \frac{2}{(1-2\delta) \epsilon \sqrt{1-2\epsilon}} \, .
\end{equation}
Now consider the deterministic DAG $\mathcal{G}$ such that each vertex at level $k \in \N\backslash\!\{0\}$ is connected to all $L_{k-1}$ vertices at level $k-1$ and all Boolean processing functions are the majority rule. (Note that when there is only one input, the majority rule behaves like the identity map.) For all $k \in \N\backslash\!\{0\}$, since $X_k$ is an exchangeable sequence of random variables given $\sigma_0$, $\sigma_k$ is a sufficient statistic of $X_k$ for performing inference about $\sigma_0$, where $\sigma_k$ is defined in \eqref{Eq: Empirical Probability of Unity Definition} (cf. subsection \ref{Random Grid Model}). We next prove a useful ``one-step broadcasting'' lemma involving $\sigma_k$'s for this model.

\begin{lemma}[One-Step Broadcasting in Unbounded Degree DAG]
\label{Lemma: One-Step Broadcasting in Unbounded Degree DAG}
Under the aforementioned assumptions, there exists $K = K(\epsilon,\delta) \in \N$ (that depends on $\epsilon$ and $\delta$) such that for all $k \geq K$, we have:
$$ \P\!\left(\sigma_{k} \geq \frac{1}{2} + \epsilon \, \middle|\, \sigma_{k-1} \geq \frac{1}{2} + \epsilon\right) \geq 1 - \left(\frac{1}{k-1}\right)^{\! 2} . $$
\end{lemma}

\begin{proof}
Suppose we are given that $\sigma_{k-1} = \sigma \geq \frac{1}{2} + \epsilon$ for any $k \in \N\backslash\!\{0\}$. Then, $\{X_{k,j} : j \in [L_k]\}$ are conditionally i.i.d. $\Ber(\P(X_{k,0} = 1|\sigma_{k-1} = \sigma))$ and $L_k \sigma_k \sim \mathsf{binomial}(L_k,\P(X_{k,0} = 1|\sigma_{k-1} = \sigma))$, where $\P(X_{k,0} = 1|\sigma_{k-1} = \sigma) = \E[\sigma_k | \sigma_{k-1} = \sigma]$. Furthermore, since $X_{k,0}$ is the majority of the values of $X_{k-1,0},\dots,X_{k-1,L_{k-1}-1}$ after passing them through independent $\mathsf{BSC}(\delta)$'s, we have:
\begin{align}
\E[\sigma_k | \sigma_{k-1} = \sigma] & = \P(X_{k,0} = 1|\sigma_{k-1} = \sigma) \nonumber \\
& = 1 - \P\!\left(\sum_{i=1}^{L_{k-1} \sigma}{Z_i} + \!\sum_{j=1}^{L_{k-1}(1-\sigma)}{Y_j} < \frac{L_{k-1}}{2}\!\right) \nonumber \\
& \geq 1 - \exp\!\left(\!-2 L_{k-1} \left(\frac{1}{2} - \sigma * \delta\right)^{\! 2}\right) \nonumber \\
& \geq 1 - \exp\!\left(\!-2 L_{k-1} \left(\frac{1}{2} - \left(\frac{1}{2} + \epsilon\right) * \delta\right)^{\! 2}\right) \nonumber \\
& = 1 - \exp\!\left(-2 L_{k-1} \epsilon^2 (1 - 2\delta)^2\right) 
\label{Eq: Lower Bound on Conditional Mean}
\end{align} 
where $Z_i$ are i.i.d. $\Ber(1 - \delta)$, $Y_j$ are i.i.d. $\Ber(\delta)$, $\{Z_i : i \in \{1,\dots,L_{k-1}\sigma\} \}$ and $\{Y_j : j \in \{1,\dots,L_{k-1}(1-\sigma)\}\}$ are independent, the first inequality follows from Hoeffding's inequality using the fact that $\sigma * \delta > \frac{1}{2}$ (because $\sigma > \frac{1}{2}$), and the second inequality holds because $\sigma \geq \frac{1}{2} + \epsilon$, which implies that $\sigma * \delta \geq \big(\frac{1}{2} + \epsilon\big) * \delta > \frac{1}{2}$. 

Next, observe that there exists $K = K(\epsilon,\delta) \in \N$ (that depends on $\epsilon$ and $\delta$) such that for all $k \geq K$, we have:
$$ \exp\!\left(-2 L_{k-1} \epsilon^2 (1 - 2\delta)^2\right) \leq \epsilon $$
because $\lim_{k \rightarrow \infty}{L_{k}} = \infty$ by assumption. So, for any $k \geq K$, this yields the bound:
$$ \E[\sigma_k | \sigma_{k-1} = \sigma] \geq 1 - \epsilon > \frac{1}{2} + \epsilon $$
using \eqref{Eq: Lower Bound on Conditional Mean} and the fact that $\epsilon < \frac{1}{4}$. As a result, we can apply the Chernoff-Hoeffding bound, cf. \cite[Theorem 1]{Hoeffding1963}, to $\sigma_k$ for any $k \geq K$ and get:
\begin{align*}
& \P\!\left(\sigma_k < \frac{1}{2} + \epsilon \, \middle|\, \sigma_{k-1} = \sigma\right) \\
& \quad \quad \quad \leq \exp\!\left(- L_k D\!\left(\frac{1}{2} + \epsilon \,\middle|\middle|\,\E[\sigma_k | \sigma_{k-1} = \sigma]\right)\right) 
\end{align*}
where $D(\alpha||\beta) \triangleq \alpha \log(\alpha/\beta) + (1-\alpha)\log((1-\alpha)/(1-\beta))$ for $\alpha,\beta \in (0,1)$ denotes the binary \textit{Kullback-Leibler divergence} (or relative entropy) function. Notice that:
\begin{align*}
& D\!\left(\frac{1}{2} + \epsilon \,\middle|\middle|\,\E[\sigma_k | \sigma_{k-1} = \sigma]\right) \\
& \quad \geq -H\!\left(\frac{1}{2}+\epsilon\right) - \left(\frac{1}{2}-\epsilon\right)\log(1-\E[\sigma_k | \sigma_{k-1} = \sigma]) \\
& \quad \geq L_{k-1}\epsilon^2 (1-2\epsilon) (1-2\delta)^2 - H\!\left(\frac{1}{2}+\epsilon\right)
\end{align*}
where $H(\cdot)$ denotes the binary Shannon entropy function (see e.g. Proposition \ref{Prop: Random Expander Graph}), the first inequality holds because $\log(\E[\sigma_k | \sigma_{k-1} = \sigma]) < 0$, and the second inequality follows from \eqref{Eq: Lower Bound on Conditional Mean}. Hence, we have for any $k \geq K$:
\begin{align*}
& \P\!\left(\sigma_k < \frac{1}{2} + \epsilon \, \middle|\, \sigma_{k-1} = \sigma\right) \\
& \leq \exp\!\left(- L_{k-1} L_{k} \epsilon^2 (1-2\epsilon) (1-2\delta)^2 + L_{k} H\!\left(\frac{1}{2}+\epsilon\right)\right) 
\end{align*}
where we can multiply both sides by $\P(\sigma_{k-1} = \sigma)$ and then sum over all $\sigma \geq \frac{1}{2} + \epsilon$ (as in the proof of \eqref{Eq: Stability whp} within the proof of Theorem \ref{Thm:Phase Transition in Random Grid with Majority Rule Processing} in section \ref{Analysis of Majority Rule Processing in Random Grid}) to get:
\begin{align*}
& \P\!\left(\sigma_k < \frac{1}{2} + \epsilon \, \middle|\, \sigma_{k-1} \geq \frac{1}{2} + \epsilon\right) \\
& \leq \exp\!\left(- L_{k-1} L_{k} \left(\epsilon^2 (1-2\epsilon) (1-2\delta)^2 - \frac{H\!\left(\frac{1}{2}+\epsilon\right)}{L_{k-1}}\right)\right) . 
\end{align*}
Since $\lim_{k \rightarrow \infty}{L_k} = \infty$ by assumption, we can choose $K = K(\epsilon,\delta)$ to be sufficiently large so that for all $k \geq K$, we also have $H\big(\frac{1}{2}+\epsilon\big)/L_{k-1} \leq \epsilon^2 (1-2\epsilon) (1-2\delta)^2 \!/2$. Thus, for every $k \geq K$:
\begin{equation}
\label{Eq: One-Step Broadcasting Inequality for Unbounded Degree DAG}
\begin{aligned}
& \P\!\left(\sigma_k < \frac{1}{2} + \epsilon \, \middle|\, \sigma_{k-1} \geq \frac{1}{2} + \epsilon\right) \\
& \quad \quad \leq \exp\!\left(- L_{k-1} L_{k} \frac{\epsilon^2 (1-2\epsilon) (1-2\delta)^2}{2}\right) . 
\end{aligned}
\end{equation}

Finally, we once again increase $K = K(\epsilon,\delta)$ if necessary to ensure that $L_{k-1} \geq A(\epsilon,\delta) \sqrt{\log(k-1)}$ for every $k \geq K$ (as presumed earlier). This implies that for all $k \geq K$:
\begin{align*}
L_{k-1} L_{k} & \geq A(\epsilon,\delta)^2 \sqrt{\log(k-1) \log(k)} \\
& \geq A(\epsilon,\delta)^2 \log(k-1) 
\end{align*}
which, using \eqref{Eq: Generalized Level Size Constant Definition} and \eqref{Eq: One-Step Broadcasting Inequality for Unbounded Degree DAG}, produces:
\begin{align*}
\P\!\left(\sigma_k < \frac{1}{2} + \epsilon \, \middle|\, \sigma_{k-1} \geq \frac{1}{2} + \epsilon\right) & \leq \exp(- 2 \log(k-1)) \\
& = \left(\frac{1}{k-1}\right)^{\! 2} 
\end{align*}
for all $k \geq K$. This proves Lemma \ref{Lemma: One-Step Broadcasting in Unbounded Degree DAG}.
\end{proof}

Lemma \ref{Lemma: One-Step Broadcasting in Unbounded Degree DAG} is an analogue of \eqref{Eq: Stability whp} in the proof of Theorem \ref{Thm:Phase Transition in Random Grid with Majority Rule Processing} in section \ref{Analysis of Majority Rule Processing in Random Grid}. It illustrates that if
the proportion of 1's is large in a given layer of $\mathcal{G}$, then it remains large in the next layer of $\mathcal{G}$ with high probability. 

To proceed, we specialize Lemma \ref{Lemma: One-Step Broadcasting in Unbounded Degree DAG} by arbitrarily selecting a particular value of $\epsilon$, say $\epsilon = \frac{7}{32} \in \big(0,\frac{1}{4}\big)$. This implies that the constant $A(\epsilon,\delta)$ becomes:
\begin{equation}
\label{Eq: Specialized Level Size Constant}
A(\delta) = A\!\left(\frac{7}{32},\delta\right) = \frac{256}{21 (1-2\delta)}
\end{equation}
using \eqref{Eq: Generalized Level Size Constant Definition}. In the proposition statement, it is assumed that $L_k \geq A(\delta) \sqrt{\log(k)}$ for all sufficiently large $k$. Thus, Lemma \ref{Lemma: One-Step Broadcasting in Unbounded Degree DAG} holds with $\epsilon = \frac{7}{32} \in \big(0,\frac{1}{4}\big)$ under the assumptions of part 1 of Proposition \ref{Prop:Broadcasting in Unbounded Degree DAG Model}. At this point, we can execute the proof of part 1 of Theorem \ref{Thm:Phase Transition in Random Grid with Majority Rule Processing} in section \ref{Analysis of Majority Rule Processing in Random Grid} mutatis mutandis (with Lemma \ref{Lemma: One-Step Broadcasting in Unbounded Degree DAG} playing the pivotal role of \eqref{Eq: Stability whp}) to establish part 1 of Proposition \ref{Prop:Broadcasting in Unbounded Degree DAG Model}. We omit the details of this proof for brevity.

\textbf{Part 2:} To prove part 2, we use the proof technique of part 1 of Proposition \ref{Prop: Slow Growth of Layers} in Appendix \ref{Proof of Proposition Slow Growth of Layers}. Recall that the $\mathsf{BSC}(\delta)$ along each edge of the DAG $\mathcal{G}$ produces its output bit by either copying its input bit with probability $1-2\delta$, or generating an independent $\Ber\big(\frac{1}{2}\big)$ output bit with probability $2\delta$. As before, consider the mutually independent events:
$$ A_k \triangleq \left\{\parbox[]{19em}{all $L_{k-1} L_k$ edges from level $k-1$ to level $k$ generate independent output bits}\right\} $$
for $k \in \N\backslash\!\{0\}$, which have probabilities $\P(A_k) = (2\delta)^{L_{k-1} L_k}$. Define the constant $B(\delta)$ as:
\begin{equation}
\label{Eq: Converse Level Size Constant Definition}
B(\delta) \triangleq \frac{1}{\sqrt{\log\!\left(\frac{1}{2\delta}\right)}} \, .
\end{equation}
Since we assume in the proposition statement that $L_k \leq B(\delta) \sqrt{\log(k)}$ for all sufficiently large $k$, we have:
$$ L_{k-1} L_k \leq \frac{\sqrt{\log(k-1)\log(k)}}{\log\!\left(\frac{1}{2\delta}\right)} \leq \frac{\log(k)}{\log\!\left(\frac{1}{2\delta}\right)} $$
for all sufficiently large $k$, which implies that:
$$ (2\delta)^{L_{k-1} L_k} \geq \frac{1}{k} \,  $$
for all sufficiently large $k$. Hence, we get:
$$ \sum_{k = 1}^{\infty}{\P(A_k)} = \sum_{k = 1}^{\infty}{(2\delta)^{L_{k-1} L_k}} = +\infty \, , $$
and the second Borel-Cantelli lemma establishes that infinitely many of the events $\{A_k : k \in \N\backslash\!\{0\}\}$ occur almost surely, or equivalently, $\P\big(\bigcap_{m = 1}^{\infty} \bigcup_{k = m}^{\infty} A_{k}\big) = 1$. As a result, we can define the events $B_m \triangleq \bigcup_{k = 1}^{m} A_{k}$ for $m \in \N\backslash\!\{0\}$ such that $\lim_{m \rightarrow \infty}{\P(B_m)} = 1$. Therefore, we have (as before):
$$ \lim_{m \rightarrow \infty} \P\!\left(h_{\mathsf{ML}}^m(X_m,\mathcal{G}) \neq X_0\right) = \frac{1}{2} $$
where $h_{\mathsf{ML}}^m(\cdot,\mathcal{G}):\{0,1\}^{L_m} \rightarrow \{0,1\}$ denotes the ML decision rule at level $m$ based on $X_m$ (given knowledge of the DAG $\mathcal{G}$). This completes the proof. 
\end{proof}

\section{Proof of Proposition \ref{Prop: Majority Grid Almost Sure Convergence}}
\label{Proof of Proposition Majority Grid Almost Sure Convergence}

\begin{proof}
Recall that $L_k \sigma_k \sim \mathsf{binomial}(L_k,g(\sigma))$ given $\sigma_{k-1} = \sigma$. This implies via Hoeffding's inequality and \eqref{Eq:Conditional Expectation} that for every $k \in \N\backslash\!\{0\}$ and $\epsilon_k > 0$:
$$ \P(|\sigma_k - g(\sigma_{k-1})| > \epsilon_k|\sigma_{k-1} = \sigma) \leq 2 \exp\!\left(-2 L_k \epsilon_k^2 \right) $$
where we can take expectations with respect to $\sigma_{k-1}$ to get:
\begin{equation}
\label{Eq:Hoeffding Consequence}
\P(|\sigma_k - g(\sigma_{k-1})| > \epsilon_k) \leq 2 \exp\!\left(-2 L_k \epsilon_k^2 \right) . 
\end{equation}
Now fix any $\tau > 0$, and choose a sufficiently large integer $K = K(\tau) \in \N$ (that depends on $\tau$) such that:
\begin{align*}
& \P(\exists k > K, \, |\sigma_k - g(\sigma_{k-1})| > \epsilon_k) \\
& \quad \quad \quad \quad \quad \quad \leq \sum_{k = K+1}^{\infty}{\P(|\sigma_k - g(\sigma_{k-1})| > \epsilon_k)} \\
& \quad \quad \quad \quad \quad \quad \leq 2 \sum_{k = K+1}^{\infty}{\exp\!\left(-2 L_k \epsilon_k^2 \right)} \\
& \quad \quad \quad \quad \quad \quad \leq \tau 
\end{align*}
where we use the union bound and the inequality \eqref{Eq:Hoeffding Consequence}, and let $\epsilon_k = \sqrt{\log(k)/L_k}$ (or equivalently, $\exp(-2 L_k \epsilon_k^2) = 1/k^2$). This implies that for any $\tau > 0$:
\begin{equation}
\label{Eq:Precursor to Borel-Cantelli result}
\P(\forall k > K, \, |\sigma_k - g(\sigma_{k-1})| \leq \epsilon_k) \geq 1 - \tau \, . 
\end{equation}
Since for every $k > K$, $|\sigma_k - g(\sigma_{k-1})| \leq \epsilon_k$, we can recursively obtain the following relation for every $k > K$:
\begin{equation}
\left|\sigma_k - g^{(k-K)}(\sigma_K)\right| \leq \sum_{m = K+1}^{k}{D(\delta,d)^{k-m}\epsilon_m} 
\end{equation} 
where $D(\delta,d)$ denotes the Lipschitz constant of $g$ on $[0,1]$ as defined in \eqref{Eq:Lipschitz constant of g}, and $D(\delta,d) \in (0,1)$ since $\delta \in \big(\delta_{\mathsf{maj}},\frac{1}{2}\big)$. Since $L_m = \omega(\log(m))$, for any $\epsilon > 0$, we can take $K = K(\epsilon,\tau) \in \N$ (which depends on both $\epsilon$ and $\tau$) to be sufficiently large so that $\sup_{m > K}{\epsilon_m} \leq \epsilon (1-D(\delta,d))$. Now observe that we have for all $k \in \N\backslash[K+1]$:
\begin{align*}
\sum_{m = K+1}^{k}{D(\delta,d)^{k-m}\epsilon_m} & \leq \left(\sup_{m > K}{\epsilon_m}\right) \sum_{j = 0}^{\infty}{D(\delta,d)^{j}} \\
& = \left(\sup_{m > K}{\epsilon_m}\right) \frac{1}{1-D(\delta,d)} \\
& \leq \epsilon \, . 
\end{align*}
Moreover, since $g:[0,1] \rightarrow [0,1]$ is a contraction when $\delta \in \big(\delta_{\mathsf{maj}},\frac{1}{2}\big)$, it has a unique fixed point $\sigma = \frac{1}{2}$, and $\lim_{m \rightarrow \infty}{g^{(m)}(\sigma_K)} = \frac{1}{2}$ almost surely by the fixed point theorem. As a result, for any $\tau > 0$ and any $\epsilon > 0$, there exists $K = K(\epsilon,\tau) \in \N$ such that:
$$ \P\!\left(\forall k > K, \, \left|\sigma_k - g^{(k-K)}(\sigma_K)\right| \leq \epsilon\right) \geq 1 - \tau $$ 
which implies, after letting $k \rightarrow \infty$, that:
$$ \P\!\left(\frac{1}{2} - \epsilon \leq \liminf_{k \rightarrow \infty}{\sigma_k} \leq \limsup_{k \rightarrow \infty}{\sigma_k} \leq \frac{1}{2} + \epsilon\right) \geq 1 - \tau \, . $$
Lastly, we can first let $\epsilon \rightarrow 0$ and employ the continuity of $\P$, and then let $\tau \rightarrow 0$ to obtain:
$$ \P\!\left(\lim_{k \rightarrow \infty}{\sigma_k} = \frac{1}{2}\right) = 1 \, . $$
This completes the proof.
\end{proof}

\section{Proof of Proposition \ref{Prop: And-Or Grid Almost Sure Convergence}}
\label{Proof of Proposition And-Or Grid Almost Sure Convergence}

\begin{proof} 
This proof is analogous to the proof of Proposition \ref{Prop: Majority Grid Almost Sure Convergence} in Appendix \ref{Proof of Proposition Majority Grid Almost Sure Convergence}. For every $k \in \N\backslash\!\{0\}$ and $\epsilon_k > 0$, we have after taking expectations in \eqref{Eq: Hoeffding Consequence 2} that:
\begin{equation}
\label{Eq: Hoeffding Consequence 3}
\P(|\sigma_{2k} - g(\sigma_{2k-2})| > \epsilon_k) \leq 4 \exp\!\left(-\frac{(L_{2k} \wedge L_{2k-1}) \epsilon_k^2}{8}\right) .
\end{equation}
Now fix any $\tau > 0$, and choose a sufficiently large integer $K = K(\tau) \in \N$ (that depends on $\tau$) such that:
\begin{align*}
& \P(\exists k > K, \, |\sigma_{2k} - g(\sigma_{2k-2})| > \epsilon_k) \\
& \quad \quad \quad \quad \leq \sum_{m = K+1}^{\infty}{\P(|\sigma_{2m} - g(\sigma_{2m-2})| > \epsilon_{m})} \\
& \quad \quad \quad \quad \leq 4 \sum_{m = K+1}^{\infty}{\exp\!\left(-\frac{(L_{2m} \wedge L_{2m-1}) \epsilon_m^2}{8}\right)} \\ 
& \quad \quad \quad \quad \leq \tau 
\end{align*}
where we use the union bound and the inequality \eqref{Eq: Hoeffding Consequence 3}, and we set $\epsilon_m = 4\sqrt{\log(m)/(L_{2m} \wedge L_{2m-1})}$ (which ensures that $\exp(-(L_{2m} \wedge L_{2m-1}) \epsilon_m^2/8) = 1/m^2$). This implies that for any $\tau > 0$:
\begin{equation}
\P(\forall k > K, \, |\sigma_{2k} - g(\sigma_{2k-2})| \leq \epsilon_k) \geq 1 - \tau \, . 
\end{equation}
Since for every $k > K$, $|\sigma_{2k} - g(\sigma_{2k-2})| \leq \epsilon_k$, we can recursively obtain the following relation for every $k > K$:
\begin{equation}
\left|\sigma_{2k} - g^{(k-K)}(\sigma_{2K})\right| \leq \sum_{m = K+1}^{k}{D(\delta)^{k-m}\epsilon_{m}} 
\end{equation} 
where $D(\delta)$ denotes the Lipschitz constant of $g$ on $[0,1]$ as shown in \eqref{Eq: Lipschitz constant}, which is in $(0,1)$ since $\delta \in \big(\delta_{\mathsf{andor}},\frac{1}{2}\big)$. Since $L_m = \omega(\log(m))$, for any $\epsilon > 0$, we can take $K = K(\epsilon,\tau) \in \N$ (which depends on both $\epsilon$ and $\tau$) to be sufficiently large so that $\sup_{m > K}{\epsilon_{m}} \leq \epsilon (1-D(\delta))$. This implies that:
$$ \forall k \in \N \backslash [K+1], \enspace \sum_{m = K+1}^{k}{D(\delta)^{k-m}\epsilon_{m}} \leq \epsilon $$
as shown in the proof of Proposition \ref{Prop: Majority Grid Almost Sure Convergence}. Moreover, since $g:[0,1] \rightarrow [0,1]$ is a contraction when $\delta \in \big(\delta_{\mathsf{andor}},\frac{1}{2}\big)$, it has a unique fixed point $\sigma = t \in [0,1]$, and $\lim_{m \rightarrow \infty}{g^{(m)}(\sigma_{2K})} = t$ almost surely by the fixed point theorem. As a result, for any $\tau > 0$ and any $\epsilon > 0$, there exists $K = K(\epsilon,\tau) \in \N$ such that:
$$ \P\!\left(\forall k > K, \, \left|\sigma_{2k} - g^{(k-K)}(\sigma_{2K})\right| \leq \epsilon\right) \geq 1 - \tau $$ 
which implies that:
$$ \P\!\left(\lim_{k \rightarrow \infty}{\sigma_{2k}} = t\right) = 1 $$
once again as shown in the proof of Proposition \ref{Prop: Majority Grid Almost Sure Convergence}. This completes the proof.
\end{proof}

\section{Proof of Corollary \ref{Cor: Lossless Expander Graph}}
\label{Proof of Corollary Lossless Expander Graph}

\begin{proof}
Fix any $\epsilon > 0$ and any $d \geq \big(\frac{2}{\epsilon}\big)^{\!5}$, and let $\alpha = d^{-6/5}$. To establish the first part of the corollary, it suffices to prove that for every $n \in \N\backslash\!\{0\}$:
\begin{align*}
n \!\left(1-(1-\alpha)^d - \sqrt{2 d \alpha H(\alpha)}\right) & \geq (1-\epsilon) d \alpha n \\
\Leftrightarrow \quad 1-\frac{1-(1-\alpha)^d - \sqrt{2 d \alpha H(\alpha)}}{d \alpha} & \leq \epsilon \, . 
\end{align*}   
Indeed, if this is true, then Proposition \ref{Prop: Random Expander Graph} immediately implies the desired lower bound on the probability that $\mathsf{B}$ is a $d$-regular bipartite lossless $(d^{-6/5},(1-\epsilon)d)$-expander graph. To this end, observe that:
\begin{align*}
1-\frac{1-(1-\alpha)^d - \sqrt{2 d \alpha H(\alpha)}}{d \alpha} & \leq 1-\frac{1-e^{-d \alpha}}{d \alpha} \\
& \quad \, + \frac{\sqrt{2 d \alpha H(\alpha)}}{d \alpha} \\
& \leq \frac{d \alpha}{2} + \sqrt{\frac{2 H(\alpha)}{d \alpha}} \\
& \leq \frac{d \alpha}{2} + \sqrt{\frac{4 \log(2)}{d \sqrt{\alpha}}} \\
& = \frac{1}{2 d^{1/5}} + \frac{2 \sqrt{\log(2)}}{d^{1/5}} \\
& \leq \frac{2}{d^{1/5}} \\
& \leq \epsilon
\end{align*}
where the first inequality follows from the standard bound $(1-\alpha)^d \leq e^{-d \alpha}$ for $\alpha \in (0,1)$ and $d \in \N\backslash\!\{0\}$, the second inequality follows from the easily verifiable bounds $0 \leq 1 - \frac{1- e^{-x}}{x} \leq \frac{x}{2}$ for $x > 0$, the third inequality follows from the well-known bound $H(\alpha) \leq 2 \log(2) \sqrt{\alpha (1-\alpha)} \leq 2 \log(2) \sqrt{\alpha}$ for $\alpha \in (0,1)$, the fourth equality follows from substituting $\alpha = d^{-6/5}$, the fifth inequality follows from direct computation, and the final inequality holds because $d \geq \big(\frac{2}{\epsilon}\big)^{\! 5}$. This proves the first part of the corollary.

The existence of $d$-regular bipartite lossless $(d^{-6/5},(1-\epsilon)d)$-expander graphs for every sufficiently large $n$ (depending on $d$) in the second part of the corollary follows from the first part by invoking the probabilistic method. This completes the proof.
\end{proof}

\section*{Acknowledgment}

Anuran Makur would like to thank Dheeraj Nagaraj and Ganesh Ajjanagadde for stimulating discussions.

\bibliographystyle{IEEEtran}
\bibliography{IEEEabrv,BroadcastRefs5}

\begin{thebibliography}{10}
\providecommand{\url}[1]{#1}
\csname url@samestyle\endcsname
\providecommand{\newblock}{\relax}
\providecommand{\bibinfo}[2]{#2}
\providecommand{\BIBentrySTDinterwordspacing}{\spaceskip=0pt\relax}
\providecommand{\BIBentryALTinterwordstretchfactor}{4}
\providecommand{\BIBentryALTinterwordspacing}{\spaceskip=\fontdimen2\font plus
\BIBentryALTinterwordstretchfactor\fontdimen3\font minus
  \fontdimen4\font\relax}
\providecommand{\BIBforeignlanguage}[2]{{%
\expandafter\ifx\csname l@#1\endcsname\relax
\typeout{** WARNING: IEEEtran.bst: No hyphenation pattern has been}%
\typeout{** loaded for the language `#1'. Using the pattern for}%
\typeout{** the default language instead.}%
\else
\language=\csname l@#1\endcsname
\fi
#2}}
\providecommand{\BIBdecl}{\relax}
\BIBdecl

\bibitem{MakurMosselPolyanskiy2019Conf}
A.~Makur, E.~Mossel, and Y.~Polyanskiy, ``Broadcasting on random networks,'' in
  \emph{Proceedings of the IEEE International Symposium on Information Theory
  (ISIT)}, Paris, France, July 7-12 2019, pp. 1632--1636.

\bibitem{Evansetal2000}
W.~Evans, C.~Kenyon, Y.~Peres, and L.~J. Schulman, ``Broadcasting on trees and
  the {I}sing model,'' \emph{The Annals of Applied Probability}, vol.~10,
  no.~2, pp. 410--433, May 2000.

\bibitem{LyonsPeres2017}
R.~Lyons and Y.~Peres, \emph{Probability on Trees and Networks}, ser. Cambridge
  Series in Statistical and Probabilistic Mathematics.\hskip 1em plus 0.5em
  minus 0.4em\relax New York, NY, USA: Cambridge University Press, 2017,
  vol.~42.

\bibitem{KestenStigum1966}
H.~Kesten and B.~P. Stigum, ``A limit theorem for multidimensional
  {G}alton-{W}atson processes,'' \emph{The Annals of Mathematical Statistics},
  vol.~37, no.~5, pp. 1211--1223, October 1966.

\bibitem{BleherRuizZagrebnov1995}
P.~M. Bleher, J.~Ruiz, and V.~A. Zagrebnov, ``On the purity of the limiting
  {G}ibbs state for the {I}sing model on the {B}ethe lattice,'' \emph{Journal
  of Statistical Physics}, vol.~79, no. 1-2, pp. 473--482, April 1995.

\bibitem{Ioffe1996a}
D.~Ioffe, ``Extremality of the disordered state for the {I}sing model on
  general trees,'' in \emph{Trees: Workshop in Versailles, June 14-16 1995},
  ser. Progress in Probability, B.~Chauvin, S.~Cohen, and A.~Rouault, Eds.,
  vol.~40.\hskip 1em plus 0.5em minus 0.4em\relax Basel, Switzerland:
  Birkh\"{a}user, 1996, pp. 3--14.

\bibitem{Ioffe1996b}
------, ``On the extremality of the disordered state for the {I}sing model on
  the {B}ethe lattice,'' \emph{Letters in Mathematical Physics}, vol.~37,
  no.~2, pp. 137--143, June 1996.

\bibitem{Mossel1998}
E.~Mossel, ``Recursive reconstruction on periodic trees,'' \emph{Random
  Structures and Algorithms}, vol.~13, no.~1, pp. 81--97, August 1998.

\bibitem{Mossel2001}
------, ``Reconstruction on trees: {B}eating the second eigenvalue,'' \emph{The
  Annals of Applied Probability}, vol.~11, no.~1, pp. 285--300, February 2001.

\bibitem{PemantlePeres2010}
R.~Pemantle and Y.~Peres, ``The critical {I}sing model on trees, concave
  recursions and nonlinear capacity,'' \emph{The Annals of Probability},
  vol.~38, no.~1, pp. 184--206, January 2010.

\bibitem{Sly2009}
A.~Sly, ``Reconstruction of random colourings,'' \emph{Communications in
  Mathematical Physics}, vol. 288, no.~3, pp. 943--961, June 2009.

\bibitem{Sly2011}
------, ``Reconstruction for the {P}otts model,'' \emph{The Annals of
  Probability}, vol.~39, no.~4, pp. 1365--1406, July 2011.

\bibitem{JansonMossel2004}
S.~Janson and E.~Mossel, ``Robust reconstruction on trees is determined by the
  second eigenvalue,'' \emph{The Annals of Probability}, vol.~32, no.~3B, pp.
  2630--2649, July 2004.

\bibitem{Bhatnagaretal2011}
N.~Bhatnagar, J.~Vera, E.~Vigoda, and D.~Weitz, ``Reconstruction for colorings
  on trees,'' \emph{SIAM Journal on Discrete Mathematics}, vol.~25, no.~2, pp.
  809--826, July 2011.

\bibitem{Mossel2004survey}
E.~Mossel, ``Survey: {I}nformation flow on trees,'' in \emph{Graphs, Morphisms
  and Statistical Physics, DIMACS Workshop, March 19-21 2001}, ser. DIMACS
  Series in Discrete Mathematics and Theoretical Computer Science,
  J.~Ne\v{s}et\v{r}il and P.~Winkler, Eds., vol.~63.\hskip 1em plus 0.5em minus
  0.4em\relax Providence, RI, USA: American Mathematical Society, 2004, pp.
  155--170.

\bibitem{vonNeumann1956}
J.~{von Neumann}, ``Probabilistic logics and the synthesis of reliable
  organisms from unreliable components,'' in \emph{Automata Studies}, ser.
  Annals of Mathematics Studies, C.~E. Shannon and J.~McCarthy, Eds.,
  vol.~34.\hskip 1em plus 0.5em minus 0.4em\relax Princeton, NJ, USA: Princeton
  University Press, 1956, pp. 43--98.

\bibitem{EvansSchulman1999}
W.~S. Evans and L.~J. Schulman, ``Signal propagation and noisy circuits,''
  \emph{IEEE Transactions on Information Theory}, vol.~45, no.~7, pp.
  2367--2373, November 1999.

\bibitem{HajekWeller1991}
B.~Hajek and T.~Weller, ``On the maximum tolerable noise for reliable
  computation by formulas,'' \emph{IEEE Transactions on Infomation Theory},
  vol.~37, no.~2, pp. 388--391, March 1991.

\bibitem{EvansSchulman2003}
W.~S. Evans and L.~J. Schulman, ``On the maximum tolerable noise of $k$-input
  gates for reliable computation by formulas,'' \emph{IEEE Transactions on
  Information Theory}, vol.~49, no.~11, pp. 3094--3098, November 2003.

\bibitem{FriedliVelenik2018}
S.~Friedli and Y.~Velenik, \emph{Statistical Mechanics of Lattice Systems: {A}
  Concrete Mathematical Introduction}.\hskip 1em plus 0.5em minus 0.4em\relax
  New York, NY, USA: Cambridge University Press, 2018.

\bibitem{Georgii2011}
H.-O. Georgii, \emph{{G}ibbs Measures and Phase Transitions}, 2nd~ed., ser. De
  Gruyter Studies in Mathematics.\hskip 1em plus 0.5em minus 0.4em\relax
  Berlin, Germany: De Gruyter, 2011, vol.~9.

\bibitem{Gray2001}
L.~F. Gray, ``A reader's guide to {G}acs's ``positive rates'' paper,''
  \emph{Journal of Statistical Physics}, vol. 103, no. 1-2, pp. 1--44, April
  2001.

\bibitem{Mossel2003}
E.~Mossel, ``On the impossibility of reconstructing ancestral data and
  phylogenies,'' \emph{Journal of Computational Biology}, vol.~10, no.~5, pp.
  669--676, July 2003.

\bibitem{Mossel2004}
------, ``Phase transitions in phylogeny,'' \emph{Transactions of the American
  Mathematical Society}, vol. 356, no.~6, pp. 2379--2404, June 2004.

\bibitem{DaskalakisMosselRoch2006}
C.~Daskalakis, E.~Mossel, and S.~Roch, ``Optimal phylogenetic reconstruction,''
  in \emph{Proceedings of the 38th Annual ACM Symposium on Theory of Computing
  (STOC)}, Seattle, WA, USA, May 21-23 2006, pp. 159--168.

\bibitem{Roch2010}
S.~Roch, ``Toward extracting all phylogenetic information from matrices of
  evolutionary distances,'' \emph{Science}, vol. 327, no. 5971, pp. 1376--1379,
  March 2010.

\bibitem{HusonRuppScornavacca2010}
D.~H. Huson, R.~Rupp, and C.~Scornavacca, \emph{Phylogenetic Networks:
  Concepts, Algorithms and Applications}.\hskip 1em plus 0.5em minus
  0.4em\relax New York, NY, USA: Cambridge University Press, 2010.

\bibitem{Thompson1986}
E.~A. Thompson, \emph{Pedigree Analysis in Human Genetics}, ser. Johns Hopkins
  Series in Contemporary Medicine and Public Health.\hskip 1em plus 0.5em minus
  0.4em\relax Baltimore, MD, USA: Johns Hopkins University Press, 1986.

\bibitem{SteelHein2006}
M.~Steel and J.~Hein, ``Reconstructing pedigrees: {A} combinatorial
  perspective,'' \emph{Journal of Theoretical Biology, Elsevier}, vol. 240,
  no.~3, pp. 360--367, June 2006.

\bibitem{MezardMontanari2006}
M.~M\'{e}zard and A.~Montanari, ``Reconstruction on trees and spin glass
  transition,'' \emph{Journal of Statistical Physics}, vol. 124, no.~6, pp.
  1317--1350, September 2006.

\bibitem{Krzakalaetal2007}
F.~Kr\c{z}aka{\l}a, A.~Montanari, F.~Ricci-Tersenghi, G.~Semerjian, and
  L.~Zdeborov{\'a}, ``{G}ibbs states and the set of solutions of random
  constraint satisfaction problems,'' \emph{Proceedings of the National Academy
  of Sciences (PNAS)}, vol. 104, no.~25, pp. 10\,318--10\,323, June 2007.

\bibitem{GerschenfeldMontanari2007}
A.~Gerschenfeld and A.~Montanari, ``Reconstruction for models on random
  graphs,'' in \emph{Proceedings of the 48th Annual IEEE Symposium on
  Foundations of Computer Science (FOCS)}, Providence, RI, USA, October 20-23
  2007, pp. 194--204.

\bibitem{MontanariRestrepoTetali2011}
A.~Montanari, R.~Restrepo, and P.~Tetali, ``Reconstruction and clustering in
  random constraint satisfaction problems,'' \emph{SIAM Journal on Discrete
  Mathematics}, vol.~25, no.~2, pp. 771--808, July 2011.

\bibitem{Abbe2018}
E.~Abbe, \emph{Community Detection and Stochastic Block Models}, ser.
  Foundations and Trends in Communications and Information Theory.\hskip 1em
  plus 0.5em minus 0.4em\relax Hanover, MA, USA: now Publishers Inc., 2018,
  vol.~14, no. 1-2.

\bibitem{Keener2010}
R.~W. Keener, \emph{Theoretical Statistics: {T}opics for a Core Course}, ser.
  Springer Texts in Statistics.\hskip 1em plus 0.5em minus 0.4em\relax New
  York, NY, USA: Springer, 2010.

\bibitem{DobrushinOrtyukov1977}
R.~L. Dobrushin and S.~I. Ortyukov, ``Lower bound for the redundancy of
  self-correcting arrangements of unreliable functional elements,''
  \emph{Problemy Peredachi Informatsii}, vol.~13, no.~1, pp. 82--89, 1977, in
  Russian.

\bibitem{EvansPippenger1998}
W.~Evans and N.~Pippenger, ``On the maximum tolerable noise for reliable
  computation by formulas,'' \emph{IEEE Transactions on Information Theory},
  vol.~44, no.~3, pp. 1299--1305, May 1998.

\bibitem{Unger2007}
F.~Unger, ``Noise threshold for universality of $2$-input gates,'' in
  \emph{Proceedings of the IEEE International Symposium on Information Theory
  (ISIT)}, Nice, France, June 24-29 2007, pp. 1901--1905.

\bibitem{Unger2008}
------, ``Noise threshold for universality of two-input gates,'' \emph{IEEE
  Transactions on Information Theory}, vol.~54, no.~8, pp. 3693--3698, August
  2008.

\bibitem{KolmogorovBarzdin1967}
A.~N. Kolmogorov and Y.~M. Barzdin, ``On the realization of networks in
  three-dimensional space,'' in \emph{Selected Works of A. N. Kolmogorov --
  Volume III: Information Theory and the Theory of Algorithms}, ser.
  Mathematics and Its Applications (Soviet Series), A.~N. Shiryayev, Ed.,
  vol.~27.\hskip 1em plus 0.5em minus 0.4em\relax Dordrecht, Netherlands:
  Springer, 1993, pp. 194--202, originally published in: Problemy Kibernetiki,
  no. 19, pp. 261-268, 1967.

\bibitem{Pinsker1973}
M.~S. Pinsker, ``On the complexity of a concentrator,'' in \emph{Proceedings of
  the 7th International Teletraffic Congress (ITC)}, Stockholm, Sweden, June
  13-20 1973, pp. 318/1--318/4.

\bibitem{Margulis1973}
G.~A. Margulis, ``Explicit constructions of concentrators,'' \emph{Problemy
  Peredachi Informatsii}, vol.~9, no.~4, pp. 71--80, 1973, in Russian.

\bibitem{Capalboetal2002}
M.~R. Capalbo, O.~Reingold, S.~P. Vadhan, and A.~Wigderson, ``Randomness
  conductors and constant-degree lossless expanders,'' in \emph{Proceedings of
  the 34th Annual ACM Symposium on Theory of Computing (STOC)}, Montr\'{e}al,
  QC, Canada, May 19-21 2002, pp. 659--668.

\bibitem{PolyanskiyWu2017}
Y.~Polyanskiy and Y.~Wu, ``Strong data-processing inequalities for channels and
  {B}ayesian networks,'' in \emph{Convexity and Concentration}, ser. The IMA
  Volumes in Mathematics and its Applications, E.~Carlen, M.~Madiman, and E.~M.
  Werner, Eds., vol. 161.\hskip 1em plus 0.5em minus 0.4em\relax New York, NY,
  USA: Springer, 2017, pp. 211--249.

\bibitem{Feder1989}
T.~Feder, ``Reliable computation by networks in the presence of noise,''
  \emph{IEEE Transactions on Information Theory}, vol.~35, no.~3, pp. 569--571,
  May 1989.

\bibitem{PolyanskiyWu2016}
Y.~Polyanskiy and Y.~Wu, ``Dissipation of information in channels with input
  constraints,'' \emph{IEEE Transactions on Information Theory}, vol.~62,
  no.~1, pp. 35--55, January 2016.

\bibitem{Margulis1974}
G.~A. Margulis, ``Probabilistic characteristics of graphs with large
  connectivity,'' \emph{Problemy Peredachi Informatsii}, vol.~10, no.~2, pp.
  101--108, 1974, in Russian.

\bibitem{Russo1981}
L.~Russo, ``On the critical percolation probabilities,'' \emph{Zeitschrift
  f\"{u}r Wahrscheinlichkeitstheorie und Verwandte Gebiete}, vol.~56, no.~2,
  pp. 229--237, June 1981.

\bibitem{Grimmett1997}
G.~Grimmett, ``Percolation and disordered systems,'' in \emph{Lectures on
  Probability Theory and Statistics: Ecole d'Et{\'e} de Probabilit{\'e}s de
  Saint-Flour XXVI-1996}, ser. Lecture Notes in Mathematics, P.~Bernard, Ed.,
  vol. 1665.\hskip 1em plus 0.5em minus 0.4em\relax Berlin, Heidelberg,
  Germany: Springer, 1997, pp. 153--300.

\bibitem{ShuttyWoottersHayden2018}
N.~Shutty, M.~Wootters, and P.~Hayden, ``Noise thresholds for amplification:
  Quantum foundations from classical fault-tolerant computation,'' October
  2018, arXiv:1809.09748v2 [cs.IT].

\bibitem{LevinPeresWilmer2009}
D.~A. Levin, Y.~Peres, and E.~L. Wilmer, \emph{Markov Chains and Mixing Times},
  1st~ed.\hskip 1em plus 0.5em minus 0.4em\relax Providence, RI, USA: American
  Mathematical Society, 2009.

\bibitem{Hoeffding1963}
W.~Hoeffding, ``Probability inequalities for sums of bounded random
  variables,'' \emph{Journal of the American Statistical Association}, vol.~58,
  no. 301, pp. 13--30, March 1963.

\bibitem{SipserSpielman1996}
M.~Sipser and D.~A. Spielman, ``Expander codes,'' \emph{IEEE Transactions on
  Information Theory}, vol.~42, no.~6, pp. 1710--1722, November 1996.

\bibitem{Bollobas2001}
B.~Bollob\'{a}s, \emph{Random Graphs}, 2nd~ed., ser. Cambridge Studies in
  Advanced Mathematics.\hskip 1em plus 0.5em minus 0.4em\relax Cambridge,
  United Kingdom: Cambridge University Press, 2001, vol.~73.

\bibitem{Feller1968}
W.~Feller, \emph{An Introduction to Probability Theory and Its Applications},
  3rd~ed.\hskip 1em plus 0.5em minus 0.4em\relax New York, NY, USA: John Wiley
  \& Sons, Inc., 1968, vol.~1.

\bibitem{CsiszarShields2004}
I.~Csisz\'{a}r and P.~C. Shields, \emph{Information Theory and Statistics: {A}
  Tutorial}, ser. Foundations and Trends in Communications and Information
  Theory, S.~Verd\'{u}, Ed.\hskip 1em plus 0.5em minus 0.4em\relax Hanover, MA,
  USA: now Publishers Inc., 2004, vol.~1, no.~4.

\bibitem{Knuth1997}
D.~E. Knuth, \emph{The Art of Computer Programming: Seminumerical Algorithms},
  3rd~ed.\hskip 1em plus 0.5em minus 0.4em\relax Reading, MA, USA:
  Addison-Wesley, 1997, vol.~2.

\bibitem{Rudin1976}
W.~Rudin, \emph{Principles of Mathematical Analysis}, 3rd~ed., ser.
  International Series in Pure and Applied Mathematics.\hskip 1em plus 0.5em
  minus 0.4em\relax New York, NY, USA: McGraw-Hill, 1976.

\bibitem{Dobrushin1956}
R.~L. Dobrushin, ``Central limit theorem for nonstationary {M}arkov chains.
  {I},'' \emph{Theory of Probability and Its Applications}, vol.~1, no.~1, pp.
  65--80, 1956.

\bibitem{Seneta1981}
E.~Seneta, \emph{Non-negative Matrices and {M}arkov Chains}, 2nd~ed., ser.
  Springer Series in Statistics.\hskip 1em plus 0.5em minus 0.4em\relax New
  York, NY, USA: Springer, 1981.

\end{thebibliography}

\balance

\end{document}